\documentclass{article}
\usepackage{kr}

\usepackage{times}
\usepackage{soul}
\usepackage{url}
\usepackage[hidelinks]{hyperref}
\usepackage[utf8]{inputenc}
\usepackage[small]{caption}
\usepackage{graphicx}
\usepackage{amssymb,amsmath}
\usepackage{amsthm}
\usepackage{booktabs}
\usepackage{algorithm}
\usepackage{algorithmic}
\usepackage{apxproof}

\usepackage{pifont} 

\usepackage{tikz}
\usetikzlibrary{automata}
\usetikzlibrary{positioning,arrows,fit,calc}
\usetikzlibrary{decorations.pathreplacing}
\usetikzlibrary{matrix}

\renewcommand{\qed}{\hfill\ding{113}}

\newcommand{\nb}[1]{}
\newcommand{\nz}[1]{}

\pgfdeclarelayer{background}
\pgfdeclarelayer{foreground}
\pgfsetlayers{background,main,foreground}

\newcommand{\avec}[1]{\boldsymbol{#1}}

\newcommand{\LTL}{\textsl{LTL}}

\newcommand{\U}{\mathbin{\mathsf{U}}}

\newcommand{\Qp}{\mathcal{Q}_p}
\newcommand{\Qpi}{\mathcal{P}^\sigma}

\newcommand{\ind}{\mathit{ind}}

\newcommand{\dep}{\textit{tdp}}
\newcommand{\rdep}{\textit{rdp}}
\newcommand{\suc}{\textit{suc}}
\newcommand{\tp}{\textit{tp}}

\newcommand{\D}{\mathcal{D}}

\newcommand{\I}{\mathcal{I}}

\newcommand{\Diamondw}{\Diamond_{\!r}}
\newcommand{\nxt}{{\ensuremath\raisebox{0.25ex}{\text{\scriptsize$\bigcirc$}}}}

\renewcommand{\L}{{\boldsymbol{L}}}
\newcommand{\op}{{\boldsymbol{o}}}

\newcommand{\var}{\textit{var}}
\newcommand{\len}{\max}

\newcommand{\A}{\mathcal A}
\newcommand{\q}{{\avec{q}}}
\newcommand{\el}{\avec{l}}
\newcommand{\s}{\avec{s}}
\renewcommand{\r}{\avec{r}}

\newcommand{\sig}{\mathit{sig}}

\newcommand{\qw}{\bar{\q}}

\newtheorem{theorem}{Theorem}

\newtheorem{example}[theorem]{Example}

\newtheorem{lemma}[theorem]{Lemma}

\newtheorem{corollary}[theorem]{Corollary}

\title{
	Unique Characterisability and Learnability of Temporal Instance Queries}

\author{%
	M.~Fortin$^1$\and
	B.~Konev$^1$\and
	V.~Ryzhikov$^2$\and
	Y.~Savateev$^{2}$\and
	F.~Wolter$^1$\and
	M.~Zakharyaschev$^2$ \\
	\affiliations
	$^1$Department of Computer Science, University of Liverpool, UK\\
	$^2$Department of Computer Science, Birkbeck, University of London, UK\\
	\emails
	\{marie.fortin,boris.konev,wolter@liverpool.ac.uk\}, \{vlad,yury,michael\}@dcs.bbk.ac.uk}

\begin{document}
	
	\maketitle
	
\begin{abstract}
We aim to determine which temporal instance queries can be uniquely characterised by a (polynomial-size) set of positive and negative temporal data examples. We start by considering queries formulated in fragments of propositional linear temporal logic \LTL{} that correspond to conjunctive queries (CQs) or extensions thereof induced by the until operator. Not all of these queries admit polynomial characterisations but by restricting them further to path-shaped queries we identify natural classes that do.
We then investigate how far the obtained characterisations can be lifted to temporal knowledge graphs queried by 2D languages combining \LTL{} with concepts in description logics $\mathcal{EL}$ or $\mathcal{ELI}$ (i.e., tree-shaped CQs). While temporal operators in the scope of description logic constructors can destroy polynomial characterisability, we obtain general transfer results for the case when description logic constructors are within the scope of temporal operators. Finally, we   apply our characterisations to establish (polynomial) learnability of temporal instance queries using membership queries in the active learning framework.
\end{abstract}


\section{Introduction}\label{intro}

Constructing queries or, more generally, logical concepts describing individuals of interest, can be difficult. Providing support to a user to cope with this problem has been a major research topic in databases, logic, and knowledge representation. For instance, in reverse engineering of database queries and concept descriptions~\cite{martins2019reverse,DBLP:journals/ml/LehmannH10,DBLP:conf/kr/JungLPW20}, one aims to identify a query using a set of positively and negatively labelled examples of answers and non-answers, respectively; and in active learning approaches, one aims to identify a query by asking an oracle (e.g., domain specialist) whether an example is an answer or a non-answer to the query~\cite{DBLP:journals/ml/AngluinFP92,DBLP:conf/ijcai/FunkJL21,DBLP:conf/icdt/CateD21}.
	
Recently, the \emph{unique characterisation} of a query by a finite (ideally, polynomial-size) set of positive and negative example answers has been identified as a fundamental link between queries and data~\cite{DBLP:conf/icdt/CateD21}. Namely, we say that a query $\q$ \emph{fits} a pair $E=(E^{+},E^{-})$ of sets $E^{+}$ and $E^{-}$ of pointed databases $(\mathcal{D},a)$ if $\mathcal{D}\models \q(a)$ for $(\mathcal{D},a) \in E^{+}$, and $\mathcal{D}\not\models \q(a)$ for $(\mathcal{D},a) \in E^{-}$. Then $E$ \emph{uniquely characterises} $\q$ within a class $\mathcal{Q}$ of queries if $\q$ is the only (up to equivalence) query in $\mathcal{Q}$ that fits $E$.

Unique (polynomial) characterisations can be used to illustrate, explain, and construct queries. They are also a `non-proce\-dural' necessary condition for (polynomial) learnability using membership queries in Angluin's~(\citeyear{DBLP:journals/ml/Angluin87}) framework of active learning,  where membership queries to the oracle take the form `does $\mathcal{D}\models \q(a)$ hold?'\!. It is shown by ten Cate and Dalmau~(\citeyear{DBLP:conf/icdt/CateD21}) that, for classes of conjunctive queries (CQs), it is often a sufficient condition as well.

In many applications, queries are required to capture the temporal evolution of individuals, making their formulation even harder. The aim of this paper is to start an investigation of the (polynomial) characterisability of temporal instance queries.
We first consider one-dimensional data instances of the form $(\delta_{0},\dots,\delta_{n})$, where $\delta_{i}$ is the set of atomic propositions that are true at timestamp $i$, describing the temporal behaviour of a single individual, and queries formulated in fragments of propositional linear temporal logic \LTL{}.
Although rather basic as a temporal data model, this restriction allows us to focus on the purely temporal aspect of unique characterisability. We then generalise our results, where possible, to standard two-dimensional temporal data instances, in which the $\delta_{i}$ are replaced by non-temporal data instances and queries are obtained by combining fragments of \LTL{} with $\mathcal{ELI}$-concept queries (or tree-shaped CQs), thereby combining a well established formalism for accessing temporal data~\cite{Chomicki2018}
with the basic concept descriptions for tractable data access from description logic~\cite{DL-Textbook}.


Our initial observation is that already very primitive temporal queries are not uniquely characterisable. For example\footnote{For detailed explanations and omitted proofs, the reader is referred to the Appendix.}\!, consider the query $\q = \Diamondw (A\wedge B)$ with the operator $\Diamondw$ `now or later' (interpreted by $\leq$ over linearly ordered timestamps). By the pigeonhole principle, no finite example set $E$ can distinguish $\q$ from a query $\q' = \Diamondw (A \land (\Diamondw B \land \Diamondw (A \land \dots)))$ with sufficiently many alternating $A$ and $B$. Similarly, the query $\q = \nxt A$ with the `next time' operator $\nxt$ is not distinguishable by a finite example set from $\q' = (\nxt \dots\nxt A) \U A$ with the strict `until' operator $\U$ and sufficiently many $\nxt$ on its left-hand side.
	

Aiming to identify natural and useful classes of temporal queries enjoying (polynomial) characterisability, in this paper we consider the conjunctive fragment of \LTL.
To begin with, we focus on two classes of path CQs: the class $\mathcal{Q}_{p}[\nxt,\Diamondw]$ of queries of the form
	\begin{equation}\label{dnpath}
		\q = \rho_0 \land \op_1 (\rho_1 \land \op_2 (\rho_2 \land \dots \land \op_n \rho_n) ),
\end{equation}
where $\op_i \in \{\nxt, \Diamondw\}$ and $\rho_i$ is a conjunction of atomic propositions, and also  the class $\Qp^\sigma[\U]$ of $\U$-queries of the form
	\begin{equation}\label{upath}
		\q = \rho_0 \land (\lambda_1 \U (\rho_1 \land ( \lambda_2 \U ( \dots (\lambda_n \U \rho_n) \dots )))),
	\end{equation}
	where $\lambda_i$ is a conjunction of atoms or $\bot$. The superscript $\sigma$ in $\Qp^\sigma[\U]$ indicates that queries are formulated in a finite signature $\sigma$, a condition required
because of the universal quantification implicit in $\U$.
Our first main result is a syntactic criterion of (polynomial) characterisability of $\mathcal{Q}_{p}[\nxt,\Diamondw]$-queries. In fact, it turns out that the query $\Diamond_{r}(A \wedge B)$ mentioned above epitomises the cause of non-characterisability in $\mathcal{Q}_{p}[\nxt,\Diamondw]$.
It follows, in particular, that the restriction $\mathcal{Q}_{p}[\nxt,\Diamond]$ of
$\mathcal{Q}_{p}[\nxt,\Diamondw]$ to queries with $\nxt$ and strict eventuality $\Diamond=\nxt\Diamondw$ is polynomially characterisable.
Our second main result is that all $\Qp^\sigma[\U]$-queries
with $\subseteq$-incomparable $\lambda_{i}$ and $\rho_{i}$, for each $i$, are polynomially characterisable within $\Qp^\sigma[\U]$.
Although we show that all $\Qp^\sigma[\U]$-queries are exponentially characterisable, it remains open whether they are polynomially characterisable in $\Qp^\sigma[\U]$.

The essential property that distinguishes $\mathcal{Q}_{p}[\nxt,\Diamondw]$ and $\Qp^\sigma[\U]$ from other queries is that they do not admit temporal branching as, for instance, in $\Diamond A \wedge \Diamond B$. In fact, we show that even within the class of queries using only $\land$ and $\Diamond$ and with a bound on the number of branches, not all queries are polynomially characterisable. A first step towards positive results covering non-path queries is made for the case of queries in which all branches are of equal length.

Our next aim is to generalise the obtained results to 2D temporal queries combining $\LTL{}$ with the description logic constructor $\exists P$ of $\mathcal{ELI}$. Our first main result is negative: even queries of the form
$
\exists P.\q_{1} \wedge \cdots \wedge \exists P.\q_{n}
$, in which $\q_{i}\in \mathcal{Q}_{p}[\nxt,\Diamond]$, are not polynomially characterisable.
The situation changes drastically, however, if we consider queries of the form \eqref{dnpath} or \eqref{upath}, in which $\rho_{i}$ and $\lambda_{i}$ are $\mathcal{ELI}$-queries. Indeed, we generalise our polynomial characterisability results for $\mathcal{Q}_{p}[\nxt,\Diamondw]$ and $\mathcal{Q}_{p}^{\sigma}[\U]$ to such queries using recent results on the computation of frontiers in the lattice of $\mathcal{ELI}$-queries~\cite{DBLP:conf/icdt/CateD21} and proving a new result on split partners in the lattice of $\mathcal{EL}$-queries (where $\mathcal{EL}$ denotes $\mathcal{ELI}$ without inverse roles).
	
Finally, we discuss applications of our results to learning temporal instance queries using membership queries of the form `does $\mathcal{D}\models\q$ hold?'\!.
As we always construct example sets effectively, our unique (exponential) characterisability results imply (exponential-time) learnability with membership queries. Obtaining polynomial-time learnability from
polynomial characterisations is more challenging. A main result here is that $\mathcal{Q}_{p}[\nxt,\Diamondw]$ with $\mathcal{ELI}$-queries is polynomial-time learnable with membership queries, assuming the learner is given the target query size in advance.



\section{Related Work}
Our contribution is closely related to work on active learnability of formal languages and on  learning temporal logic formulas interpreted over finite and infinite traces.
It is also related to learning database queries and other formal expressions for accessing data. In the former area, the seminal paper by Angluin~(\citeyear{DBLP:journals/iandc/Angluin87}) has given rise to a large body of work on active learning of regular languages or variations, for example,~\cite{shahbaz2009inferring,aarts2010learning,DBLP:journals/fac/CasselHJS16,DBLP:conf/dagstuhl/HowarS16}. This work has mainly focused on learning various types of finite state machines or automata using a combination of membership queries with other powerful types of queries such as equivalence queries. The use of two or more types of queries is motivated by the fact that otherwise one cannot efficiently learn a wide variety of important languages, including regular languages. In fact, the main difference between this work and our contribution is that we focus on queries for which the corresponding formal languages form only a small subset of the regular languages and it is this restriction that enables us to focus on characterisability and learnability with membership queries.

Rather surprisingly, there has hardly been any work on active learning of temporal formulas over finite or infinite traces; we refer the reader to~\cite{DBLP:conf/aips/CamachoM19}, also for a discussion of the relationship between learning automata and \LTL{}-formulas. In contrast, passive learning of \LTL{}-formulas has recently received significant attention; see~\cite{lemieux2015general,DBLP:conf/fmcad/NeiderG18,DBLP:journals/corr/abs-2102-00876,dl202l} and, in the context of explainable AI, also~\cite{DBLP:conf/aips/CamachoM19} for an overview.

In the database and KR communities, there has been extensive work on identifying queries and concept descriptions from data examples. For instance, in reverse engineering of queries, the goal is typically to decide whether there is a query that fits (or separates) a set of positive and negative examples. Relevant work under the closed world assumption include~\cite{DBLP:journals/tods/ArenasD16,DBLP:conf/icdt/Barcelo017} and under the open world assumption~\cite{GuJuSa-IJCAI18,DBLP:conf/ijcai/FunkJLPW19}.
Related work on active learning not yet discussed include the identification of
$\mathcal{EL}$-queries~\cite{DBLP:conf/ijcai/FunkJL21} and ontologies~\cite{DBLP:conf/aaai/KonevOW16,DBLP:journals/jmlr/KonevLOW17}, and of schema-mappings~\cite{DBLP:journals/tods/CateDK13,DBLP:conf/pods/CateK0T18}. Again this work differs from our contribution as it focuses on learning using membership and equivalence queries rather than only the former.
The use of unique characterisations to explain and construct schema mappings has been promoted and investigated by Kolaitis~(\citeyear{kolaitis:LIPIcs:2011:3359}) and Alexe et al.~(\citeyear{DBLP:journals/tods/AlexeCKT11}).

Combining \LTL{} and description logics for temporal conceptual modelling and data access has a long tradition~\cite{DBLP:conf/time/LutzWZ08,DBLP:conf/time/ArtaleKKRWZ17}. For querying purposes,  sometimes description logic concepts have been replaced by general CQs. Our restriction to $\mathcal{ELI}$-concepts instead of general CQs is motivated by results of~\cite{DBLP:conf/icdt/CateD21} showing that only CQs that are acyclic modulo cycles through the answer variables are polynomially characterisable within the class of CQs. Hence very strong acyclicity conditions are needed to ensure polynomial characterisability. We conjecture that our results can be extended to this class.

The class of queries in which no $\exists P$ is within the scope of temporal operators was first introduced by~\cite{DBLP:journals/ws/BaaderBL15,DBLP:conf/ijcai/BorgwardtT15} in the context of monitoring applications. The lcs and msc in temporal DLs are considered by Tirtarasa and Turhan~\citeyear{TiTu-SAC2022}.


\section{Preliminaries}\label{Sec:prelims}

By a \emph{signature} we mean any finite set $\sigma \ne \emptyset$ of \emph{atomic concepts} $A,B,C,\dots$ representing observations, measurements, events, etc. A $\sigma$-\emph{data instance} is any finite sequence $\D = (\delta_0, \dots, \delta_n)$ with $\delta_i \subseteq \sigma$, saying that $A \in \delta_i$ happened at moment $i$. The \emph{length} of $\D$ is $\len(\D) = n$ and the \emph{size} of $\D$ is $|\D| = \Sigma_{i\le n} |\delta_i|$. We do not distinguish between $\D$ and its variants of the form $(\delta_0,\dots,\delta_n, \emptyset, \dots, \emptyset)$.

We access data by means of \emph{queries}, $\q$, constructed from atoms, $\bot$ and $\top$ using $\land$ and the temporal operators $\nxt$, $\Diamond$, $\Diamondw$  and $\U$. The set of atomic concepts occurring in $\q$ is denoted by $\sig(\q)$. The set of queries that only use the operators from $\Phi \subseteq \{\nxt, \Diamond, \Diamondw, \U\}$ is denoted by $\mathcal{Q}[\Phi]$; $\mathcal{Q}^\sigma[\Phi]$ is the restriction of $\mathcal{Q}[\Phi]$ to a signature $\sigma$. The \emph{size} $|\q|$ of $\q$ is the number of symbols in $\q$, and the \emph{temporal depth} $\dep(\q)$ of $\q$ is the maximum number of nested temporal operators in $\q$.

$\mathcal{Q}[\nxt, \Diamond, \Diamondw]$-queries can be equivalently defined as \emph{tree-shaped} conjunctive queries (CQs) with the binary predicates $\suc$, $<$, $\leq$ over $\mathbb N$, and atomic concepts as unary predicates. Each such CQ is a set $Q(t_{0})$ of assertions of the form $A(t)$, $\suc(t,t')$, $t<t'$, and $t\leq t'$, with a distinguished variable $t_0$, such that, for every variable $t$ in $Q(t_{0})$, there exists exactly one path from $t_{0}$ to $t$ along the binary predicates $\suc$, $<$, $\leq$.

The set of $\mathcal{Q}[\nxt, \Diamond, \Diamondw]$-queries with \emph{path-shaped} CQ counterparts is denoted by $\Qp[\nxt,\Diamond, \Diamondw]$. Such queries $\q$ take the form~\eqref{dnpath}, where $\op_i \in \{\nxt, \Diamond, \Diamondw \}$ and $\rho_i$ is a conjunction of atoms (the empty conjunction is $\top$).
%
%
Similarly, $\Qp[\U]$-queries take the form~\eqref{upath}.


Given a data instance $\D = (\delta_0,\dots,\delta_n)$, the \emph{truth-relation} $\D,\ell \models \q$, for $\ell < \omega$, is defined as follows:
\begin{align*}
& \D,\ell \models A \ \text{ iff } \ A \in \delta_\ell, \qquad \D,\ell \models \top, \qquad \D,\ell \not\models \bot,\\
& \D,\ell \models \q_1 \land \q_2 \text{ iff } \D,\ell \models \q_1 \text{ and } \D,\ell \models \q_2,\\
& \D,\ell \models \nxt \q \ \text{ iff } \ \D,\ell+1 \models \q,\\
& \D,\ell \models \Diamond \q \ \text{ iff } \ \D,m \models \q, \text{ for some $m > \ell$}, \\
& \D,\ell \models \Diamondw \q \ \text{ iff } \ \D,m \models \q, \text{ for some $m \ge \ell$}, \\
& \D,\ell \models \q_1 \U \q_2 \ \text{ iff \ \ there is $m > \ell$ such that } \D,m \models \q_2\\
& \hspace*{2.2cm} \text{ and } \D,k \models \q_1, \text{ for all $k$ with $\ell < k < m$}.
\end{align*}
Note that $\D,n \models \Diamond \top \land \nxt \top \land (\q \U \top)$ as $(\delta_0,\dots,\delta_n,\emptyset)$ is a variant of $\D$.
We write $\q \models \q'$ if $\D,\ell \models \q$ implies \mbox{$\D,\ell \models \q'$} for any $\D$ and $\ell$. If $\q \models \q'$ and $\q' \models \q$, we call $\q$ and $\q'$ \emph{equivalent} and write $\q \equiv \q'$. Since $\nxt \q \equiv \bot \U \q$, $\Diamond \q \equiv \top \U \q$ and $\Diamond \q \equiv \nxt\Diamondw \q$, one can assume that $\mathcal{Q}[\nxt,\Diamond] \subseteq \mathcal{Q}[\U]$, $\mathcal{Q}[\Diamond] \subseteq \mathcal{Q}[\nxt, \Diamondw]$ and $\mathcal{Q}[\nxt,\Diamondw] = \mathcal{Q}[\nxt, \Diamondw,\Diamond]$.


\section{Unique Characterisability}
\label{sec:chara}

An \emph{example set} is a pair $E = (E^+,E^-)$ with finite sets $E^+$ and $E^-$ of data instances  that are called \emph{positive} and \emph{negative examples}, respectively. A query $\q$ \emph{fits} $E$ if $\D^+,0 \models \q$ and $\D^-,0 \not\models \q$, for all $\D^+ \in E^+$ and $\D^- \in E^-$. We say that $E$ \emph{uniquely} \emph{characterises} $\q$ within a class $\mathcal{Q}$ of queries if $\q$ fits $E$ and $\q \equiv \q'$ for any $\q' \in \mathcal{Q}$ that fits $E$. If all $\q \in \mathcal{Q}$ are characterised by some $E$ within $\mathcal{Q}' \supseteq \mathcal{Q}$, we call $\mathcal{Q}$ \emph{uniquely} \emph{characterisable} within $\mathcal{Q}'$. Further, $\mathcal{Q}$ is \emph{polynomially characterisable} within $\mathcal{Q}' \supseteq \mathcal{Q}$ if there is a polynomial $f$ such that every $\q \in \mathcal{Q}$ is characterised within $\mathcal{Q}'$ by some $E$ of size $|E| \le f(|\q|)$, where $|E| = \Sigma_{\D\in (E^+ \cup E^-)} |\D|$. Let $\mathcal{Q}^{n}$ be the set of queries in $\mathcal{Q}$ of size at most $n$. We say that $\mathcal{Q}$ is \emph{polynomially characterisable for bounded query size} if there is a polynomial $f$ such that every $\q \in \mathcal{Q}^{n}$ is characterised by some $E$ of size $\le f(n)$ within $\mathcal{Q}^{n}$.

Observe that (polynomial) characterisability is anti-monotone:
if a query $\q$ is (polynomially) characterisable within $\mathcal{Q}$ and $\mathcal{Q}'\subseteq\mathcal{Q}$, then $\q$ is (polynomially)
characterisable within $\mathcal{Q}'$. In counterexamples to characterisability,
we therefore only provide the smallest natural class of queries within which non-characterisability holds.
The following examples illustrate (non-)characterisability within the classes $\mathcal{Q}_{p}[\Diamondw]$ and $\mathcal{Q}_{p}[\nxt,\Diamondw]$.
\begin{example}\label{EX:lone}\em
	$(i)$ Recall from Section~\ref{intro} that \mbox{$\Diamondw (A \land B)$} is not uniquely characterisable within $\Qp[\Diamondw]$. The same argument shows non-characterisability of \mbox{$\Diamond (A \land B)$} within $\Qp[\Diamondw,\Diamond]$. On the other hand, the query $\Diamond (A \wedge B)$ is characterised within $\mathcal{Q}_{p}[\Diamond,\nxt]$ by the example set with positive examples $(\emptyset,\{A,B\})$ and $(\emptyset,\emptyset,\{A,B\})$ and negative examples $(\emptyset,\{A\})$ and $(\emptyset,\{B\})$.
	
	$(ii)$ The conjunction of atoms does not always lead to non-characterisability within classes of queries with $\Diamondw$.
	For example, $\q = \Diamondw (A \wedge \nxt (A\wedge B))$ is characterised within $\mathcal{Q}_{p}[\nxt,\Diamondw]$ by $E=(E^+,E^-)$ in which $E^+$ contains two data instances $(\{A\},\{A,B\})$ and $(\emptyset,\{A\}, \{A,B\})$ and $E^-$ also two instances:
\begin{align*}
& (\emptyset,\emptyset,\{A, B\}), \quad 
(\emptyset,\{A\},\{A\},\{B\},\{A, B\}).
\end{align*}
	The intuition here is that some instances from $E^{-}$ have to satisfy the query $\Diamondw(A \wedge \nxt(B \wedge \Diamondw (A \land B)))$ as well as the query $\Diamondw (A \wedge \nxt(A \wedge \Diamondw (A \land B)))$.
	
	$(iii)$ While the query $\Diamondw (A \wedge B)$ from $(i)$ is not characterisable,
	there is a polynomial $f$ such that, for all $n \in \mathbb{N}$, it is
	characterisable within $\Qp^n[\nxt,\Diamondw]$ by some $E_n$ of size $\le f(n)$.
	Namely, we take $E^+=\{(\{A,B\}), (\emptyset,\{A,B\})\}$ and
	$E^-=\{(\underbrace{ \{A\},\{B\},\ldots,\{A\},\{B\} }_{n\text{ times}})\}$.
\end{example}
Observe that one can always separate $\q\in \mathcal{Q}[\nxt,\Diamondw]$
from any other $\q'\in \mathcal{Q}[\nxt,\Diamondw]$ with $\textit{sig}(\q') \supsetneq \text{sig}(\q) = \sigma$ using the positive example $(\sigma,\dots,\sigma)$ with $\dep(\q) +1$-many copies of $\sigma$. One can therefore focus on characterisability within the relevant class of queries over the same signature as the input query. However, this is not the case for $\mathcal{Q}[\U]$:
\begin{example}
\label{ex:second}\em
The query $\q = \bot \U A \equiv \nxt A$ is not uniquely characterisable within $\mathcal{Q}_{p}[\U]$. Indeed, suppose $\q$ fits $E$ and $\sigma$ comprises all atoms occurring in $E$. Then \mbox{$\mathcal{D},0 \models C \U A$} iff
$\mathcal{D},0\models \nxt A$, for all $\mathcal{D}$ in $E$ and $C\not\in\sigma$, and
so $E$ does not characterise $\q$. On the other hand, for the signature $\sigma = \{A,B\}$, the query $\q$ is characterised within $\mathcal{Q}^\sigma_{p}[\U]$ by the example set  $(E^+,E^-)$ in which  $E^+ = \{(\emptyset, \{A\})\}$ and $E^-= \{(\sigma, \{B\}, \{A\})\}$  
as $A \U A \equiv (A\land B) \U A \equiv \nxt A$.
\end{example}

As noted in Section~\ref{intro}, $\bot \U A$ is not uniquely characterisable within $\mathcal{Q}^{\{A\}}[\U]$ because of nested $\U$-operators on the left-hand side of $\U$. This observation prompts us to consider the subclass $\mathcal{Q}_-^\sigma[\U]$ of $\mathcal{Q}^\sigma[\U]$-queries $\q$ in which any subquery $\q' \U \q''$ does not contain occurrences of $\U$ in $\q'$. Note that $\Qp^\sigma[\U] \subseteq \mathcal{Q}_-^\sigma[\U]$.
We show that $\mathcal{Q}_-^\sigma[\U]$ is uniquely characterisable. To simplify notation, we give $\sigma$-data instances as \emph{words} over the alphabet $2^\sigma$ using the standard notation of regular languages. Instead of $\D,0 \models \q$ we simply write $\D \models \q$. By the semantics of $\U$, for any $\q \in \mathcal{Q}_-^\sigma[\U]$, we have
\begin{align}\label{depth}
&\sigma^{d} \not\models \q \ \text{ for $d \le \dep(\q)$}, \quad \sigma^{d} \models \q \ \text{ for $d > \dep(\q)$}
\end{align}
where $\sigma^d$ is a word with $d$-many $\sigma$.
Note also that there are finitely-many, say $N_d < \omega$, pairwise non-equivalent queries of any depth $d < \omega$ in $\mathcal{Q}_-^\sigma(\U)$.

\begin{lemma}\label{gen-size}
If $\q,\q' \in \mathcal{Q}_-^\sigma[\U]$ are of depth $d$ and $\q \not\models \q'$, then there is  $\D$ such that $\len(\D) \le N_d$, $\D\models \q$ and $\D\not\models \q'$.
\end{lemma}
\begin{proof}
Consider $\D$ of minimal length such that $\D \models \q$ and $\D \not\models \q'$.  Let $\tp(i)$ comprise all of the subqueries $\s$ of $\q$ and $\q'$ with $\D,i \models \s$.
By the choice of $\D$, we have \mbox{$\tp(i) \ne \tp(j)$} for any distinct $i,j \in [0,\len(\D)]$ (otherwise we could cut the interval $[i,j)$ out of $\D$ to obtain a shorter instance separating $\q$ from $\q'$). It follows that $\len(\D) \le N_d$.
\end{proof}

\begin{theorem}\label{th:u-uc}
For any $\sigma$, $\mathcal{Q}_-^\sigma[\U]$ is uniquely characterisable.
\end{theorem}
\begin{proof}
Any $\q \in \mathcal{Q}_-^\sigma(\U)$ is uniquely characterised by $E$ with
\begin{align*}
& E^+ = \{ \D \models \q \mid \len(\D) \le N_{\dep(\q)} \},\\
& E^- = \{ \D \not\models \q \mid \len(\D) \le N_{\dep(\q)} \}.
\end{align*}
Indeed, let $\q' \in \mathcal{Q}_-^\sigma(\U)$ fit $E$. Then $\dep(\q') = \dep(\q)$ by \eqref{depth}, and so $\q \equiv \q'$ by Lemma~\ref{gen-size}.
\end{proof}

It follows from the proof that $\mathcal{Q}[\nxt,\Diamond]$ is uniquely characterisable as well.


%
%


\section{Characterisability in $\Qp[\nxt,\Diamondw]$}\label{Sec:weak-path}

In this section, we prove a criterion of (polynomial) unique characterisability of queries within $\Qp[\nxt,\Diamondw]$. The criterion is applicable to $\Qp[\nxt,\Diamond, \Diamondw]$-queries in a normal form, which is defined and illustrated  below.

\begin{example}\label{Ex:normal}\em
It is readily checked that the $\Qp[\nxt,\Diamondw]$-query
$
\q = \nxt \Diamondw \nxt \Diamondw (A \land B \land C \land \Diamondw (B \land \Diamondw ( B \land C)))
$
is equivalent to the $\Qp[\nxt,\Diamond,\Diamondw]$-query
$
\q^\textit{nf} = \Diamond\Diamond (A \land B \land C).
$
\end{example}

We define the normal form for $\Qp[\nxt,\Diamond,\Diamondw]$-queries represented as a  first-order CQ by a list of atoms. For example, the query $\q^\textit{nf}$ above is given by the CQ
\begin{align*}
\q^\textit{nf}(t_0) = t_0 < t_1, t_1 < t_2, A(t_2), B(t_2), C(t_2)
\end{align*}
with one free (answer) variable $t_0$ and existentially quantified $t_1$ and $t_2$.
In general, any $\q \in \Qp[\nxt,\Diamond,\Diamondw]$ is represented as a CQ
$$
\rho_{0}(t_0), R_{1}(t_0,t_1), \dots, \rho_{n-1}(t_{n-1}), R_{n}(t_{n-1},t_n),\rho_{n}(t_n),
$$
where $\rho_i$ is a set of atoms, $\rho_i(t_i) = \{A(t_i) \mid A \in \rho_i\}$ and $R_{i}\in \{\suc,<,\leq\}$.
We divide $\q$ into \emph{blocks} $\q_i$ such that
\begin{align}\label{fullq}
\q = \q_{0} \mathcal{R}_{1} \q_{1} \dots \mathcal{R}_{n} \q_{n}
\end{align}
with $\mathcal{R}_{i} = R_{1}^{i}(t_{0}^{i},t_{1}^{i}) \dots  R_{n_{i}}^{i}(t_{n_{i}-1}^{i},t_{n_{i}}^{i})$, for \mbox{$R_{j}^{i}\in \{<,\leq\}$},
\begin{align*}
\q_{i}= \rho_{0}^{i}(s_{0}^{i}) \suc (s_{0}^{i},s_{1}^{i}) \rho_{1}^{i}(s_{1}^{i}) \dots \suc(s_{k_{i}-1}^{i},s_{k_{i}}^{i}) \rho_{k_{i}}^{i}(s_{k_{i}}^{i})
\end{align*}
and $s_{k_{i}}^{i}=t_{0}^{i+1}$, $t_{n_{i}}^{i}=s_{0}^{i}$. If $k_{i}=0$, the block $\q_{i}$ is \emph{primitive}. A primitive block $\q_{i}=\rho_{0}^{i}(s_{0}^{i})$ with $i>0$ and $|\rho_{0}^{i}|\geq 2$ is called a \emph{lone conjunct} of $\q$.
\begin{example}\label{exm:abab}\em
The query $\Diamondw(A\wedge B)$ in Example~\ref{EX:lone}$(i)$, whose CQ representation is $t_{0}\leq t_{1}, \rho_{1}(t_{1})$, for $\rho_{1}=\{A,B\}$, has a lone conjunct $\rho_{1}(t_{1})$. In \mbox{$\Diamondw(A \wedge \nxt(A \wedge B))$} from Example~\ref{EX:lone}$(ii)$, represented as $t_{0}\leq t_{1}$,  $A(t_{1}),\suc(t_{1},t_{2}), \rho_{1}(t_{2})$, the conjunct $\rho_{1}(t_{2})$ is not lone.
\end{example}

Now, we say that $\q$ given by~\eqref{fullq} is in \emph{normal form} if the following conditions are satisfied:
\begin{description}
\item[(n1)] $\rho_{0}^{i}\not=\emptyset$ if $i>0$, and $\rho_{k_{i}}^{i}\not=\emptyset$ if $i>0$ or $k_{i}>0$
(thus, of all the first/last $\rho$ in a block, only $\rho_0^0$ can be empty);

\item[(n2)] each $\mathcal{R}_{i}$ is either a single $t_{0}^{i}\leq t_{1}^{i}$ or a sequence of $<$;
	
\item[(n3)] $\rho_{k_{i}}^{i}\not\supseteq \rho_{0}^{i+1}$ if $\q_{i+1}$ is primitive and $R_{i+1}$ is $\le$;

\item[(n4)] $\rho_{k_{i}}^{i}\not\subseteq \rho_{0}^{i+1}$ if $i>0$, $\q_{i}$ is primitive and $R_{i+1}$ is $\le$.
\end{description}
The queries in Example~\ref{exm:abab} are in normal form with two blocks each; the query $\q^\textit{nf}$ above is in normal form with two blocks $\q_0 = \top(t_0)$ and $\q_1 = A(t_2) \land B(t_2) \land C(t_2)$.

\begin{lemma}
Every query in $\Qp[\nxt, \Diamondw]$ is equivalent to a query in normal form that can be computed in linear time.
\end{lemma}

A query $\q \in \mathcal{Q}_{p}[\nxt,\Diamond]$ is \emph{safe} if it is equivalent to a query $\q'\in \mathcal{Q}_{p}[\nxt,\Diamond]$ in normal form not containing lone conjuncts.
We are now in the position to formulate the criterion.
\begin{theorem}\label{thm:nextdiamond}
$(i)$ A query $\q\in \Qp[\nxt, \Diamondw]$ is uniquely characterisable within $\Qp[\nxt, \Diamondw]$ iff $\q$ is safe.

$(ii)$ Those queries that are uniquely characterisable within $\Qp[\nxt, \Diamondw]$ are actually polynomially characterisable within $\Qp[\nxt, \Diamondw]$.

$(iii)$ The class $\Qp[\nxt, \Diamondw]$ is polynomially characterisable for bounded query size.

$(iv)$ The class $\Qp[\nxt,\Diamond]$ is polynomially characterisable.
\end{theorem}
\begin{proofsketch}
A detailed proof is given in the Appendix. Here, we define
a polysize example set $E=(E^{+},E^{-})$ characterising a query $\q$
in normal form~\eqref{fullq}, which does not contain lone conjuncts. Let $b$ be the number of $\nxt$ and $\Diamond$ in $\q$ plus 1. For each block $\q_i$ in~\eqref{fullq}, we take two words
$$
\qw_{i} = \rho_{0}^{i} \dots \rho_{k_{i}}^{i}, \quad \qw_{i} \Join \qw_{i+1} = \rho_{0}^{i} \dots (\rho_{k_{i}}^{i}\cup \rho_{0}^{i+1}) \dots \rho_{k_{i+1}}^{i+1}.
$$
The set $E^{+}$ contains the data instances given by the words
\begin{itemize}
\item[--] $\D_{b} = \qw_{0} \emptyset^{b} \dots \qw_{i} \emptyset^{b} \qw_{i+1} \dots \emptyset^{b} \qw_{n}$,

\item[--] $\D_{i} = \qw_{0} \emptyset^{b} \dots \qw_{i} \! \Join \! \qw_{i+1} \dots \emptyset^{b} \qw_{n}$ if $\mathcal{R}_{i+1}$ is $\leq$,

\item[--] $\D_{i} = \qw_{0} \emptyset^{b} \dots \qw_{i} \emptyset^{n_{i+1}} \qw_{i+1} \dots \emptyset^{b} \qw_{n}$ otherwise.
\end{itemize}
Here, $\emptyset^b$ is a sequence of $b$-many $\emptyset$ and similarly for $\emptyset^{n_{i+1}}$.
The set $E^{-}$ contains all data instances of the form
\begin{itemize}
\item[--] $\D_i^- = \qw_{0} \emptyset^{b} \dots \qw_{i} \emptyset^{n_{i+1} - 1} \qw_{i+1} \dots \emptyset^{b} \qw_{n}$ if $n_{i+1} > 1$;

\item[--] $\D^-_i = \qw_{0} \emptyset^{b} \dots \qw_{i} \! \Join \! \qw_{i+1} \dots \emptyset^{b} \qw_{n}$ if $\mathcal{R}_{i+1}$ is a single $<$,
\end{itemize}
and also the data instances obtained from $\D_{b}$ by
\begin{description}
\item[\rm (a)] removing a single atom from some $\rho^i_j \ne \emptyset$ or removing the whole $\rho^i_j = \emptyset$, for $i \ne 0$ and $j \ne 0$, from some $\qw_i$;

\item[\rm (b)] replacing $\qw_i = \rho_0^i \dots \rho^i_l \rho^i_{l+1} \dots \rho_{k_i}^i$ ($k_i > 0$) by $\qw'_i \emptyset^b \qw''_i$, where $\qw'_i = \rho_0^i \dots \rho^i_l$, $\qw''_i = \rho^i_{l+1} \dots \rho_{k_i}^i$ and $l \ge 0$;


\item[\rm (c)] replacing some $\rho_l^i \ne \emptyset$, $0 < l < k_i$, by $\rho_l^i \emptyset^b \rho_l^i$;


\item[\rm (d)] replacing $\rho^i_{k_i}$ ($k_i > 0$, $|\rho^i_{k_i}| \ge 2$) with $\rho^i_{k_i} \setminus \{A\} \emptyset^b \rho^i_{k_i}$, for some $A \in \rho^i_{k_i}$, or replacing $\rho^i_{0}$ ($k_i > 0$, $|\rho^i_{0}| \ge 2$) with $\rho^i_{0}  \emptyset^b \rho^i_{0}\setminus \{A\}$, for some $A \in \rho^i_{0}$;





\item[\rm (e)] replacing $\rho_0^0 \ne \emptyset$ with $\rho^0_0 \setminus \{A\} \emptyset^b \rho^0_0$, for some $A\in \rho^0_0$,\\[1pt] if $k_0 = 0$, and with $\rho_0^0 \emptyset^b \rho_0^0$ if $k_0 > 0$.
\end{description}
The size of $E$ is clearly polynomial in $|\q|$. It is readily seen that $\D \models \q$ for all $\D \in E^+$. To continue the proof sketch, note that $\D \models \q$ iff there is a \emph{homomorphism} $h$ from the set $\var(\q)$ of variables in $\q$ to $[0,\len(\D)]$, i.e., $h(t_0) = 0$, $A(h(t)) \in \D$ if $A(t) \in \q$, $h(t') = h(t) +1$ if $\suc(t,t') \in \q$, and $h(t)\, R\, h(t')$ if $R(t,t') \in \q$ for $R \in \{<,\leq\}$. Using the assumption that $\q$ is in normal form, one can  show that there is no homomorphism witnessing $\D\models \q$, for any $\D\in E^{-}$.

Suppose now that $\q'\in \mathcal{Q}_{p}[\nxt,\Diamondw]$ in normal form is given and $\q' \not\equiv \q$. If $\D_{b}\not\models \q'$, we are done as $\D_{b}\in E^{+}$. Otherwise, let $h$ be a homomorphism witnessing $\D_{b}\models \q'$. Then one can show that either the restriction of $h$ to the blocks of $\q'$ is an isomorphism onto the blocks of $\q$ or there exists a data instance $\D$ obtained using one of the rules (a)--(e) such that a suitably modified $h$ is a homomorphism from $\q'$ to $\D$. In the latter case, we are done as $\D\in E^{-}$ and $\D\models \q'$. In the former case, $\q$ and $\q'$ coincide with the exception of the sequences of $\Diamond$ and $\Diamondw$
between blocks. Then $\q$ can be separated from $\q'$ using the examples $\D_{i}$ and $\D_{i}^{-}$.
\end{proofsketch}


\section{Polynomial Characterisability in $\Qp^\sigma[\U]$}\label{Sec:until}


\LTL-queries with $\U$ do not correspond to CQs (because of the universal quantification in its semantics), and so require a different approach. We view  them as defining regular languages.
With each $\Qp^\sigma[\U]$-query of the form~\eqref{upath} we associate the following  regular expression over the alphabet $2^\sigma$:
\begin{equation}\label{regupath}
\q = \rho_0 \lambda_1^* \rho_1 \lambda_2^* \dots \lambda_n^* \rho_n \lambda_{n+1}^*
\end{equation}
%
%
where $\lambda_{n+1} = \emptyset$ and $\bot^* = \varepsilon$. We regard the words of the language $\L(\q)$ over $2^\sigma$ as data instances. Clearly, $\D' \models \q$ iff there is $\D \in \L(\q)$ such that $\D \Subset \D'$, i.e., $\D = (\delta_0, \dots, \delta_k)$ and $\D' = (\delta'_0, \dots, \delta'_k)$, for some $k < \omega$, and $\delta_i \subseteq \delta'_i$, for all $i \le k$.
The language $\boldsymbol{L}_{\q}$ of all $\sigma$-data instances $\D\models\q$ (regarded as words over $2^\sigma$) can be given by the NFA $\mathfrak A_{\q}$ below, where each $\to_\alpha$, for $\alpha \ne \bot$, stands for all transitions $\to_\beta$ with $\alpha \subseteq \beta \subseteq \sigma$ (note that $\bot \notin\sigma$):\\[-5pt]
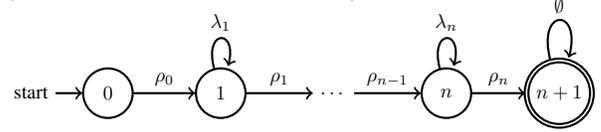
\centerline{
\begin{tikzpicture}[->,thick,node distance=2cm, transform shape, scale=0.75]
\node[state, initial] (i) {$0$};
\node[state, right  of =i] (i1) {$1$};
\node[right  of= i1] (dots) {$\dots$};
\node[state, right  of= dots] (s1) {$n$};
\node[state,accepting, right of=s1] (s0) {$n+1$};
\draw (i) edge [above] node{$\rho_0$} (i1)
(i1) edge [above] node{$\rho_{1}$} (dots)
(i1) edge [loop above, above] node{$\lambda_{1}$} (i1)
(dots) edge [above] node{$\rho_{n-1}$} (s1)
(s1) edge [above] node{$\rho_{n}$} (s0)
(s1) edge [loop above, above] node{$\lambda_{n}$} (s1)
(s0) edge [loop above, above] node{$\emptyset$} (s0)
;
\end{tikzpicture}
}
\\
Without loss of generality we assume that all our $\q$ are \emph{minimal} in the sense that by replacing any $\lambda_i \ne \bot$ with $\bot$ in $\q$ we obtain a query that is \emph{not equivalent} to $\q$. For example, in  minimal $\q$, $\rho_j \supseteq \dots \supseteq \rho_i\supseteq \lambda_{i}$ and $\lambda_{l}=\bot$ for all $l \in (j,i)$  imply $\rho_j \not \subseteq \lambda_{j}$ as otherwise
$
\lambda_{j} \U (\rho_j \land (\bot\U\ldots (\lambda_{i} \U \varphi) \dots))$ is equivalent to $\bot \U (\rho_j \land (\bot\U\ldots (\lambda_{i} \U \varphi) \dots))$.
%
%
%
Using standard automata-theoretic techniques, one can show:

\begin{theorem}\label{separUpath}
Any $\Qp^\sigma[\U]$-queries $\q \not\equiv \q'$ can be separated by some $\D$ with $\max(\D) \le O((\min \{\dep(\q), \dep(\q')\})^2)$.
\end{theorem}

Using Theorem~\ref{separUpath} in the proof of Theorem~\ref{th:u-uc} we obtain:

\begin{corollary}
The class $\Qp^\sigma[\U]$ is exponentially characterisable within $\Qp^\sigma[\U]$.
\end{corollary}

%
%

The following examples illustrate difficulties in finding short unique  characterisations of $\Qp^\sigma[\U]$-queries, namely, that in general, data instances of different shapes and forms are needed to separate $\Qp^\sigma[\U]$-queries. To unclutter notation we omit $\{\}$ in singletons like $\{A\}$.

\begin{example}\label{baddies}\em
%
(a) The shortest data instance separating
\begin{align*}
\q & =  X \emptyset^*  A \bot^*B  \bot^* A B^* A A^*B\emptyset^* , \\
\q' & = X \emptyset^* A \bot^* B A^* A  B^* A \bot^* B\emptyset^*
\end{align*}
%
is $\D = XABABBAAB$ with $\D \models \q$ and $\D \not\models \q'$ (e.g.,  $XABABAAB$ satisfies both $\q$ and $\q'$).


(b) For $l>0$, let $\q_l=(AB^*)^{l-1} AA^*BB^*$. Then
\begin{align*}
& XA^*\q_{l_1}\q_{l_2}\dots\q_{l_k}X\emptyset^* \ \not\equiv \ X \bot^* \q_{l_1} \q_{l_2} \dots \q_{l_k} X\emptyset^* ,\\
& XA^*\q_{l_1}\q_{l_2}\dots\q_{l_k} A \emptyset^* \ \equiv \ X \bot^* \q_{l_1} \q_{l_2} \dots \q_{l_k} A\emptyset^*.
\end{align*}
If $1<l_1\le\dots\le l_k$, the former inequivalence is witnessed by the instance $XA^{l_1}BA^{l_2}B\dots A^{l_k}BA^{l_k}BX$.
Less generally, $XA^* \q_2\q_3X\emptyset^* \not\equiv X\bot^* \q_2\q_3X\emptyset^*$ can be shown by $XAABAAABAAABX$ or by $XAABABAABABX$ (spot the difference and see~$(\mathfrak n_2)$ below).
\end{example}


Here, we consider the class $\Qpi[\U]$ of \emph{peerless} queries given by~\eqref{regupath}, in which, for any $i$, either $\lambda_i = \bot$ or the sets $\lambda_i$ and $\rho_i$ are \emph{incomparable} with respect to $\subseteq$. Our main result is that $\Qpi[\U]$ is polynomially characterisable within $\Qp^\sigma[\U]$.

We start with a general observation. Consider two queries $\q = \rho_0 \lambda_1^*  \dots \lambda_n^* \rho_n\emptyset^* $ and $\q' = \rho_0 \mu_1^* \dots \mu_n^* \rho_n\emptyset^* $. We say that $\lambda_i \ne \bot$ \emph{subsumes} $\mu_j \ne \bot$ if either $i=j$ and $\mu_j \subseteq \lambda_i$, or $j < i$ and $\mu_j \rho_j \dots \rho_{i-1} \Subset \rho_j \dots \rho_{i-1} \lambda_i$, or $j > i$ and $\rho_i \dots \rho_{j-1} \mu_j \Subset \lambda_i \rho_i \dots \rho_{j-1}$. In the last two cases,
\begin{equation*}
\mu_j \subseteq \rho_j \subseteq  \dots \subseteq \rho_{i-1} \subseteq \lambda_i,\ \
 \mu_j \subseteq \rho_{j-1} \subseteq  \dots \subseteq \rho_{i} \subseteq \lambda_i,
\end{equation*}
respectively. Note that, for peerless $\q$, the last inclusion is impossible. If $\lambda_i$ and $\mu_j$ subsume each other, in which case $\lambda_i = \mu_j$, we call $(\lambda_i,\mu_j)$ a \emph{matching pair}. Observe also that, for $\D^{i}_{\q} = \rho_0\dots \rho_{i-1} \lambda_i \rho_i \dots \rho_n$, if
$\D^{i}_{\q}\models\q'$, then $\lambda_i$ subsumes some $\mu_j$: $\rho_0 \dots \rho_n \emptyset \Subset \D_\q^i$ means that $\lambda_i$ subsumes $\mu_{n+1} = \emptyset$, and $\rho_0 \dots \mu_j \dots \rho_n \Subset \D_\q^i$ that $\lambda_i$ subsumes $\mu_{j}$. The proof of the next lemma is given in the Appendix:

\begin{lemma}\label{matching}
For any queries $\q$ and $\q'$ as above, either $(i)$ each $\lambda_i \ne \bot$ subsumes $\mu_j$ occurring in some matching pair $(\lambda_k,\mu_j)$ or $(ii)$ $\q$ and $\q'$ are separated by a data instance of the form $\D^i_\q$ 
 or $\smash{\D^j_{\q'}}$
. Also, if $\q$ is peerless, $\lambda_i$ can only subsume $\mu_j$ in the matching pair $(\lambda_i,\mu_j)$ with $i \ge j$, in which case
$
\mu_j = \rho_j =  \dots = \rho_{i-1} = \lambda_i.
$
\end{lemma}

Note that the number of data instances of the form $\D^i_{\q'}$ for all possible $\Qp^\sigma[\U]$-queries $\q'$ can be exponential in $|\sigma|$. The following example indicates how to overcome this issue.

\begin{example}\label{twolambdas}\em
Let $\sigma = \{A,B,C,D,X\}$. To separate the query $X\{C,D\}^*A \emptyset^*$ from any  $X\lambda^*A\emptyset^*$ with $A,D \notin \lambda$, we can use $\D = X\sigma \setminus \{A,D\}A$.
\end{example}

\begin{theorem}\label{uno}
The class $\Qpi[\U]$ is polynomially characterisable within $\Qp^\sigma[\U]$.
\end{theorem}
\begin{proofsketch}
We show that any $\q = \rho_0 \lambda_1^* \rho_1 \lambda_2^* \dots \lambda_n^* \rho_n\emptyset^* $ in $\Qpi[\U]$ is characterised by the example set $E = (E^+,E^-)$ where $E^+$ contains all data instances of the following forms:
\begin{description}
\item[$(\mathfrak p_0)$] $\rho_0\dots\rho_n$,

\item[$(\mathfrak p_1)$] $\rho_0\dots\rho_{i-1}\lambda_i\rho_i \dots \rho_n = \D^i_\q$,

\item[$(\mathfrak p_2)$] 
$\rho_0 \dots \rho_{i-1} \lambda_{i}^k \rho_{i} \dots \rho_{j-1} \lambda_{j}\rho_j\dots\rho_n=\D^j_{i,k}$, for  $i < j$ and $k = 1,2$;

\end{description}
and $E^-$ has all instances that are \emph{not} in $\boldsymbol{L}(\q)$ of the forms:
\begin{description}
\item[$(\mathfrak n_0)$] 
$\sigma^{n}$ and $\sigma^{n-i}\sigma \setminus \{A\} \sigma^{i}$, for $A\in\rho_i$,

\item[$(\mathfrak n_1)$] $\rho_0\dots\rho_{i-1}\sigma \setminus \{A,B\}\rho_i\dots\rho_n$, for $A\in \lambda_i\cup\{\bot\}$ and $B\in\rho_i\cup\{\bot\}$,

\item[$(\mathfrak n_2)$]
for all $i$ and $A\in\lambda_i\cup\{\bot\}$, \emph{some} data instance
\begin{equation}\label{raznost}
\D^i_{\!A} = \rho_0 \dots \rho_{i-1} (\sigma\setminus\{A\})\rho_i\lambda_{i+1}^{k_{i+1}} \dots \lambda_{n}^{k_n}\rho_n,
%
\end{equation}
if any, such that $\max(\D^i_{\!A}) \le (n+1)^2$
%
%
%
and $\D^i_{\!A}\not\models \q^\dag$ for $\q^\dag$ obtained from $\q$ by replacing $\lambda_j$, for all $j \le i$, with $\bot$.

Note that $\D^i_{\!A} \not\models \q$ for peerless $\q$.
\end{description}
By definition, $\q$ fits $E$ and $|E|$ is polynomial in $|\q|$.
We prove in the Appendix that $E$ uniquely characterises $\q$.
\end{proofsketch}

One reason why this construction does not generalise to the whole $\Qp^\sigma[\U]$ is that $\D^i_{\!A} \not\models \q^\dag$ does not imply $\D^i_{\!A} \not\models \q$ for non-peerless $\q$, as shown by the following example:

\begin{example}\label{badexample}\em
Let $\q=XA^*AB^*A\bot^*AB^*AA^*BB^*X\emptyset^*$. For any data instance $\D^3_{\bot}$ satisfying~\eqref{raznost}---for example, $\D^3_{\bot} = XAA\sigma ABABX$---we have $\D^3_\bot\models\q$.
\end{example}


\section{Characterisability in $\mathcal{Q}[\Diamond]$}
\label{sec:branching}
In the previous two sections, we have investigated characterisability of path-shaped queries. Here, we first justify that restriction by exhibiting two examples that show how temporal branching can destroy polynomial characterisability in $\mathcal{Q}[\Diamond]$. Both examples make use of \emph{unbalanced} queries, in which different branches have different length. We then show that this is no accident: one can at least partially restore polynomial characterisability for classes without unbalanced queries.

We start by observing that, without loss of generality, it is enough to consider conjunctions of path queries only:

\begin{lemma}\label{lem:conjof}
	For every $\q\in \mathcal{Q}[\Diamond]$, one can compute
	in polynomial time an equivalent query of the form $\q_{1}\wedge \cdots \wedge\q_{n}$ with $\q_{i}\in \mathcal{Q}_{p}[\Diamond]$, for $1 \le i \leq n$.
\end{lemma}

The first example showing non-polynomial characterisability is rather straightforward but requires unbounded branching and an unbounded number of atoms.
We write queries $\q \in \mathcal{Q}_{p}^{\sigma}[\Diamond]$ of the form
\begin{equation}\label{sdnpath}
	\q = \rho_{0} \land \Diamond (\rho_1 \land \Diamond (\rho_2 \land \dots \land \Diamond \rho_n) )
\end{equation}
as words $\rho_{0}\rho_{1}\dots\rho_{n}$ over $2^{\sigma}$ (omitting but not forgetting $\lambda_i^* = \emptyset^*$ from~\eqref{regupath}) and also use $\rho_{0}\rho_{1}\ldots\rho_{n}$ to denote the data instance defined by $\q$.

\begin{example}
\label{thm:superpolb}\em	
Consider $\q_n = \s_1 \land \dots \land \s_n$, where $n \ge 2$ and each $\s_i$ is a word repeating $n$ times the sequence $A_1\dots A_n$ (of singletons) with omitted $A_i$.
	%
	%
Now, consider the queries $\q^{\avec{p}}_n = \q_n \land \avec{p}$, where $\avec{p} = \Diamond (A_{i_1} \land \Diamond (A_{i_2} \land \dots \land \Diamond A_{i_n}) )$ and $A_{i_1}\dots A_{i_n}$ is a permutation of $A_1\dots A_n$. Then $\q^{\avec{p}}_n \models \q_n$ and $\q_n \not\models \q^{\avec{p}}_n$ as shown by the data instance $\s_{i_1} \s_{i_2} \dots \s_{i_n}$. Moreover, if
	%
$\D \models \q_n$, $\D \not\models \avec{p}$ and $\avec{p}'\ne \avec{p}$, then $\D \models \avec{p}'$. It follows that, in any $E = (E^+,E^-)$ uniquely characterising $\q_n$, the set $E^+$ contains at least $n!$ data instances.
	%
\end{example}
The class $\mathcal{Q}_{\leq n}[\Diamond]$ of queries of \emph{branching factor} at most $n$ contains all queries in
$\mathcal{Q}[\Diamond]$ that are equivalent to a query of the form
$\q_{1} \wedge \cdots \wedge \q_{m}$ with $m\leq n$ and $\q_{i}\in \mathcal{Q}_{p}[\Diamond]$. We next provide an example of non-polynomial characterisability that requires four atoms and bounded branching only.
\begin{example}
\label{thm:cc}\em
Let $\sigma=\{A_{1},A_{2},B_{1},B_{2}\}$, $\q_{1}= \emptyset (\s\sigma)^{n} \s$, and
$\q_{2} = \emptyset\sigma^{2n+1}$, where
$
\s=\{A_{1},A_{2}\}\{B_{1},B_{2}\}$. Consider the set $P$ of $2^{n+1}$-many queries of the form
$\emptyset\s_{1}\dots \s_{n+1}
$
with $\s_{i}$ either $\{A_{1}\}\{A_{2}\}$ or $\{B_{1}\}\{B_{2}\}$. Then $\q_{1}\wedge \q_{2}\not\models \q$ for any
$\q\in P$ and, for any $\mathcal{D}$ with $\mathcal{D}\models \q_{1}\wedge \q_{2}$,  there is at most one $\q\in P$ with $\mathcal{D}\not\models \q$ (the proof is rather involved).
It follows that $\q_{1} \wedge \q_{2} \not\models \q_{1} \wedge \q_{2} \wedge \q$ for all $\q\in P$, but $2^{n+1}$ positive examples are needed to separate $\q_{1} \wedge \q_{2}$ from all $\q_{1} \wedge \q_{2} \wedge \q$ with $\q\in P$.
\end{example}

We next identify polynomially characterisable classes of $\mathcal{Q}[\Diamond]$-queries, assuming as before that $\rho_{n}\ne\emptyset$ in any $\q$ of the form \eqref{dnpath}. We call a query $\q_{1}\wedge \cdots \wedge \q_{n}\in \mathcal{Q}[\Diamond]$ with $\q_{1},\dots,\q_{n}\in \mathcal{Q}_{p}[\Diamond]$ \emph{balanced} if all $\q_{i}$ have the same depth; further, we call it \emph{simple} if, in each $\q_i$ given by~\eqref{dnpath}, $|\rho_j| = 1$ for all $j$. Let $\mathcal{Q}_{b}[\Diamond]$ denote the class of queries in $\mathcal{Q}[\Diamond]$ that are equivalent to a balanced query.
\begin{theorem}\label{thm:branching1}
$(i)$ The class of simple queries in $\mathcal{Q}_{b}[\Diamond]$ is polynomially characterisable within $\mathcal{Q}_{b}[\Diamond]$.
	
$(ii)$ For any $n$, the class $\mathcal{Q}_{b}[\Diamond]\cap \mathcal{Q}_{\leq n}[\Diamond]$ is polynomially characterisable.
\end{theorem}
\begin{proofsketch}
Let $\q\in \mathcal{Q}_{p}^{\sigma}[\Diamond]$. We start with a lemma on the existence of polynomial-size $\sigma$-data instances $\mathcal{D}_{\q,k}$ such that  $\mathcal{D}_{\q,k}\not\models \q$ and $\mathcal{D}_{\q,k}\models \q'$ for all $\q'\in \mathcal{Q}_{p}^{\sigma}[\Diamond]$ with $\q'\not\models \q$ and $\dep(\q')\leq k$. Note that such $\mathcal{D}_{\q,k}$ do not exist in general.
\begin{example}\em
	Let $\q=A\wedge B$. Then $A\not\models \q$ and $B\not\models\q$ but there does not exist any $\mathcal{D}_{\q,0}$ such that $\mathcal{D}_{\q,0}\not\models \q$, $\mathcal{D}_{\q,0}\models A$ and $\mathcal{D}_{\q,0}\models B$.
\end{example}

In the following lemma, we therefore assume that $\q$ does not speak about the initial timepoint.
\begin{lemma}\label{lem:positive}
Let $\q\in \mathcal{Q}_{p}^{\sigma}[\Diamond]$ be of the form $\Diamond\q'$
	and let $k>0$. Then one can construct in polynomial time a $\sigma$-data instance
	$\mathcal{D}_{\q,k}$ such that $\mathcal{D}_{\q,k}\not\models\q$ and $\mathcal{D}_{\q,k}\models \q'$ for all $\q'\in \mathcal{Q}_{p}^{\sigma}[\Diamond]$
		with $\q'\not\models \q$ and $\dep(\q')\leq k$.
\end{lemma}
\begin{proof}
Assuming that
	$
	\q = \Diamond (\rho_1 \land \Diamond (\rho_2 \land \dots \land \Diamond \rho_n) )
	$
	with $\rho_{i}= \{A_{1}^{i},\dots , A_{n_{i}}^{i}\}$ for $i\geq 1$, we set
	$$
	\mathcal{D}_{\q,k}= \sigma\s_{1}^{k}\sigma\cdots \s_{n-1}^{k}\sigma \s_{n}^{k},
	$$
	where $\s_{i}= \sigma\setminus\{A^{i}_{1}\}\dots\sigma\setminus\{A^{i}_{n_{i}}\}$. One can show by induction that $\mathcal{D}_{\q,k}$ is as required.
\end{proof}
Using Lemma~\ref{lem:positive}, for any $\q\in \mathcal{Q}^{\sigma}[\Diamond]$, one can construct  a polynomial-size
set of negative examples as follows. Suppose $\q=\q_{1}\wedge \cdots \wedge \q_{n}\in \mathcal{Q}^{\sigma}[\Diamond]$ with
$$
\q_{i} = \rho_{0}^{i} \land \Diamond (\rho_1^{i} \land \Diamond (\rho_2^{i} \land \dots \land \Diamond \rho_{n_{i}^{i}}) ).
$$
Let $\rho=\bigwedge_{i=1}^{n} \rho_{0}^{i}$ and let $\q_{i}^{-}$ be $\q_{i}$ without the conjunct $\rho_{0}^{i}$, so Lemma~\ref{lem:positive} is applicable to $\q_{i}^{-}$.
Now let $E^{-}_{\q,m}$ contain the $\sigma$-data instances $\mathcal{D}_{\q_{i}^{-},m}$ and $\sigma\setminus\{A\}\sigma^{m}$ for all $A\in\rho$.
\begin{lemma}\label{lem:negexamples}
$(i)$ For any $\mathcal{D}\in E^{-}_{\q,m}$, we have $\mathcal{D}\not\models\q$.

$(ii)$ For any $\q'\in \mathcal{Q}^{\sigma}[\Diamond]$ with $\q'\not\models \q$ and $\dep(\q')\leq m$,  there exists $\mathcal{D}\in E^{-}_{\q,m}$ with $\mathcal{D}\models \q'$.
\end{lemma}
It follows from Lemma~\ref{lem:negexamples} that non-polynomial characterisability of $\mathcal{Q}[\Diamond]$-queries can only be caused by the need for super-polynomially-many positive examples. We now discuss the construction of positive examples in the proof of Theorem~\ref{thm:branching1} $(ii)$; part $(i)$ is dealt with in the Appendix. Let $\q= \q_{1}\wedge \cdots \wedge \q_{m}\in \mathcal{Q}^{\sigma}_{b}[\Diamond]\cap \mathcal{Q}^{\sigma}_{\leq n}[\Diamond]$ with $m\leq n$ and
$$
\q_{i} = \rho_{0}^{i}\wedge \Diamond (\rho_1^{i} \land \Diamond (\rho_2^{i} \land \dots \land \Diamond \rho_N^{i}) ).
$$
For any map $f \colon \{1,\ldots,m\}\rightarrow \{1,\dots,N\}$, construct a $\sigma$-data instance $\mathcal{D}_{\!f}$ by inserting $\rho^{i}_{f(i)}$ into the data instance $\sigma^{N}$ in position $f(i)$. Let $E^{+}$ contain the data instance $\rho\sigma^{N}$ for $\rho=\bigcup_{i=1}^{m}\rho_{0}^{i}$ and all the data instances $\mathcal{D}_{f}$. One can show that $(E^{+},E^{-})$ characterises $\q$ in $\mathcal{Q}_{b}[\Diamond]\cap \mathcal{Q}_{\leq n}[\Diamond]$.
\end{proofsketch}


\section{2D Temporal Instance Queries}
\label{sec:instquery}
Now we consider `two-dimensional' query languages that combine instance queries (over the object domain) in the standard description logics $\mathcal{EL}$ and $\mathcal{ELI}$~\cite{DL-Textbook} with the \LTL-queries (over the temporal domain) considered above. Our aim is to understand how far the characterisability results of the previous sections can be generalised to the 2D case.
A \emph{relational signature} is a finite set $\Sigma \ne \emptyset$ of unary
and binary predicate symbols. A $\Sigma$-\emph{data instance} $\A$ is a finite set of \emph{atoms} $A(a)$ and $P(a,b)$ with $A,P \in \Sigma$ and \emph{individual names} $a,b$. Let  $\ind(\A)$ be the set of individual names in~$\A$. We assume that $P^-(a,b)\in\A$ iff $P(b,a)\in\A$, calling $P^-$ the \emph{inverse} of $P$ (with $P^{--} = P$). Let $S := P \mid P^-$.
\emph{Temporal instance queries} are defined by the grammar
\begin{equation*}
\q  \ := \ \top \ \mid \  \bot \ \mid \ A \ \mid \ \exists S.\q \ \mid \ \q_1 \land \q_2 \ \mid \ \mathop{\boldsymbol{op}} \q \ \mid \ \q_1\U \q_2 ,
\end{equation*}
where $\boldsymbol{op} \in \{\nxt, \Diamond, \Diamondw\}$. Such queries without temporal operators are called $\mathcal{ELI}$-\emph{queries}; those of them without inverses $P^-$ are $\mathcal{EL}$-\emph{queries}. A \emph{temporal $\Sigma$-data instance} $\D$ is a finite sequence $\mathcal{A}_{0},\ldots,\mathcal{A}_{n}$ of $\Sigma$-data instances. We set $\ind(\D)=\bigcup_{i=1}^{n}\ind(\A_{i})$.
For any $\ell \in \mathbb N$ and $a \in \ind(\D)$, the \emph{truth-relation} $\D,a,\ell \models \q$  is defined by induction:
\begin{align*}
& \D,a,\ell \models A \ \text{ iff } \ A(a) \in \A_{\ell}, \\ 
& \D,a,\ell \models \exists S. \q \text{ iff \ there is $b \in \ind(\A_{\ell})$ such that } \\
&\mbox{}\hspace{4cm}  S(a,b) \in \A_{\ell}\text{ and } \D,b,\ell \models \q,
%
\end{align*}
with the remaining clauses being obvious generalisations of the \LTL{} ones.
An \emph{example set} is a pair $E = (E^+, E^-)$ with finite sets $E^+$ and $E^-$ of \emph{pointed temporal data instances} $(\D,a)$ such that $a \in \ind(\D)$. We say that $\q$ \emph{fits} $E$ if $\D^+,a^+,0 \models \q$ and $\D^-,a^-,0 \not\models \q$, for all $(\D^+,a^+) \in E^+$ and $(\D^-,a^-) \in E^-$\!. As before, $E$ \emph{uniquely characterises} $\q$ if $\q$ fits it and every $\q'$ fitting $E$ is logically equivalent to $\q$.


We need the following result on the unique characterisability of $\mathcal{ELI}$-queries.
\begin{theorem}[ten Cate and Dalmau 2021]\label{thm:tenCate0}
The class of $\mathcal{ELI}$-queries is polynomially characterisable.
\end{theorem}
Theorem~\ref{thm:tenCate} is proved by constructing frontiers in the set of $\mathcal{ELI}$-queries partially ordered by entailment, where a set $\mathcal{F}$ of $\mathcal{ELI}$-queries is called a \emph{frontier} of an $\mathcal{ELI}$-query $\q$ if the following hold:
\begin{itemize}
	\item $\q\models \q'$ and $\q'\not\models \q$, for all $\q'\in \mathcal{F}$;
	\item if $\q\models \q''$ for some $\mathcal{ELI}$-query $\q''$,
	then $\q''\models \q$ or there exists $\q'\in \mathcal{F}$ with $\q'\models \q''$.
\end{itemize}
\begin{theorem}[ten Cate and Dalmau 2021]\label{thm:tenCate} A frontier $\mathcal{F}(\q)$ of any $\mathcal{ELI}$-query $\q$ can be computed in polynomial time.
\end{theorem}

Theorem~\ref{thm:tenCate0} follows from Theorem~\ref{thm:tenCate}. For any $\mathcal{ELI}$-query $\q$, we denote by $\hat{\q}$ the tree-shaped data instance defined by $\q$ with designated root $a$. Then $\q$ is characterised by $E$ with $E^{+}=\{\hat{\q}\}$ and $E^{-}=\{\hat{\r} \mid \r \in \mathcal{F}(\q)\}$.

For any unrestricted temporal query language $\mathcal{Q}[\Phi]$ and $\mathcal{L}\in \{\mathcal{EL},\mathcal{ELI}\}$, we denote by $\mathcal{Q}[\Phi]\otimes \mathcal{L}$
the set of all temporal instance queries with operators in $\Phi$ with (for $\mathcal{ELI}$) or without (for $\mathcal{EL}$) inverse predicates. We generalise
the path-shaped queries $\mathcal{Q}_{p}[\Phi]$ as follows: $\mathcal{Q}_{p}[\Phi]\otimes \mathcal{L}$ denotes the class of queries $\q$ in $\mathcal{Q}[\Phi]\otimes\mathcal{L}$ such that, for any subquery $\q_{1}\wedge \q_{2}$ of $\q$,  either $\q_{1}$ or $\q_{2}$ do not have an occurrence of any operator in $\Phi$ that is not in the scope of $\exists S$. To illustrate, $\exists S.\Diamond A \wedge \Diamond \exists S.A$ is in $\mathcal{Q}_{p}[\Phi]\otimes\mathcal{L}$, but $\Diamond A \wedge \Diamond \exists S.A$ is not. We make two observations about unique characterisability in these `full' combinations.

\begin{theorem}\label{thm:firstone}
	$(i)$ $\mathcal{Q}[\nxt,\Diamond]\otimes \mathcal{EL}$ is uniquely characterisable.
	
	$(ii)$ $\Qp[\nxt] \otimes\mathcal{ELI}$ and $\Qp[\Diamond] \otimes \mathcal{ELI}$ are polynomially characterisable.
\end{theorem}
Here,~$(i)$ is shown similarly to Theorem~\ref{th:u-uc} (it remains open whether it can be extended  to $\mathcal{Q}[\nxt,\Diamond]\otimes \mathcal{ELI}$);~$(ii)$ is proved by generalising Theorem~\ref{thm:tenCate} to temporal data instances.

We now show that the application of the DL constructor $\exists P$ to temporal queries with both $\nxt$ and $\Diamond$ destroys polynomial characterisability.
Denote by $\mathcal{EL}(\mathcal{Q}_{p}[\nxt,\Diamond])$ the class
of queries in $\mathcal{Q}_{p}[\nxt,\Diamond] \otimes \mathcal{EL}$ that contain no  $\exists P$ in the scope of a temporal operator.
\begin{theorem}\label{thm:badd}
$\mathcal{EL}(\Qp[\nxt,\Diamond])$ is not polynomially characterisable.
\end{theorem}
%
\begin{proofsketch}
Consider queries
$\q_{n} = \exists P.\q_1^n \land \dots \land \exists P.\q_n^n$, in which each $\q_i^n$ corresponds to the regular expression
$$
\underbrace{BB\emptyset^*A}_1 \dots \underbrace{BB\emptyset^*A}_{i-1} \underbrace{\emptyset^*B\emptyset^*A}_i \underbrace{BB\emptyset^*A}_{i+1}  \dots \underbrace{BB\emptyset^*A}_n \emptyset^*
$$
(with omitted $\bot^* = \varepsilon$ in $BB$).
One can show that any unique characterisation of $\q_{n}$ contains at least $2^n$ positive examples to separate it from all queries $\q_{n} \land \exists P.\s$ with
$$
\s = \op_1 (B \land \Diamond (A \land \op_2 (B \land \Diamond (A \land \dots \land \op_n (B \land \Diamond A) \dots) ))),
$$
where $\op_i$ is $\nxt$ or $\Diamond \nxt$ if $i> 1$, and blank or $\Diamond$ if $i = 1$.
\end{proofsketch}

The situation changes drastically if we do not admit temporal operators in the scope of $\exists P$.
We start by investigating the class $\Qp[\nxt,\Diamondw](\mathcal{ELI})$ of queries of the form
$$
	\q = \r_0 \land \op_1 (\r_1 \land \op_2 (\r_2 \land \dots \land \op_n \r_n) ),
$$
where the $\r_{i}$ are $\mathcal{ELI}$-queries and $\op_i \in \{\nxt,\Diamondw\}$.
We can generalise the CQ-representation, the normal form, and the notion of lone conjunct from $\Qp[\nxt,\Diamondw]$ to $\Qp[\nxt,\Diamondw](\mathcal{ELI})$ in a straightforward way. To formulate conditions ${\bf (n1)}$--${\bf (n4)}$, we replace the set inclusions `$\rho_{i}\subseteq \rho_{j}$' by entailment `$\r_{i}\models \r_{j}$'\!. For example, ${\bf (n4)}$ becomes
\begin{description}
	\item[(n4$'$)] $\r_{0}^{i+1}\not\models\r_{k_{i}}^{i}$ if $i>0$, $\q_{i}$ is primitive and $R_{i+1}$ is $\le$.
\end{description}
The condition for lone conjuncts now requires that $\r$ is not equivalent to any $\q_{1}\wedge \q_{2}$ with $\mathcal{ELI}$-queries $\q_{1},\q_{2}$ such that $\q_{i}\not\models \r$ for $i=1,2$. Then one can show again that every $\Qp[\nxt,\Diamondw](\mathcal{ELI})$-query is equivalent to a query in
normal form, which can be computed in polynomial time.

\begin{theorem}\label{thm:nextdiamond2}
	The statements of Theorem~\ref{thm:nextdiamond} $(i)$--$(iv)$ also hold if one replaces $\Qp[\nxt,\Diamondw]$ by  $\Qp[\nxt,\Diamondw](\mathcal{ELI})$.
\end{theorem}

The proof generalises the example set defined in Theorem~\ref{thm:nextdiamond} using the frontiers provided by Theorem~\ref{thm:tenCate} as a \emph{black box}.
Indeed, in the definition of examples, replace $\rho_{i}$ by $\hat{\r}_{i}$, the data instance corresponding to the $\mathcal{ELI}$-query $\r_{i}$, and replace `$\rho\setminus \{A\}$ for $A\in \rho$' by `the data instance corresponding to a query in $\mathcal{F}(\r)$'. We choose a single individual, $a$, as the root of these data instances.
For example, item (a) becomes:
\begin{description}
\item[\rm (a$'$)] replacing some $\r^{i}_{j}$ by the data instance corresponding to a query in $\mathcal{F}(\r^{i}_{j})$ or removing the whole $\r^{i}_{j}=\emptyset$ for $i\not=0$ and $j\not=0$ from some $\q_{i}$.
\end{description}
Next, consider the class $\mathcal{Q}_{p}[\U](\mathcal{L})$ of queries of the form	
\begin{equation}\label{upatheli}
	\q = \r_0 \land (\el_1 \U (\r_1 \land ( \el_2 \U ( \dots (\el_n \U \r_n) \dots )))),
\end{equation}
where $\r_{i}$ is an $\mathcal{L}$-query and $\el_{i}$ is either an $\mathcal{L}$-query or $\bot$, for $\mathcal{L}\in \{\mathcal{EL},\mathcal{ELI}\}$. For the same reason as in the 1D case, we fix a finite signature $\Sigma$ of predicate symbols. Denote by $\mathcal{L}(\Sigma)$ and $\mathcal{Q}_{p}[\U](\mathcal{L})$ the set of queries in $\mathcal{L}$ and $\mathcal{Q}_{p}^{\Sigma}[\U](\mathcal{L})$, respectively,
with predicate symbols in $\Sigma$. Aiming to generalise Theorem~\ref{uno}, we again translate set-inclusion to entailment, so  the \emph{peerless queries} $\mathcal{P}^{\Sigma}[\U](\mathcal{L})$ take the form \eqref{upatheli} such that either $\el_{i}=\bot$ or $\el_{i}\not\models \r_{i}$ and $\r_{i}\not\models \el_{i}$.

\begin{theorem}\label{uno2}
	Let $\Sigma$ be a finite relational signature. Then $\mathcal{P}^{\Sigma}[\U](\mathcal{EL})$ is polynomially characterisable within $\mathcal{Q}_{p}^{\Sigma}[\U](\mathcal{EL})$, while $\mathcal{P}^{\Sigma}[\U](\mathcal{ELI})$ is exponentially, but not polynomially, characterisable within $\mathcal{Q}_{p}^\Sigma[\U](\mathcal{ELI})$.
\end{theorem}
To prove Theorem~\ref{uno2}, we generalise the example set from the proof of Theorem~\ref{uno}. The positive examples are straightforward: simply replace $\rho_{i}$ and $\lambda_{i}$ by the data instances corresponding to $\r_{i}$ and $\el_{i}$ (and choose a single root individual). For the negative examples, we have to generalise the construction of $\sigma$, $\sigma\setminus\{A\}$, and $\sigma\setminus\{A,B\}$. For $\sigma$, this is straightforward as its role can now be played by the $\Sigma$-data instance $\mathcal{A}_{\Sigma}=\{A(a),R(a,a)\mid A,R\in \Sigma\}$ for which $\mathcal{A}_{\Sigma}\models \q(a)$ for all $\q \in \mathcal{ELI}(\Sigma)$.
For $\sigma\setminus\{A\}$ and $\sigma\setminus\{A,B\}$, we require \emph{split partners} defined as follows. Let $Q$ be a finite set of $\mathcal{L}(\Sigma)$-queries. A set $\mathcal{S}(Q)$ of pointed $\Sigma$-data instances $(\mathcal{A},a)$ is called a \emph{split partner} of $Q$ in $\mathcal{L}(\Sigma)$ if the following conditions are equivalent for all $\mathcal{L}(\Sigma)$-queries $\q'$:
\begin{itemize}
	\item $\mathcal{A}\models \q'(a)$ for some $(\mathcal{A},a) \in \mathcal{S}(Q)$;
	\item $\q'\not\models \q$ for all $\q\in Q$.
\end{itemize}
\begin{example}\em
The split partner of $\{\q=A\}$ in $\mathcal{EL}(\Sigma)$ is the singleton set
containing $(\mathcal{A}_{\Sigma}^{-A},a)$ with $\mathcal{A}_{\Sigma}^{-A}$ defined as
$
\{B(a),R(a,b),R(b,b),B'(b)\mid B\in\Sigma\setminus\{A\},R,B'\in \Sigma\}.
$	
\end{example}

\begin{theorem}\label{thm:split}
Fix $n>0$. For any set $Q$ of $\mathcal{EL}(\Sigma)$-queries with $|Q|\leq n$, one can compute in polynomial time a split partner $\mathcal{S}(Q)$ of $Q$ in $\mathcal{EL}(\Sigma)$.	
For $\mathcal{ELI}$, one can compute a split partner in exponential time, which is optimal as even for singleton sets $Q$ of $\mathcal{ELI}(\Sigma)$-queries, no polynomial-size split partner of $Q$ in $\mathcal{ELI}(\Sigma)$ exists in general.
\end{theorem}

The proof, given in the Appendix, requires (as does $\mathcal{A}_{\Sigma}$) the construction of non-tree-shaped data instances. Our results for $\mathcal{ELI}$ are closely related to the study of generalised dualities for homomorphisms between relational structures~\cite{DBLP:journals/ejc/FoniokNT08,DBLP:journals/siamdm/NesetrilT05} but use pointed relational structures. The construction of $\mathcal{S}(Q)$ for $\mathcal{ELI}$ is based on a construction first introduced in~\cite{DBLP:journals/tods/BienvenuCLW14}.


We obtain the negative examples for $\q$ of the form~\eqref{upatheli} by taking the following pointed data instances $(\mathcal{D},a)$ (assuming that split partners take the form $(\mathcal{A},a)$ for a fixed $a$):
\begin{description}
\item[$(\mathfrak n_0')$] 
	$(\mathcal{A}_{\Sigma}^n,a)$ and $(\mathcal{A}_{\Sigma}^{n-i}\mathcal{A}\mathcal{A}_{\Sigma}^{i},a)$, for $(\mathcal{A},a)\in \mathcal{S}(\{\r_i\})$;
	
\item[$(\mathfrak n_1')$] $(\D,a)=(\hat{\r}_0\dots\hat{\r}_{i-1}\mathcal{A}\hat{\r}_i\dots\hat{\r}_n,a)$ with $\D,a,0\not\models\q$ and $(\mathcal{A},a)\in \mathcal{S}(\{\el_{i},\r_{i}\})\cup \mathcal{S}(\{\el_{i}\}) \cup \mathcal{S}(\{\r_{i}\})\cup \{(\mathcal{A}_{\Sigma},a)\}$;
	
	\item[$(\mathfrak n_2')$]
	for all $i$ and $(\mathcal{A},a) \in \mathcal{S}(\{\el_i\})\cup\{(\mathcal{A}_{\Sigma},a)\}$, \emph{some} data instance
	$$
	(\D^i_{\mathcal{A}},a) = (\hat{\r}_{0}\dots \hat{\r}_{i-1} \mathcal{A} \hat{\r}_{i}\hat{\el}_{i+1}^{k_{i+1}}\hat{\r}_{i+1} \dots \hat{\el}_{n}^{k_n}\r_n,a),
	$$
	if any, such that $\max(\D^i_{\mathcal{A}}) \le (n+1)^2$ and $\D^i_{\mathcal{A}},a,0\not\models \q^\dag$ for $\q^\dag$ obtained from $\q$ by replacing all $\el_j$, $j \le i$, with $\bot$.
\end{description}
We illustrate the construction by generalising Example~\ref{ex:second}.
\begin{example}\em
For $\q= \nxt A$ and any relational signature $\Sigma\ni A$, we obtain, after removing redundant instances, that $E^{+}=\{(\emptyset\{A(a)\},a)\}$ and $E^{-}=\{(\mathcal{A}_{\Sigma}\mathcal{A}_{\Sigma}^{-A}\{A(a)\},a)\}$
characterise $\q$ within $\mathcal{Q}_{p}^{\Sigma}[\U](\mathcal{EL})$.
\end{example}

We finally generalise Theorem~\ref{thm:branching1} $(ii)$ (part~$(i)$ is not interesting since simple queries do not generalise to any new class of $\mathcal{ELI}$-queries). Query classes such as $\mathcal{Q}[\Diamond](\mathcal{EL})$ are defined in the obvious way by replacing in $\mathcal{Q}[\Diamond]$-queries conjunctions of atoms by $\mathcal{EL}$-queries.
\begin{theorem}\label{thm:branching2}
	The class $\mathcal{Q}_{b}[\Diamond](\mathcal{EL})\cap \mathcal{Q}_{\leq n}[\Diamond](\mathcal{EL})$ is polynomially characterisable for any $n < \omega$.
\end{theorem}
Again the positive and negative examples are obtained from the 1D case by replacing $\sigma$ by $\mathcal{A}_{\Sigma}$ and $\sigma\setminus\{A\}$ by appropriate split partners.

\section{Applications to Learning}\label{Sec:learning}

We apply our results on unique characterisability to exact learnability of temporal instance queries. Given a class $\mathcal{Q}$ of such queries, we aim to
identify a \emph{target query} $\q\in \mathcal{Q}$ using queries to an oracle. The learner knows $\mathcal{Q}$ and the signature $\sigma$ ($\Sigma$ in the 2D case) of $\q$. We allow only one type of queries, called \emph{membership queries}, in  which the learner picks a $\sigma$-data instance $\D$ and asks the oracle whether $\mathcal{D}\models \q$ holds. (In the 2D case, the learner picks a pointed $\Sigma$-data instance $(\D,a)$ and asks whether $\mathcal{D},a,0\models \q$ holds.)
The oracle answers `yes' or `no' truthfully. The class $\mathcal{Q}$ is (\emph{polynomial time}) \emph{learnable with membership queries} if
there exists an algorithm that halts for any $\q\in\mathcal{Q}$ and computes
(in polynomial time in the size of $\q$ and $\sigma/\Sigma$),
using membership queries, a query $\q'\in\mathcal{Q}$ that is equivalent to
$\q$. By default, the learner does not know $|\q|$ in advance
but reflecting Theorem~\ref{thm:nextdiamond}~$(iii)$, we also consider the case when $|\q|$ is known (which is common in active learning).

Obviously, unique characterisability is a necessary condition for learnability with membership queries. Conversely, if there is an algorithm that computes, for every $\q\in \mathcal{Q}$, an example set that uniquely characterises $\q$ within $\mathcal{Q}^{\sig(\q)}$, then $\mathcal{Q}$ is learnable with membership queries: enumerate $\mathcal{Q}^{\sig(\q)}$ starting with the smallest query $\q$, compute a characterising set $E$ for $\q$ and check using membership queries whether $\q$ is equivalent to the target query. Eventually the algorithm will terminate with a query that is equivalent to the target query. As all of our positive results on unique characterisability provide algorithms computing example sets, we directly obtain learnability with membership queries. Moreover, if the example sets are computed in exponential time, then we obtain an exponential-time learning algorithm: in the enumeration above only $|\sig(\q)|^{|\q|}$ queries are checked before the target query is found. Unfortunately, we cannot infer polynomial-time learnability from polynomial characterisability in this way.

A detailed analysis of polynomial-time learnability using membership queries is beyond the scope of this paper. Instead, we focus on one main result, the polynomial-time learnability of
$\mathcal{Q}_{p}[\nxt,\Diamond](\mathcal{ELI})$.
\begin{theorem}\label{thm:learning}
$(i)$ The class of safe queries in $\Qp[\nxt, \Diamondw](\mathcal{ELI})$
is polynomial-time learnable with membership queries.

$(ii)$ The class $\Qp[\nxt, \Diamondw](\mathcal{ELI})$ is polynomial-time learnable with membership queries if the learner knows the size of the target query in advance.

$(iii)$ The class $\Qp[\nxt,\Diamond](\mathcal{ELI})$ is polynomially-time learnable with membership queries.
\end{theorem}
%
\begin{proofsketch}
We consider the 1D case without $\mathcal{ELI}$-queries first. $(i)$ Our proof strategy is to construct a query $\q'$ that agrees with
  $\q$ on the positive and negative examples for $\q'$ from
  Theorem~\ref{thm:nextdiamond}.
The algorithm proceeds by computing a data instance $\D$. Our aim is to
	arrive at $\D_b$ through iterations of steps, from which the required
	query can be  `read off'\!.
\begin{description}
\item[Step 1.] First, identify the number of $\nxt$ and $\Diamond$ in
		$\q$ by asking membership queries of the form $\sigma^k$ incrementally,  starting
		from $k=1$, and then set $b = \min\{k\mid\sigma^k\models \q\}+1$ and $\D_0 = \sigma^b$. Initialise $\D = \D_0$.
		
\item[Step 2.] Suppose that a data instance $\D'$ is obtained from $\D$
		by applying one of the rules $(a)$--$(e)$ of Theorem~\ref{thm:nextdiamond}.  If
		$\D'\models \q$ then replace $\D$ with $\D'$.
		Repeat as long as possible.
                One can show that the number of applications of each rule
                is bounded by a polynomial in $|\sigma|$ and the size of $\q$,
		and so \textbf{Step 2} finishes in polynomial time.


		

\item[Step 3.]
		Suppose $\D$ contains $\emptyset^b\rho^i_0\emptyset^b$ and
		$|\rho^i_0|\geq 2$. Since rule $(a)$ does not apply, every homomorphism
		\mbox{$h \colon \q\to \D$} sends some $t_1,\dots,t_l$ to $\rho^i_0$, for $l\geq 1$.
		As $\q$ does not contain lone conjuncts, $\q$ contains singleton primitive blocks at positions
		$t_1,\dots,t_l$. Suppose $\rho^i_0 = \{A_1,\dots, A_{|\rho^i_0|}\}$ and let $w = \{A_1\}\emptyset^b
		\{A_2\}\emptyset^b\dots \{A_{|\rho^i_0|}\}\emptyset^b$ (the order in which
		$A_1$,\dots, $A_{|\rho^i_0|}$, the elements of $\rho^i_0$,  are enumerated does
		not matter, we fix any one).
		Let $\D^i_k$ be obtained from $\D$ by replacing
		$\emptyset^b\rho^i_0\emptyset^b$ with $\emptyset^b(w)^k$.
		Notice that, for $k=|\q|$, we have $\D^i_k\models \q$; however, the
		algorithm is not given this $k$. Instead, the algorithm incrementally iterates
		starting from $k=1$ until $\D^i_k\models \q$. Since $k\leq |\q|$, this
		takes polynomially-many iterations.
		Let $\D'$ be obtained from $\D^i_k$ by removing primitive blocks as long as
		$\D'\models \q$. Notice that rules $(a)$--$(e)$ do not apply to $\D'$.
		Replace $\D$ with $\D'$. Repeat \textbf{Step 3} as long as possible.
		Since no new lone conjuncts are introduced, the process finishes
		after polynomially-many steps.

\item[Step 4.] At this point of computation, the algorithm has identified all
		blocks of $\q$ but not the sequences of $\Diamond$ and $\Diamondw$ between
		them. They can be easily determined based on the positive and negative examples $\mathcal{D}_{i}$ and $\mathcal{D}_{i}^{-}$.
%
	\end{description}
The proof of $(ii)$ is similar, with a modified \textbf{Step 3}. Finally, $(iii)$ is a consequence of $(ii)$ as the size of the query $\q$ does not exceed $n = |\sigma|b$.

We obtain a learning algorithm for $\Qp[\nxt, \Diamondw](\mathcal{ELI})$ by combining
the learning algorithm above with the learning algorithm for $\mathcal{ELI}$-queries by~\citeauthor{DBLP:conf/icdt/CateD21}~(\citeyear{DBLP:conf/icdt/CateD21}) using the positive and negative examples given in Theorem~\ref{thm:nextdiamond2}. Note that the data instance $\mathcal{A}_{\Sigma}$ is now used instead of $\sigma$ and that one has to `unfold' such non tree-shaped data instances into tree-shaped ones.
\end{proofsketch}

\section{Conclusions}

In this paper, we have considered temporal instance queries with \LTL{} operators and started investigating their unique (polynomial) characterisability and exact learnability using membership queries. We have obtained both positive and negative results, depending on the available temporal operators and the allowed interaction between the temporal and object dimensions in queries. The results indicate that finding complete classifications of 1D and 2D temporal queries according to (polynomial) characterisability and learnability could be a very difficult task. In particular, interesting open problems include the polynomial characterisability of full $\mathcal{Q}_{p}^{\sigma}[\U]$, more general criteria of polynomial characterisability for temporal branching queries and other temporal operators,  and the polynomial-time learnability of $\mathcal{Q}_{p}^{\sigma}[\U]$ and 2D extensions. From a conceptual viewpoint, it would be of interest to develop a framework that spells out explicitly the conditions that non-temporal queries should satisfy so that their combination with \LTL{}-queries preserves polynomial characterisability and polynomial-time learnability.


\section*{Acknowledgments}

This research was supported by the EPSRC UK grants EP/S032207 and EP/S032282 for the joint project `quant$^\text{MD}$: Ontology-Based Management for Many-Dimensional Quantitative Data'\!.

\bibliographystyle{kr}
\bibliography{bib,local}

\ifdefined\dontincludethispart\else
\appendix


\section{Proofs for Section~\ref{sec:chara}}

We prove the claims made in the introduction and Example~\ref{EX:lone}.

\medskip

(1) The query $\q=\Diamondw (A \wedge B)$ is not uniquely characterisable within $\Qp[\Diamondw]$. Indeed, consider the queries $\q_1 = \Diamondw (A \land \Diamondw B)$ and $\q_i = \Diamondw (A \land \Diamondw (B \land \Diamondw \q_{i-1}))$. Clearly, $\q \models \q_i$ and $\q_i \not\models \q$ for all $i \ge 1$. Suppose $\q$ fits $(E^+,E^-)$ and $n$ is the length of the longest example in $E^-$. Then $\q_{n+1}$ also fits $(E^+,E^-)$ as $\D,0 \not\models \Diamondw (A \land B)$, and so $\D,0 \not\models \q_{n+1}$, for any $\D \in E^-$.
	
(2) The query $\q = \bot \U A$ (i.e., $\nxt A$) is not uniquely characterisable within $\mathcal{Q}^{\{A\}}[\U]$. For suppose $\q$ fits $(E^+,E^-)$ and $n$ is the length of the longest example in $E^-$. Consider $\q' = (\nxt^{n+1} A) \U A$. Clearly, $\q' \not\models \q$ and $E \models \q'$ ($\D \not \models \q'$, for any $\D \in E^-$, because $\D,1 \not\models A$ and $\D,1 \not\models \nxt^{n+1} A$).

(3) While the query $\Diamondw (A \wedge B)$ is not characterisable,
there is a polynomial $f$ such that for all $n \in \mathbb{N}$, it is
characterisable within $\Qp^n[\nxt,\Diamondw]$ by some $E_n$ of size $\le f(n)$.
Take $E_n = (E^+,E^-)$ with
$E^+=\{(\{A,B\}), (\emptyset,\{A,B\})\}$ and
$E^-=\{(\underbrace{ \{A\},\{B\},\ldots,\{A\},\{B\} }_{n\text{ times}})\}$.
If $\q' \in \Qp^n[\nxt,\Diamondw]$ fits $E$ then we can assume without
loss of generality that it does not use $\nxt$, as $(\{A,B\}) \models \q'$.
This means that $\q'$ is of the form
$\rho_0 \land \Diamondw (\rho_1 \land \Diamondw (\rho_2 \land \cdots \land
(\Diamondw \rho_m) \cdots))$, with $m < n$.\nb{why? -Y}
Moreover, $\rho_0 = \emptyset$ as $(\emptyset,\{A,B\}) \models \q'$.
And since $(\{A,B\}) \models \q'$, we must have $\rho_i \subseteq \{A,B\}$
for all $i$.
Finally, as $(\underbrace{ \{A\},\{B\},\ldots,\{A\},\{B\} }_{n\text{ times}})
\not\models \q'$, there must be $i$ such that $\rho_i = \{A,B\}$.
Thus, $\q' \equiv \Diamondw(A \wedge B)$.

\section{Proofs for Section~\ref{Sec:weak-path}}
We show Theorem~\ref{thm:nextdiamond}. We use the fact that $\D \models \q$ iff there is a \emph{homomorphism} $h$ from the set $\var(\q)$ of variables in $\q$ to $[0,\len(\D)]$, i.e., $h(t_0) = 0$, $A(h(t)) \in \D$ if $A(t) \in \q$, $h(t') = h(t) +1$ if $\suc(t,t') \in \q$, and $h(t)\, R\, h(t')$ if $R(t,t') \in \q$ for $R \in \{<,\leq\}$.

\medskip
\noindent
{\bf Theorem~\ref{thm:nextdiamond}.}
{\em
	$(i)$ A query $\q\in \Qp[\nxt, \Diamondw]$ is uniquely characterisable within $\Qp[\nxt, \Diamondw]$ iff $\q$ is safe.
	
	$(ii)$ Those queries that are uniquely characterisable within $\Qp[\nxt, \Diamondw]$ are actually polynomially characterisable within $\Qp[\nxt, \Diamondw]$.
	
	$(iii)$ The class $\Qp[\nxt, \Diamondw]$ is polynomially characterisable for bounded query size.
	
	$(iv)$ The class $\Qp[\nxt,\Diamond]$ is polynomially characterisable.
}

\begin{proof}
	$(i)$ Assume that $\q$ is given. To show $(\Leftarrow)$, we may assume $\q$ in normal form~\eqref{fullq} does not contain lone conjuncts. Let $b$ be the number of $\nxt$ and $\Diamond$ in $\q$ plus 1.
	Consider the example set $E = (E^+, E^-)$ constructed in the main paper.
	Then $E$ is polynomial in $|\q|$ and $\D \models \q$ for all $\D \in E^+$. Using the condition that $\q$ is in normal form, we show the following:
	
	\smallskip
	\noindent
	\textbf{Claim 1.} $(i)$ \emph{There is only one homomorphism $h \colon \q \to \D_{b}$, and it maps isomorphically each $\q_i$ onto $\qw_i$.}
	
	$(ii)$ \emph{$\D^-_i \not\models \q$, for any $\mathcal{R}_i$ different from $\le$.}
	
	$(iii)$ \emph{If $\D_{b}'$ is obtained from $\D_b$ by replacing some $\qw_{i}$ with $\qw'_{i}$ such that $\qw'_{i},\ell \not\models \q_i$ for any $\ell \le \max (\qw'_i)$, then $\D_{b}' \not\models \q$, and so $\D \not\models \q$, for all $\D \in E^-$.}
	
	\smallskip
	\noindent
	\emph{Proof of claim.} $(i)$ As $\q$ is in normal form and the gaps between $\qw_i$ and $\qw_{i+1}$ are not shorter than any block in $\q$, each block $\q_i$ in $\q$ is mapped by $h$ to a single block $\qw_j$ of $D_{b}$. The function $f \colon [0,n] \rightarrow [0,n]$ defined by taking $f(i) = j$ is such that $f(0)=0$, $i<j$ implies $f(i)\leq f(j)$, and block $\q_{i}$ is satisfied in $\qw_{f(i)}$. It also follows from the definition of normal form that if $f(i)=i$, then $h$ isomorphically maps $\q_i$ onto $\qw_i$ and $f(i-1)<i$ and $f(i+1)>i$. To show that $f(i)=i$ for all $i$, we first observe that $f(1) \geq 1$ and $f(j)=j$, for $j = \max \{i\mid f(i) \ge i\}$, from which $f(j-1)<j$ and $f(j+1)>j$. Then we can proceed in the same way inductively by considering $f$ restricted to the smaller intervals $[j,n]$ and $[0,j]$.
	
	$(ii)$ Suppose $\mathcal{R}_i$ is not $\le$ but there is a homomorphism $h \colon \q \to \D^-_i$. Consider the location of $h(s^i_0) = \ell$. Suppose $\ell$ is in $\qw_i$. Since $\rho^{i+1}_{k_{i+1}} \ne \emptyset$ and by the construction of $\D^-_i$,   $h(s^{i+1}_{0})$ lies in some $\qw_j$ with $j > i+1$. But then there is a homomorphism $h'\colon \q \to \D_b$ different from the one in $(i)$, which is impossible. We arrive to the same contradiction if we assume that $\ell$ lies in $\qw_j$ with $j< i$ or $j >i$.
	
	$(iii)$ is proved analogously. \hfill \qed
	
	\smallskip
	
	Thus, $\q$ fits $E$. Suppose now $\q'$ is any $\Qp[\nxt,\Diamond, \Diamondw]$-query in normal form. For $t \in \var(\q')$, denote by $\tau_{t}$ the set of atoms $A$ with $A(t)\in \q'$ and call it the \emph{type} of $t$ in $\q'$. Similarly, for $\ell \in [0,\len(\D_b)]$, denote by  $\rho_{\ell}$ the set of atoms $A$ with $A(\ell) \in \D_{b}$ and call it the \emph{type} of $\ell$ in $\D_b$. A homomorphism $h \colon \q' \to \D_{b}$ is \emph{block surjective} if every point in every block $\qw_i$ of $\D_{b}$ is in the range $\textit{ran}(h)$ of $h$; it is \emph{type surjective} if $\rho_{\ell} = \bigcup_{h(t)=\ell} \tau_{t}$ for all $\ell \in \textit{ran}(h)$.
	The following claim follows immediately from the definitions:
	
	\smallskip
	\noindent
	\textbf{Claim 2.} $(i)$ \emph{If there is a homomorphism $h \colon \q' \to \D_{b}$ that is not block or type surjective, then $\D \models \q'$ for some $\D \in E^-$ obtained from $\D_b$ by $(a)$.}
	
	$(ii)$ \emph{If there exist a homomorphism $h \colon \q' \to \D_{b}$ and $(t<t')\in \q'$ or $(t\leq t') \in \q'$ such that $h(t)\ne h(t')$ and $h(t),h(t') \in \qw_i$, for some block $\qw_i$, then $\D \models \q'$ for some $\D \in E^-$ obtained from $\D_b$ by $(b)$.}

	\smallskip
	
	Suppose now that $h \colon \q' \to \D_{b}$ is a block and type surjective homomorphism, $(t\le t') \in \q'$ and $h(t) = h(t') = \ell$ lies in $\qw_{i}$.
	Then $h^{-1}(\ell) = \{t_{1},\dots,t_{k}\}$ with $k \ge 2$ and $(t_{j}\leq t_{j+1}) \in \q'$, $1 \le j < k$. By (n3) and (n4), $\tau_{t_j} \ne \emptyset$ for at least one $t_j$, and so $\rho_\ell \ne \emptyset$.
	Consider possible locations of $\ell$ in $\qw_i$.
	
	\emph{Case} 1: $\ell$ has both a left and a right neighbour in $\qw_{i}$. Then there is  $\D \in E^-$ obtained by (c)---i.e., by replacing the appropriate $\rho_l^i$ with $\rho_l^i \emptyset^b \rho_l^i$---and a homomorphism $h' \colon \q' \to \D$, which `coincides' with $h$ except that $h'(t_1)$ is the point with the first $\rho_l^i$ and  $h'(t_j)$, for $j = 2,\dots,k$, is the point with the second $\rho_l^i$.
	
	\emph{Case} 2: $\ell$ has no neighbours in $\qw_{i}$, so this block is primitive and $\rho_{\ell}$ is a singleton (as $\q$ has no lone conjuncts by our assumption). Then $t_{1}$ is the last variable in its block in $\q'$, $t_{k}$ is the first variable in its block in $\q'$, and the $t_i$ with $1 < i <k$, if any, are all primitive blocks. But then the types $\tau_{t_{i}}$ and $\tau_{t_{i+1}}$ are not comparable with respect to $\subseteq$, contrary to $\rho_{\ell}$ being a singleton. Thus, Case 2 cannot happen.
	
	\emph{Case} 3: $\ell$ has a left neighbour in $\qw_{i}$ but no right neighbour. As $h$ is type surjective and in view of (n3), $\tau_{t_1} \subsetneq \rho_\ell$. Let $A \in \rho_\ell \setminus \tau_{t_1}$ and let $\D \in E^-$ be obtained by the first part of  (d) by replacing $\rho^i_{k_i}$ with $\rho^i_{k_i} \setminus \{A\} \emptyset^b \rho^i_{k_i}$. Then there is a homomorphism $h' \colon \q' \to \D$ that sends $t_1$ to  the point with $\rho^i_{k_i} \setminus \{A\}$ and the remaining $t_j$ to the point with $\rho^i_{k_i}$.
	
	\emph{Case} 4: $\ell$ has a right neighbour in $\qw_i$, $i \ne 0$, but no left neighbour. This case is dual to Case 3 and we use the second part of (d).
	
	\emph{Case} 5: $\ell = 0$. If $\qw_0$ is primitive, then all of the $t_i$ are primitive blocks in $\q'$. By (n3), $\tau_{t_{2}}\not\subseteq \tau_{t_{1}}$; by type surjectivity,  $\tau_{t_{1}}\subsetneq \rho_{\ell}$, and so there is $A \in \rho_\ell \setminus \tau_{t_1}$. By the first part of (e), we have $\D \in E^-$ obtained by replacing $\rho^0_0$ with $\rho^0_0 \setminus \{A\} \emptyset^b \rho^0_0$. Then there is a homomorphism $h' \colon \q' \to \D$ that sends $t_1$ to the point with $\rho^0_0  \setminus \{A\}$ and the remaining $t_j$ to the point with $\rho^0_0$. Finally, if $\qw_0$ is not primitive, the second part of (e) gives $\D \in E^-$ by replacing $\rho_0^0$ in $\D_b$ with $\rho_0^0 \emptyset^b \rho_0^0$. We obtain a homomorphism from $\q'$ to $\D$ by sending $t_1$ to the first $\rho_0^0$ and the remaining $t_j$ to the second $\rho_0^0$.
	
	\smallskip
	
	It remains to consider the case when there is a homomorphism $h \colon \q' \to \D_b$ that is an isomorphism between the blocks in $\q'$ and the blocks in $\D_b$, and so the difference between $\q'$ and $\q$ can only be in the sequences of $\Diamond$ and  $\Diamondw$ between blocks. To be more precise, $\q$ is of the form~\eqref{fullq},
	\begin{align}
		\q' = \q_{0} \mathcal{R}'_{1} \q_{1} \dots \mathcal{R}'_{n} \q_{n}
	\end{align}
	and $\mathcal{R}_i \ne \mathcal{R}'_i$ for some $i$.
	Four cases are possible:
	\begin{description}
		\item[\rm $(i)$] $\mathcal{R}_i = (r_0 \le r_1)$ and $\mathcal{R}_i' = (s_0 < s_1) \dots (s_{l-1} < s_l)$, for $l \ge 1$. In this case, $\D_i \not \models \q'$, for $\D_i \in E^+$.
		
		\item[\rm $(ii)$] $\mathcal{R}_i = (r_0 < r_1) \dots (r_{k-1} < r_k)$, $\mathcal{R}_i' = (s_0 < s_1) \dots$ $(s_{l-1} < s_l)$, for $l > k$. Then again $\D_i \not \models \q'$.
		
		\item[\rm $(iii)$] $\mathcal{R}_i = (r_0 < r_1) \dots (r_{k-1} < r_k)$, $\mathcal{R}_i' = (s_0 \le s_1)$, for $k \ge 1$. In this case $\D_i^- \models \q'$, for $\D^-_i \in E^-$.
		
		
		\item[\rm $(iv)$] $\mathcal{R}_i = (r_0 < r_1) \dots (r_{k-1} < r_k)$ and $\mathcal{R}_i' = (s_0 < s_1) \dots$ $(s_{l-1} < s_l)$, for $l < k$. Then again $\D_i^- \models \q'$.
	\end{description}
	
	$(\Rightarrow)$ Suppose $\q$ in normal form~\eqref{fullq} does contain a lone conjunct $\q_i = \rho$. Let $\rho^-$ be the last type of the block $\q_{i-1}$ and let $\rho^+$ be the first type of the block $\q_{i+1}$. Then $\rho$ is a disjoint union of some nonempty $\tau$ and $\tau'$ such that at least one of the queries $\s'_1$ or $\s''_1$ below is in normal form:
	\begin{align*}
		& \s'_1 = \q_{0} \mathcal{R}_{1} \dots \mathcal{R}_{i} \tau (\le) \tau' \mathcal{R}_{i+1}  \dots \mathcal{R}_{n} \q_{n},\\
		& \s''_1 = \q_{0} \mathcal{R}_{1} \dots \mathcal{R}_{i} \tau (\le) \tau' (\le) \tau   \mathcal{R}_{i+1}  \dots \mathcal{R}_{n} \q_{n}
	\end{align*}
	For example, if $\rho^- = \{A,A'\}$, $\rho = \{A,B\}$, $\rho^+ = \{A,B'\}$ and $\mathcal{R}_{i}$ and $\mathcal{R}_{i+1}$ are both $\le$, we take $\tau = \{B\}$, $\tau' = \{A\}$, for which $\s'_1$ is not in normal form, while $\s''_1$ is. Pick one of $\s'_1$ and $\s''_1$, which is in normal form, and denote it by $\s_1$. For $n \ge 2$, let $\s_n$ be the query obtained from $\s_1$ by duplicating $n$ times the part $\tau (\le) \tau'$ in $\s_1$ and inserting $\le$ between the copies. It is readily seen that $\s_n$ is in normal form.
	Clearly, $\q \models \s_n$ and, similarly to the proof of Claim~1, one can show that $\s_n \not\models \q$, for any $n \ge 1$.
	
	Now suppose $E = (E^+,E^-)$ characterises $\q$ and let $n = \max \{\len(\D) \mid \D \in E^-\} + 1$. Then $E \models \s_n$, which is impossible. Indeed, consider any $\D \in E^-$. To show that $\D \not\models \s_n$, suppose otherwise. Then there is a homomorphism $h \colon \s_n \to \D$. By the pigeonhole principle, $h$ maps some variables of types $\tau$ and $\tau'$ in $\s_n$ to the same point in $\D$. But then  $h$ can be readily modified to obtain a homomorphism $h' \colon \q \to \D$, contrary to $E^- \not\models \q$.
	
	$(ii)$ follows from the proof of $(i)$ as $(E^{+},E^{-})$ is of polynomial size.
	
	$(iii)$ We aim to characterize $\q$ in normal form~\eqref{fullq} which may contain lone conjuncts within the class of queries in $\mathcal{Q}_{p}[\nxt,\Diamond]$ of size at most $n$, where $n$ is the size of $\q$. The set $E^{+}$ of positive examples is defined as before
	and we extend the set of rules (a) to (e) in the definition of $E^{-}$ as follows: if $\q_{i}= \rho(s)$ with $\rho=\{A_{1},\ldots,A_{k}\}$ is a block in $\mathcal{D}_{b}$ with $\rho$ a lone conjunct in $\q$, then
	\begin{description}
		\item[\rm (f)] replace $\rho$ with $(\rho\setminus \{A_{1}\}\emptyset^{b}\cdots \emptyset^{b} \rho\setminus \{A_{k}\})^{n}$.
	\end{description}
	For the proof that $(E^{+},E^{-})$ characterizes $\q$ within the class of queries of size at most $n$, observe that with the exception of \emph{Case 2} the proof of $(i)$ still goes through. In Case 2, however, we can now apply the assumption that the size of $\q'$ is bounded by $n$ as we then obtain a data instance $\mathcal{D}\in E^{-}$ and a homomorphism $h':\q'\rightarrow \mathcal{D}$.
	
	$(iv)$ Assume $\q$ in $\mathcal{Q}_{p}[\nxt,\Diamond]$ is given. The proof of $(i)$  shows that $(E^{+},E^{-})$, defined in the same way as in $(i)$ except
	that the rules (c),(d), and (e) are not used to construct $E^{-}$,
	characterizes $\q$ within $\mathcal{Q}_{p}[\nxt,\Diamond]$ even if $\q$ contains lone conjuncts.
\end{proof}


\section{Proofs for Section~\ref{Sec:until}}

\medskip
\noindent
{\bf Theorem~\ref{separUpath}.}
{\em Any $\Qp^\sigma[\U]$-queries $\q \not\equiv \q'$ can be separated by some $\D$ with $\max(\D) \le O((\min \{\dep(\q), \dep(\q')\})^2)$.}

\begin{proof}
Let $\q = \rho_0 \lambda_1^* \rho_1 \dots \rho_n\emptyset^*$ and $\q' = \tau_0 \mu_1^* \tau_1  \dots \tau_k\emptyset^*$. If $n < k$, then $\rho_0\rho_1\dots \rho_n$ separates $\q$ from $\q'$. Suppose $n = k$. If $\rho_i \subsetneq \tau_i$, for some $i \le n$, then again $\rho_0\rho_1\dots \rho_n$ separates $\q$ from $\q'$. So suppose $\rho_i=\tau_i$ for all $i\leq n$.

Let $\q \not\models \q'$. Then there is $\D = \rho_0\lambda_1^{k_1} \rho_1 \dots \lambda_n^{k_n}\rho_n\emptyset^{k_{n+1}}$ separating $\q$ from $\q'$. We show that  $k_i \le n + 1$, for all $i \le n$. To see this, we convert $\mathfrak A_{\q'}$ to a DFA $\mathfrak B_{\q'}$ using the subset construction and observe that whenever there are transitions $Q_1 \to_{\lambda_i} \dots \to_{\lambda_i} Q_{n+1} \to_{\lambda_i} Q_{n+2}$ in $\mathfrak B_{\q'}$ (with $Q_j \subseteq [0,n+1]$), then $Q_{n+1} =  Q_{n+2}$ because, by the structure of $\mathfrak A_{\q'}$, we have $\delta'_{\alpha^{n+1}} = \delta'_{\alpha^{n+2}}$, for any $\alpha$, where $\delta'_w$ is the transition function on the states of $\mathfrak B_{\q'}$ corresponding to the word $w$.
\end{proof}

\noindent
{\bf Lemma~\ref{matching}.}
{\em For any queries $\q$ and $\q'$ as above, either $(i)$ each $\lambda_i \ne \bot$ subsumes $\mu_j$ occurring in some matching pair $(\lambda_k,\mu_j)$ or $(ii)$ $\q$ and $\q'$ are separated by a data instance of the form $\D^i_\q$ 
 or $\smash{\D^j_{\q'}}$. Also, if $\q$ is peerless, $\lambda_i$ can only subsume $\mu_j$ in the matching pair $(\lambda_i,\mu_j)$ with $i \ge j$, in which case
$
\mu_j = \rho_j =  \dots = \rho_{i-1} = \lambda_i.
$}

\begin{proof}
If $\lambda_i \ne \bot$ does not subsume any $\mu_j$, then, as we know, $\D^i_\q \models \q$ and $\D^i_\q \not \models \q'$. So suppose $\lambda_i$ subsumes some $\mu_j$. Then either $(\lambda_i, \mu_j)$ is a matching pair or $\mu_j \subsetneq \lambda_i$. Note that the latter is impossible if $\q$ is peerless. If $\mu_j$ does not subsume any $\lambda_l$, then $\D^j_{\q'} \models \q'$ and $\D^j_{\q'} \not \models \q$. Otherwise, we consider $\lambda_l$ subsumed by $\mu_j$, etc. Since $\lambda_l \subsetneq \lambda_i$, sooner or later this process will terminate.
\end{proof}

\begin{lemma}\label{singlelambda}
Suppose that $\q = \rho_0 \lambda_1^* \rho_1 \lambda_2^* \dots \lambda_{n}^*\rho_n \emptyset^*$ and $\q' = \rho_0 \rho_1 \lambda_2^* \dots \lambda_{n}^*\rho_n \emptyset^*$ with $\q \not\models \q'$. Then there is a data instance $\D$ of the form $\rho_0\lambda_1\rho_1 \lambda_2^{k_2}\dots \lambda_n^{k_n}\rho_n$ such that $\D \models \q$ and $\D \not\models \q'$.
\end{lemma}
\begin{proof}
Take any $\D = \rho_0\lambda_1^{k_1} \rho_1 \lambda_2^{k_2}\dots \lambda_n^{k_n}\rho_n$ with $\D \models \q$ and $\D \not\models \q'$. Let $\D_i$ be $\D$ with  $k_1 = i$. Choose $\D_i \not\models \q'$ with $\D_l \models \q'$, for all $l<i$. Consider some $\D'=\rho_0\rho_1 u$ with $\D' \Subset \D_{i-1}$ and set $\D''=\rho_0\lambda_1\rho_1 u$. If $\D'' \models \q'$, then $\D_i \models \q'$ as  $\D''\Subset \D_i$, which is impossible. Thus, $\D''$ is the data instance we need.
\end{proof}

\noindent
{\bf Theorem~\ref{uno}.}
{\em $\Qpi[\U]$ is polynomially characterisable within $\Qp^\sigma[\U]$.}

\begin{proof}
We show that any $\q = \rho_0 \lambda_1^* \rho_1 \lambda_2^* \dots \lambda_n^* \rho_n\emptyset^* $ in $\Qpi[\U]$ is characterised by $E = (E^+,E^-)$, where $E^+$ contains all data instances of the following forms:
\begin{description}
\item[$(\mathfrak p_0)$] $\rho_0\dots\rho_n$,

\item[$(\mathfrak p_1)$] $\rho_0\dots\rho_{i-1}\lambda_i\rho_i \dots \rho_n = \D^i_{\q}$,

\item[$(\mathfrak p_2)$] 
$\rho_0 \dots \rho_{i-1} \lambda_{i}^k \rho_{i} \dots \rho_{j-1} \lambda_{j}\rho_j\dots\rho_n=\D^j_{i,k}$, for $i < j$ and $k=1,2$;

\end{description}
and $E^-$ has all instances that are \emph{not} in $\boldsymbol{L}(\q)$ of the forms:
\begin{description}
\item[$(\mathfrak n_0)$] 
$\sigma^{n}$ and $\sigma^{n-i}\sigma \setminus \{A\} \sigma^{i}$, for $A\in\rho_i$,

\item[$(\mathfrak n_1)$] $\rho_0\dots\rho_{i-1}\sigma \setminus \{A,B\}\rho_i\dots\rho_n$, for $A\in \lambda_i\cup\{\bot\},B\in\rho_i\cup\{\bot\}$,

\item[$(\mathfrak n_2)$]
for all $i$ and $A\in\lambda_i\cup\{\bot\}$, \emph{some} data instance
\begin{equation}\label{raznost1}
\D^i_{\!A} = \rho_0 \dots \rho_{i-1} (\sigma\setminus\{A\})\rho_i\lambda_{i+1}^{k_{i+1}} \dots \lambda_{n}^{k_n}\rho_n,
%
\end{equation}
if any, such that $\max(\D^i_{\!A}) \le (n+1)^2$
%
%
%
and $\D^i_{\!A}\not\models \q^\dag$ for $\q^\dag$ obtained from $\q$ by replacing $\lambda_j$, for all $j \le i$, with $\bot$.

Note that $\D^i_{\!A} \not\models \q$ for peerless $\q$.
\end{description}
By definition, $\q$ fits $E$ and $|E|$ is polynomial in $|\q|$. Suppose $\q' = \tau_0 \mu_1^* \tau_1  \dots \mu_m^* \tau_m \emptyset^*$ also fits $E$.
By $(\mathfrak p_0)$ and $(\mathfrak n_0)$, we have $n=m$ and $\rho_i = \tau_i$, for $i \le n$.
%
Consider the maximal $i$ with $\lambda_i\ne\mu_i$ (if there is no such, $\q \equiv \q'$).
%
%
\begin{description}
\item[\rm\emph{\rm\it Case} 1:]  $\mu_i\ne\bot$,  $\mu_i\not\subseteq\lambda_i$. By Lemma~\ref{matching}, if $\mu_i$ does not subsume any $\lambda_j$, then $\mathfrak n_1$ separates $\q$ and $\q'$.
So suppose $\mu_i$ subsumes $\lambda_j$. As $\q$ is peerless, $j>i$ and we have $\rho_i\subseteq\mu_i$.
%
%
By the minimality of $\q'$, there is $\D$ such that $\D\models\rho_{i-1} \mu_i^*\rho_i \dots \mu_n^* \rho_n\emptyset^*$ but $\D\not\models\rho_{i-1} \rho_i \dots \mu_n^* \rho_n\emptyset^*$. 
By Lemma \ref{singlelambda}, we can choose $\D = \rho_{i-1}\mu_i\rho_i u$. Consider $\D'=\rho_{i-1}(\sigma\setminus\{A\})\rho_iu$, where $A\in\lambda_i\setminus\mu_i$ if $\lambda_i\ne\bot$, or $A=\bot$ otherwise.  Since $\rho_i\subseteq\mu_i\subseteq\sigma\setminus\{A\}$ and $\lambda_i\not\subseteq\sigma\setminus\{A\}$, we have $\rho_0\dots\rho_{i-2}\D'\models\rho_0 \dots \rho_{i-1} (\sigma\setminus\{A\})^*\rho_i\lambda_{i+1}^* \dots \lambda_{n}^*\rho_n$ but $\D'\not\models \rho_0\dots\rho_{i-1}\rho_i\lambda_{i+1}^* \dots \lambda_{n}^*\rho_n\emptyset^*$
, for  otherwise $\D\models\rho_{i-1} \rho_i \dots \mu_n^* \rho_n$ 
 as $\lambda_j=\mu_j$ for all $j>i$. Therefore, there is $\D^i_A\in\mathfrak n_2$ with $\D^i_A\models\q'$.

\item[\rm\emph{\rm\it Case} 2:] $\lambda_i \ne \bot$, $\mu_i\subsetneq\lambda_i$. By Lemma~\ref{matching}, if $\mu_i$ does not subsume any $\lambda_j$, then $\mathfrak n_1$ separates $\q$ and $\q'$.
So suppose $\mu_i$ subsumes $\lambda_j$. As $\q$ is peerless, $j>i$ and we have $\rho_i\subseteq\mu_i$. But then $\rho_i\subseteq\lambda_i$, which is a contradiction.

\item[\rm\emph{\rm\it Case} 3:] $\lambda_i\ne\bot$, $\mu_i=\bot$. Find the maximal  $j<i$ such that $\lambda_i$ subsumes $\mu_j$. Then $\mu_j=\rho_j=\ldots=\rho_{i-1}=\lambda_i$ by Lemma~\ref{matching}.

Suppose $\lambda_{j'}=\bot$ for all $j'\in [j+1,i-1]$. By Lemma~\ref{matching}, $\mu_{j'}=\bot$ for all $j'\in [j+1,i-1]$. For if $\mu_{j'}\ne \bot$, then it either does not subsume anything, or it subsumes $\lambda_i$. In the latter case, we have $\rho_{j'-1}\subsetneq\mu_{j'}$ and either there is $j''\in(j,j')$ with $\rho_{j''}\not\subset\mu_{j''}$, in which case $\mu_{j''}$ does not subsume anything, or we violate the minimality condition. The queries $\rho_{j-1}\mu_j^*\rho_j\ldots\rho_{i-1}\rho_i\ldots$ and $\rho_{j-1}\rho_j\ldots\rho_{i-1}\mu_j^*\rho_i\ldots$ are obviously equivalent, so we can change this part of $\q'$ and look for the next place where the queries differ.

Now assume that there is $\lambda_{j'}\ne\bot$ with $j'\in[j+1,i-1]$.  Suppose there is no data instance in $(\mathfrak p_1)$, $(\mathfrak p_2)$ or $(\mathfrak n_1)$ separating $\q$ and $\q'$. Then $\mu_{j'}=\lambda_{j'}$. If $j'>j+1$, consider $\D^i_{j',1}$. Suppose $\D^i_{j',1}\models\q'$.  Then there is $\D'=\rho_0\dots\rho_{m-1}\mu_m\rho_m\dots\rho_{l-1}\mu_l\rho_l\dots\rho_n\Subset D^i_{j',1}$.
By analysing all possible orderings of $i,j',m,l$  we see that it is only possible when $l=j'-1$ and $\mu_{j'-1}\subseteq\lambda_{j'}$. Then we have $\lambda_{j'-1}=\mu_{j'-1}$. After that, we consider $\D^i_{j'-1,1}$ and so on. We now have $\mu_{j+1}\subseteq\mu_{j+2}\subseteq\dots\subseteq\mu_{j'}= \lambda_{j'}$.

Consider $\D^i_{j+1,1}$. If $\D^i_{j+1,1}\models\q'$, then, in view of $\mu_j\not\subseteq\lambda_{j+1}$, we have $\rho_{j-1}\subseteq\lambda_{j+1}$.

Considering $\D^i_{j+1,2}$, we see that either $\rho_{j-2}\subseteq\lambda_{j+1}$ and $\mu_{j-1}\subseteq\rho_j$ or $\rho_{j-2}\subseteq\rho_{j}$ and $\mu_{j-1}\subseteq\lambda_{j+1}$ with $\mu_{j-1}$ and $\rho_{j-2}$ incomparable, since otherwise we will violate the peerlessness of $\q$.

If $\rho_{j-2}\subseteq\lambda_{j+1}$ and $\mu_{j-1}\subseteq\rho_j$, then there is $j_1<j-1$ such that $\mu_{j_1}\subseteq\rho_{j_1}\subseteq\ldots\subseteq\rho_{j-1}\subseteq\lambda_{j+1}$. This $\mu_{j_1}$ is paired with some $\lambda_l$ for $l>j$. By peerlessness, we have $l=j$ and $\mu_j=\rho_{j_1}=\ldots=\rho_{j-1}=\lambda_{j}\subseteq\lambda_{j+1}$. It follows that $\lambda_{j-1}=\mu_{j-1}$ and we can repeat a previous argument and move further.

If $\rho_{j-2}\subseteq\rho_{j}$ and $\mu_{j-1}\subseteq\lambda_{j+1}$, then we have $\lambda_j=\rho_{j-1}=\mu_{j-1}$ and either $\mu_{j-2}\subseteq\rho_{j-1}$ or $\rho_{j-3}\subseteq\rho_{j-1}$.

Suppose $\mu_{j-2}\subseteq\rho_{j-1}$. Then $\rho_{j-3}\subseteq\rho_{j-2}$ and there is $j_2\le j-3$ such that $\mu_{j_2}\subseteq \rho_{j_2}\subseteq\ldots\subseteq\rho_{j-2}$. We know that $\mu_{j_2}$ subsumes some $\lambda_l$; by peerlessness it can only be $\lambda_{j-1}$ with $\mu_{j_2}=\rho_{j_2}=\ldots=\rho_{j-2}=\lambda_{j-1}$, and $\rho_{j-2}\not\subseteq\mu_{j-2}$. So we repeat a previous argument and move further.

Suppose  $\rho_{j-3}\subseteq\rho_{j-1}$. In this case we have either $\mu_{j-3}\subseteq\rho_{j-2}$ or $\rho_{j-4}\subseteq\rho_{j-2}$ and we move further.

Either way we cannot stop moving further and the process cannot terminate. Therefore,  there is a data instance from $(\mathfrak p_1)$, $(\mathfrak p_2)$ or $(\mathfrak n_1)$ separating $\q$ and $\q'$.
\end{description}
This completes the proof of the theorem.
\end{proof}


\section{Proofs for Section~\ref{sec:branching}}
We start by giving the proof of Lemma~\ref{lem:conjof}.

\medskip
\noindent
{\bf Lemma~\ref{lem:conjof}.}
{\em
	For every $\q\in \mathcal{Q}[\Diamond]$ one can compute
	in polynomial time an equivalent query of the form $\q_{1}\wedge \cdots \wedge\q_{n}$ with $\q_{i}\in \mathcal{Q}_{p}[\Diamond]$ for $i \leq n$.
}
\begin{proof}
	It is sufficient to observe that a query $\q$ of the form
	$$
	\rho_{0} \wedge \Diamond(\rho_{1} \wedge \bigwedge_{i=1}^{n}\Diamond\q_{i})
	$$
	with $\q_{1},\ldots,\q_{n}\in \mathcal{Q}[\Diamond]$ is equivalent to
	$$
	\q'=\rho_{0} \wedge \bigwedge_{i=1}^{n} \Diamond(\rho_{1} \wedge \Diamond \q_{i})
	$$
	To see this, observe that $\q\models\q'$ is trivial. For the converse direction
	consider any data instance $\mathcal{D}$ with $\mathcal{D},0\models \q'$. Then take the minimum $\ell_{0}$ of all $\ell>0$ such that $\mathcal{D},\ell \models \rho_{1} \wedge \Diamond \q_{i}$. We have $\mathcal{D},\ell_{0} \models \rho_{1} \wedge \Diamond \q_{i}$ for all $i\leq n$ and so  $\mathcal{D},0\models \q$.
\end{proof}

\bigskip
\noindent
{\bf Details for Example~\ref{thm:cc}.}
Recall that
$$
\q_{1}= \emptyset (\s,\sigma)^{n} \s, \quad \q_{2} = \emptyset\sigma^{2n+1}
$$
where
$$
\s=\{A_{1},A_{2}\}\{B_{1},B_{2}\}
$$
and consider the set $P$ of queries of the form
$$
\emptyset\s_{1}\cdots \s_{n+1}
$$
with $\s_{i}$ either $\{A_{1}\}\{A_{2}\}$ or $\{B_{1}\}\{B_{2}\}$. $P$ contains $2^{n+1}$ queries.
\begin{lemma}\label{lemma:helpp}
	$\q_{1}\wedge \q_{2}\not\models \q$ for any $\q\in P$. For any data instance
	$\mathcal{D}$ with $\mathcal{D}\models \q_{1}\wedge \q_{2}$ there is at most one $\q\in P$ with $\mathcal{D}\not\models \q$.
\end{lemma}
\begin{proof}
	We first construct for every $\q =\emptyset\s_{1},\ldots\s_{n+1}\in P$ a data instance $\mathcal{D}$ with $\mathcal{D}\models \q_{1}\wedge \q_{2}$ and $\mathcal{D}\not\models\q$.
	The data instance $\mathcal{D}_{\q}$ is defined by taking the data instance
	$\overline{\q}_{1}=\rho_{0},\ldots,\rho_{3n+2}$ and
	replacing $\rho_{i}$ by $\sigma$ in it if, for some $j\leq n$:
	\begin{itemize}
		\item $i=3j+1$ and $\s_{j+1}=\{A_{1}\}\{A_{2}\}$; or
		\item $i=3j+2$ and $\s_{j+1}=\{B_{1}\}\{B_{2}\}$.
	\end{itemize}
	\begin{figure}
		\begin{tikzpicture}[every node/.append style={align=center,inner sep=2pt},node distance=0.2cm and 0.1cm,xscale=0.83]
			\begin{scope}[font=\scriptsize]
				\node (0) at (0,0) {$0$} ;
				\node (1) at (1,0) {$1$} ;
				\node (2) at (2,0) {$2$} ;
				\node (3) at (3,0) {$3$} ;
				\node (d1) at (3.7,0) {$\cdots$} ;
				\node (j1) at (4.5,0) {$3j+1$} ;
				\node (j2) at (5.5,0) {$3j+2$} ;
				\node (j3) at (6.55,0) {$3j+3$} ;
				\node (d2) at (7.2,0) {$\cdots$} ;
				\node (n1) at (8,0) {$3n+1$} ;
				\node (n2) at (9,0) {$3n+2$} ;
			\end{scope}
			
			\node[below=of 0] (0') {$\emptyset$} ;
			\node[below=of 1] (1') {$A_1$\\$A_2$\\$\textcolor{red}{B_1}$\\$\textcolor{red}{B_2}$} ;
			\node[below=of 2] (2') {$B_1$\\$B_2$} ;
			\node[below=of 3] (3') {$A_1$\\$A_2$\\$B_1$\\$B_2$} ;
			\node[below=of d1,yshift=-0.2cm] {$\cdots$} ;
			\node[below=of j1] (j1')  {$A_1$\\$A_2$\\$\textcolor{red}{B_1}$\\$\textcolor{red}{B_2}$} ;
			\node[below=of j2] (j2') {$B_1$\\$B_2$} ;
			\node[below=of j3] (j3') {$A_1$\\$A_2$\\$B_1$\\$B_2$} ;
			\node[below=of d2,yshift=-0.2cm] (d2') {$\cdots$} ;
			\node[below=of n1] (n1') {$A_1$\\$A_2$\\$\textcolor{red}{B_1}$\\$\textcolor{red}{B_2}$} ;
			\node[below=of n2] (n2') {$B_1$\\$B_2$} ;
			
			\begin{pgfonlayer}{background}
				\fill[gray!25,rounded corners=4pt] (-0.2,0.2) rectangle (9.5,-2.2) ;
			\end{pgfonlayer}
			
			\node at (4.5,0.5) {$\mathcal{D}_\q$} ;

			\begin{scope}[yshift=-4cm,font=\scriptsize]
				\node (0q) at (0,0) {$0$} ;
				\node (1q) at (1,0) {$1$} ;
				\node (2q) at (2,0) {$2$};
				\node (d1q) at (2.75,0) {$\cdots$} ;
				\node (j1q) at (3.5,0) {$2j+1$} ;
				\node (j2q) at (4.5,0) {$2j+2$} ;
				\node (d2q) at (5.25,0) {$\cdots$} ;
				\node (n1q) at (6,0) {$2n+1$} ;
				\node (n2q) at (7,0) {$2n+2$} ;
			\end{scope}
			
			\node[above=of 0q] (0q') {$\emptyset$} ;
			\node[above=of 1q] (1q') {$A_1$} ;
			\node[above=of 2q] (2q') {$A_2$} ;
			\node[above=of d1q] (d1q') {$\cdots$} ;
			\node[above=of j1q] (j1q') {$A_1$} ;
			\node[above=of j2q] (j2q') {$A_2$} ;
			\node[above=of d2q] (d2q') {$\cdots$} ;
			\node[above=of n1q] (n1q') {$A_1$} ;
			\node[above=of n2q] (n2q') {$A_2$} ;
			
			\begin{pgfonlayer}{background}
				\fill[gray!25,rounded corners=3pt] ($(0q')+(-0.2,0.4)$) rectangle ($(n2q)+(0.5,-0.3)$);
			\end{pgfonlayer}
			\node[below=0.1cm of j1q] {$\overline{\q}$} ;
			
			\path[-stealth]
			(0q') edge node[left,pos=0.3] {$f_\q$} (0')
			(1q') edge (1')
			(2q'.north) edge (3'.south)
			(j1q'.north) edge (j1'.south)
			(j2q'.north) edge (j3'.south)
			(n1q'.north) edge (n1'.south) ;
			
		\end{tikzpicture}
		\caption{Definition of $\mathcal{D}_{\q}$ (additions to $\overline{\q}_{1}$ are indicated in red) and $f_\q$ for $\q = \emptyset(\{A_1\}\{A_2\})^{n+1}$.}
		\label{fig:example7}
	\end{figure}
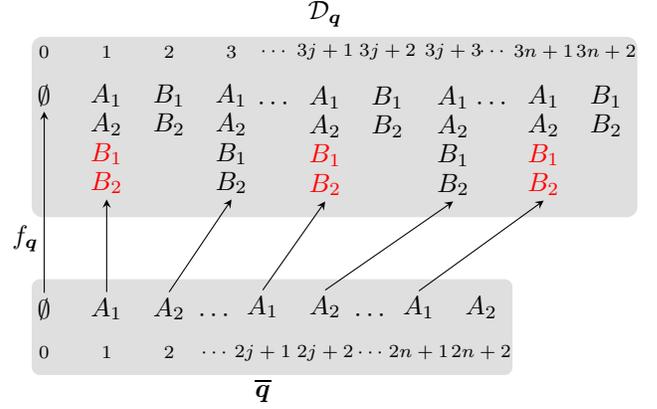
	The data instance $\mathcal{D}_{\q}$ with $\q = \emptyset(\{A_1\}\{A_2\})^{n+1}$ is depicted in Figure~\ref{fig:example7}. It also illustrates that $\mathcal{D}_{\q}\not\models \q$.
	
	We have $\mathcal{D}_{\q}\models \q_{1} \wedge \q_{2}$. We next show that $\mathcal{D}_{\q}\not\models\q$ for all $\q\in P$. We introduce a method to prove this that regards the query $\q$ as a CQ. Namely, for $\q=\rho_{0}\ldots\rho_{n}$ and  $\mathcal{D}=\delta_{0}\ldots\delta_{m}$ we have $\mathcal{D}\models \q$ iff the following inductive definition of a partial assignment $f$ from $[0,n]$ to $[0,m]$ has domain $[0,n]$:
	\begin{itemize}
		\item $f(0)=0$ if $\rho_{0} \subseteq \delta_{0}$, otherwise $f(0)$ is undefined;
		\item if $f(i)$ is defined then let $f(i+1)$ be the minimal $j>f(i)$ such that $j\leq m$ and $\rho_{i+1}\subseteq \delta_{j}$, if such a $j$ exists. Otherwise $f(i+1)$ is undefined.
	\end{itemize}
	If the domain of $f$ does not equal $[0,n]$, then the definition either fails at $i_{0}=0$ or at some induction step $i_{0}=i+1$. In either case we say that \emph{the definition of a satisfying assignment fails at $i_{0}$.}
	
	We return to the proof that $\mathcal{D}_{\q}\not\models\q$. We show that the definition of an assignment fails at $2n+2$. In fact, the partial assignment $f_{\q}$ in $\mathcal{D}_{\q}$ is defined as follows
	\begin{itemize}
		\item $f_{\q}(0) = 0$,
		\item $f_{\q}(2i+1)= 3i + 1$ if $\s_{i+1}=\{A_{1}\}\{A_{2}\}$, for $0\leq i\leq n$;
		\item $f_{\q}(2i+1)= 3i + 2$ if $\s_{i+1}=\{B_{1}\}\{B_{2}\}$, for $0\leq i \leq n$;
		\item $f_{\q}(2i+2) = 3i + 3$, for $0\leq i<n$,
	\end{itemize}
	and $f_{\q}(2n+2)$ remains undefined.
	
	Next we show that there does not exist a database $\mathcal{D}$ with $\mathcal{D}\models \q_{1}\wedge \q_{2}$ such that $\mathcal{D}\not\models \q$ and $\mathcal{D}\not\models \q'$ for distinct $\q,\q'$ in $P$. Assume for a proof by contradiction that $\mathcal{D}=\delta_{0}\ldots\delta_{m}$ is such a data instance.
	
	Let $\mathcal{D}_{\q}=\delta_{0}'\ldots\delta_{3n+2}'$ and $\mathcal{D}_{\q'}=\delta_{0}''\ldots\delta_{3n+2}''$ be the data instances defined above by extending $\rho_{0}\ldots\rho_{3n+2}$ with $\sigma$ depending on $\q$ and $\q'$, respectively.
	
	As $\mathcal{D}\models \q_{1}$ there exists a bijective function $g$ from some subset
	$A\subseteq [1,m]$ to $[1,3n+2]$ satisfying the following
	\begin{itemize}
		\item $i<j$ iff $g(i)<g(j)$;
		\item $\delta_{i}\supseteq \rho_{g(i)}$ for all $i\in A$.
	\end{itemize}
	As $\mathcal{D}\models \q_{2}$ and $\q\not=\q'$,
	\begin{itemize}
		\item there exists $i\in [1,m]\setminus A$ such that $\delta_{i}=\sigma$; or
		\item there exists $i\in A$ such that $\delta_{i}=\sigma$ and $\delta_{g(i)}'\not=\sigma$; or
		\item there exists $i\in A$ such that $\delta_{i}=\sigma$ and $\delta_{g(i)}''\not=\sigma$.
	\end{itemize}
	We may thus assume w.l.o.g. that $\mathcal{D}$ is obtained from $\rho_{0}\ldots\rho_{3n+2}$ by (1) inserting $\sigma$ after some $\rho_{i}$ with $i>0$ or (2) by replacing some $\rho_{i}$ with $i>0$ and $\delta_{i}'\not=\sigma$ by $\sigma$ or (3) by replacing some $\rho_{i}$ with $i>0$ and $\delta_{i}''\not=\sigma$ by $\sigma$.
	
	Case~2 and Case~3 are equivalent, so we only only consider Case~1 and Case~2.
	
	\medskip
	
	Case 1. One can show that then $\mathcal{D}\models \q''$ for all $\q''=\emptyset\s_{1}''\ldots\s_{n+1}''\in P$. We do this for $\q$.
	
	Assume first that $\sigma$ is inserted in $\rho_{0}\ldots\rho_{m}$ directly before timepoint 1. We define an assignment $f_{\q}^{\ast}$ by mapping $1$ to the new node and
	(we assume that the original nodes are still numbered as before and not shifted):
	\begin{itemize}
		\item $f_{\q}^{\ast}(0) = 0$,
		\item $f_{\q}^{\ast}(2i+2)= 3i+1$ if $\s_{i+1}=\{A_{1}\}\{A_{2}\}$, for $0\leq i\leq n$;
		\item $f_{q}^{\ast}(2i+2)= 3i + 2$ if $\s_{i+1}=\{B_{1}\}\{B_{2}\}$, for $0\leq i \leq n$;
		\item $f_{\q}^{\ast}(2(i+1)+1) = 3i + 3$, for $0\leq i<n$.
	\end{itemize}
	Hence $\mathcal{D}\models \q$. If $\sigma$ is inserted later, the argument is similar. In this case the assignment is defined by first taking $f_{\q}$
	as defined above and then from the point where $\sigma$ is inserted $f_{\q}^{\ast}$.
	
	Case 2. This is proved similarly to Case 1.
\end{proof}

We next provide further details of the proof of Theorem~\ref{thm:branching1}.
We first provide the missing proof of Lemma~\ref{lem:negexamples}.

\medskip
\noindent
{\bf Lemma~\ref{lem:negexamples}.}
{\em $(i)$ For any $\mathcal{D}\in E^{-}_{\q,m}$, we have $\mathcal{D}\not\models\q$.
	
	$(ii)$ For any $\q'\in \mathcal{Q}^{\sigma}[\Diamond]$ with $\q'\not\models \q$ and $\dep(\q')\leq m$,  there exists $\mathcal{D}\in E^{-}_{\q,m}$ with $\mathcal{D}\models \q'$.
}

\begin{proof}
	$\mathcal{D}\not\models\q$ for all $\mathcal{D}\in E^{-}_{\q,m}$ holds by definition and Lemma~\ref{lem:positive}.
	
	Now assume $\q'\in \mathcal{Q}^{\sigma}[\Diamond]$ with $\q'\not\models \q$ and $\dep(\q')\leq m$. $\q'$ is equivalent to a conjunction
	$\q_{1}'\wedge \cdots \wedge \q_{k}'$ with $\q_{i}'\in \mathcal{Q}_{p}[\Diamond]$ and
	$$
	\q_{i}' = \tau_{0}^{i} \land \Diamond (\tau_1^{i} \land \Diamond (\tau_2^{i} \land \ldots \land \Diamond \tau_{\ell_{i}}^{i}) )
	$$
	If there exists $A\in\rho$ with $A$ not in any $\tau_{0}^{i}$, then $\sigma^{A}\sigma^{m}\models \q'$, as required. Otherwise, $\q'\not\models \bigwedge_{i=1}^{n}\q_{i}^{-}$. Pick an $i$ with $\q'\not\models\q_{i}^{-}$.
	Then $\q_{j}'\not\models\q_{i}^{-}$ for all $j$. Hence $\mathcal{D}_{\q_{i}^{-},m}\models \q_{j}'$ for all $j$, and we obtain $\mathcal{D}_{\q_{i}^{-},m}\models \q'$.
\end{proof}
We come to the construction of positive examples required for the proof of Part~1 of Theorem~\ref{thm:branching1}. For
$$
\q = A_{0} \land \Diamond (A_1 \land \Diamond (A_2 \land \dots \land \Diamond A_{n}) )
$$
we call any $\q_{|k}:=A_{1},\ldots,A_{k}$ with $k\leq n$ a \emph{prefix} of $\q$ (note that $A_{0}$ is not taken into account).
\begin{lemma}\label{th:words}
	Let $w=A_{1},\ldots,A_{k}$ be a sequence of atomic concepts in a signature $\sigma$ and assume $k<n$. Then one can construct in polynomial time a $\sigma$-data instance $\mathcal{D}_{w,n}$ such that for all simple
	$\q\in \mathcal{Q}_{p}^{\sigma}[\Diamond]$ with $\dep(\q)= n$:
	$\q_{|k}=w$ iff $\mathcal{D}_{w,n}\not\models\q$.
\end{lemma}
\begin{proof}
	Assume $w$ and $n$ are given. Then $\mathcal{D}_{w,n}$ defined as
	$$
	\sigma(\sigma\setminus \{A_{1}\})^{n}\sigma(\sigma\setminus \{A_{2}\})^{n}\ldots(\sigma\setminus\{A_{k}\})^{n}\sigma^{n-k-1}
	$$
	is as required.
\end{proof}
Assume $\q=\q_{1}\wedge \cdots \wedge \q_{m}$ is balanced, simple, $n=\dep(\q)$, and
$$
\q_{i}= A_{0}^{i} \land \Diamond (A_1^{i} \land \Diamond (A_2^{i} \land \dots \land \Diamond A_{n}^{i}) )
$$
with $A_{0}=A_{0}^{1}=\cdots=A_{0}^{m}$. Then let $E^{+}$ contain the $\sigma$-data instances $$
\{A_{0}\}\sigma^{n}, \quad
\{A_{0}^{1}\}\ldots\{A_{n}^{1}\}\ldots \{A_{0}^{m}\}\ldots\{A_{n}^{m}\},
$$
and
$\mathcal{D}_{w,n}$, for $w\in I$, where
$I$ is the set of all $A_{1}^{i},\ldots,A_{k}^{i},A$ such that
$A\in \sigma$, $1\leq i\leq m$, and $A_{1}^{i},\ldots,A_{k}^{i},A$ is not a prefix of any  $\q_{j}$. Clearly $\mathcal{D}\models \q$ for all $\mathcal{D}\in E^{+}$.
\begin{lemma}
	$(E^{+},E^{-})$ characterises $\q$ within the class of balanced queries in $\mathcal{Q}_{b}[\Diamond]$.
\end{lemma}
\begin{proof}	
	First observe that
	$$
	\{A_{0}^{1}\}\ldots\{A_{n}^{1}\}\ldots \{A_{0}^{m}\}\ldots\{A_{n}^{m}\}\not\models \q'
	$$
	for any $\q'$ that uses symbols not in $\sigma$ or that is not simple. It remains to show that if $\q\not\models\q'$ and $\q'\models \q$, where $\q'=\q_{1}'\wedge \cdots \wedge \q_{m'}$ is simple, balanced, and uses symbols in $\sigma$ only, then
	there exists $\mathcal{D}\in E^{+}$ with $\mathcal{D}\not\models \q'$.
	Recall that $n=\dep(\q)$. If $\dep(\q')>n$, then $\{A_{0}\}\sigma^{n}\not\models \q'$ and we are done. As $\q'\models \q$, we then have $\dep(\q')= n$.
	Then there exists $\q_{i}'$ with $\q\not\models\q_{i}'$. If $\q_{i}$ does not start with $A_{0}$, then $\{A_{0}\}\sigma^{n}\not\models \q_{i}'$ and we are done. Otherwise let $k$ be maximal such that there exists $\q_{j}$ with ${\q_{j}}_{|k}={\q_{i}'}_{|k}$. Then $k<n$
	and $w={\q'_{i}}_{|k+1}\in I$. Hence $\mathcal{D}_{w,n}\not\models \q_{i}'$.
\end{proof}
We finally show that the positive examples defined in the main paper are as required for Part 2 of Theorem~\ref{thm:branching1}. Assume $\q= \q_{1}\wedge \cdots \wedge \q_{m}\in \mathcal{Q}^{\sigma}_{b}[\Diamond]\cap \mathcal{Q}^{\sigma}_{\leq n}[\Diamond]$ and that $\q\not\models \q'$ with $\q'\in \mathcal{Q}_{b}[\Diamond]\cap \mathcal{Q}_{\leq n}[\Diamond]$. Assume $\q'=\q_{1}'\wedge \cdots \wedge \q_{m}'$. If $\q'$ uses a symbol
not in $\sigma$ or $\dep(\q')>N$, then $\rho\sigma^{N}\not\models \q'$ and we are done. If the initial
conjunct of some $\q_{i}'$ contains an $A\not\in\rho$, then $\rho\sigma^{N}\not\models \q'$, and we are done.
If $\dep(\q')<N$, then $\q'\not\models \q$, and so there exists $\mathcal{D}\in E^{-}$ with $\mathcal{D}\models \q'$. Hence we may assume that $\dep(\q')=N$.
Take a $j$ with $\q\not\models \q_{j}'$ and assume that
$$
\q_{j}'= \tau_{0}\wedge \Diamond (\tau_{1} \land \Diamond (\tau_{2} \land \dots \land \Diamond \tau_N) )
$$
Define $f \colon \{1,\ldots,m\}\rightarrow \{1,\ldots,N\}$ by taking for every
$i$ a $\rho_{j}^{i}$ such that $\tau_{j}\not\subseteq \rho_{j}^{i}$ and setting $f(i)=j$. Such $j$ exist since $\q\not\models \q_{j}'$ and since we excluded any other reason for non-entailment already.
Then $\mathcal{D}_{f}\not\models \q_{j}'$, as required.

\section{Proofs for Section~\ref{sec:instquery}}
We begin by introducing some notation.
A \emph{pointed} atemporal $\Sigma$-data instance is a pair $(\mathcal{A},a)$ with $\mathcal{A}$ an atemporal $\Sigma$-data instance and $a\in \text{ind}(\mathcal{A})$. We associate with $\mathcal{A}$ the undirected graph
$$
G_{\mathcal{A}}=(\text{ind}(\mathcal{A}),\bigcup_{P\in \sigma}\{\{a,b\} \mid P(a,b)\in \mathcal{A}\})
$$
and call $\mathcal{A}$ \emph{acyclic} if $G_{\mathcal{A}}$ is acyclic and $P(a,b)\in \mathcal{A}$ implies $Q(a,b)\not\in \mathcal{A}$ for any $Q\not=P$ and $Q(b,a)\not\in \mathcal{A}$ for any $Q$. $\mathcal{A}$ is \emph{connected} if $G_{\mathcal{A}}$ is connected. We sometimes call acyclic and connected data instances \emph{tree-shaped}.
We assume the standard representation of an $\mathcal{ELI}$-query $\q$ as a set of atoms of the form $A(x)$, $P(x,y)$ with $x,y$ variables and a \emph{distinguished variable}, also called the \emph{answer variable} of $\q$. For instance, $\r = B \land \exists P. \exists P^-. A$ is represented as $\{B(x), P(x,y), P(z,y), A(z) \}$ with the distinguished variable $x$.

Every $\mathcal{ELI}$-query $\q$ defines a pointed data instance $\hat{\q} = (\q,x)$ (with $\ind(\q)$ being the variables), where $\q$ is tree shaped. Conversely, every pointed database $(\mathcal{A},a)$ with tree-shaped $\mathcal{A}$ defines an $\mathcal{ELI}$-query.

Let $\mathcal{A}$ and $\mathcal{B}$ be data instances. We call a mapping $h$ from $\text{ind}(\mathcal{A})$ to $\text{ind}(\mathcal{B})$ a \emph{homomorphism} from $\mathcal{A}$ to $\mathcal{B}$, in symbols $h: \mathcal{A} \rightarrow \mathcal{B}$, if
\begin{itemize}
	\item $A(a) \in \mathcal{A}$ implies $A(h(a)) \in \mathcal{B}$;
	\item $P(a,b) \in \mathcal{A}$ implies $P(h(a),h(b))\in \mathcal{B}$.
\end{itemize}
We call $h$ a \emph{homomorphism} from pointed $(\mathcal{A},a)$ to pointed $(\mathcal{B},b)$ if it is a homomorphism from $\mathcal{A}$ to $\mathcal{B}$ and $h(a)=b$.
We write $(\mathcal{A},a) \rightarrow (\mathcal{B},b)$ if there exists a homomorphism from $(\mathcal{A},a)$ to $(\mathcal{B},b)$.
\begin{lemma}
	For all $\mathcal{ELI}$-queries $\q_{1}$ and $\q_{2}$, we have $\q_{1}\models \q_{2}$ iff
	$\hat{\q}_{2} \rightarrow \hat{\q}_{2}$.
\end{lemma}
We call a pointed data instance $(\mathcal{A},a)$ \emph{core} if every homomorphism $h:(\mathcal{A},a)\rightarrow (\mathcal{A},a)$ is an isomorphism. Pointed structures $(\mathcal{A},a)$ and $(\mathcal{B},b)$ are
\emph{homomorphically equivalent} if $(\mathcal{A},a)\rightarrow (\mathcal{B},b)$ and
$(\mathcal{B},b) \rightarrow (\mathcal{A},a)$.
\begin{theorem}\label{thm:core}
	For every tree-shaped $(\mathcal{A},a)$, one can construct in polynomial time a core that is tree-shaped and homomorphically equivalent to $(\mathcal{A},a)$.
\end{theorem}
We have defined frontiers within the set of $\mathcal{ELI}$-queries partially ordered by entailment.
It is sometimes more convenient to define frontiers on the class of tree-shaped data instance partially ordered
by homomorphisms. A set $\mathcal{F}$ of tree-shaped $(\mathcal{A}',a')$ is called a \emph{frontier} of tree-shaped $(\mathcal{A},a)$ if
\begin{itemize}
	\item $(\mathcal{A}',a')\rightarrow (\mathcal{A},a)$ for all $(\mathcal{A}',a') \in \mathcal{F}$;
	\item $(\mathcal{A},a)\not\rightarrow (\mathcal{A}',a')$ for all $(\mathcal{A}',a')\in \mathcal{F}$;
	\item if $(\mathcal{B},b)\rightarrow (\mathcal{A},a)$, then either
	$(\mathcal{A},a)\rightarrow (\mathcal{B},b)$ or there exists $(\mathcal{A}',a')\in \mathcal{F}$ with $(\mathcal{B},b)\rightarrow (\mathcal{A}',a')$.
\end{itemize}
The frontier of $(\mathcal{A},a)$ is denoted by $\mathcal{F}(\mathcal{A},a)$.

\subsection{Proof of Theorem~\ref{thm:tenCate0}}

We describe the construction of the frontier $\mathcal{F}(\mathcal{A},a)$ for an tree-shaped pointed data instance $\mathcal{A},a$. (This construction is a minor adaptation of the construction in~\cite{DBLP:conf/icdt/CateD21}, where the proof of its correctness is given.) In this case, we can view $(\mathcal{A},a)$ as a labelled directed tree with the set of nodes $\ind(\mathcal{A})$ labelled with (possibly empty) sets of concept names $\{A_1, \dots, A_n\}$, edges $(b,c)$ labelled with either $P$ or $P^-$, for a role name $P \in \Sigma$, and rooted at $a$. First, we construct \emph{pre-frontier}, which is a set of structures $\mathcal{P}(\mathcal{A},a)$ defined inductively as follows:
\begin{itemize}
  \item if $a$ is a leave of $\mathcal{A}$ and $a$ has no concept name labels, then $\mathcal{P}(\mathcal{A},a) = \emptyset$;
  \item if $a$ is a leave of $\mathcal{A}$ labelled with $\{A_1, \dots, A_m\}$, then $\mathcal{P}(\mathcal{A},a) = \{ (\mathcal{B}^1, a), \dots, (\mathcal{B}^m, a) \}$, where $\mathcal{B}^i$ is\\
\begin{tikzpicture}
\node[label=above:{$\{A_1, \dots, A_{i-1}, A_{i+1}, A_m\}$}] {$a$};
\end{tikzpicture}\\
  \item otherwise, let $\mathcal A$ be\\
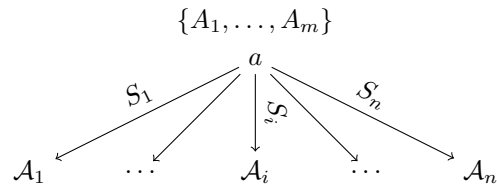
  \begin{tikzpicture}[grow=down, sloped]
\node[label=above:{$\{A_1, \dots, A_m\}$}] {$a$}
    child {
        node {$\mathcal A_1$}
            edge from parent[->]
            node[above] {$S_1$}
    }
    child[->] {node {$\dots$}}
    child {
        node {$\mathcal A_i$}
            edge from parent[->]
            node[above] {$S_i$}
    }
    child[->] {node {$\dots$}}
    child {
        node {$\mathcal A_n$}
            edge from parent[->]
            node[above] {$S_n$}
    };
\end{tikzpicture}
\end{itemize}
where each $\mathcal{A}_i$ has a root $a_i$ and let $\mathcal{P}(\mathcal{A}_i,a_i) = \{ (\mathcal{B}_i^1, a_i), \dots, (\mathcal{B}_i^k, a_i) \}$, for $k \geq 0$. We construct another set $\{(\mathcal{C}_i^j, c_i^j)\}$ of pointed structures, where the domain of $\mathcal{C}_i^j$ is equal to the set of sequences $bj$ where $b$ is in the domain of $\mathcal{B}_i^j$, and $S(bj, cj)$ is in $\mathcal{C}_i^j$ iff $S(b, c)$ is in $\mathcal{B}_i^j$ (and similarly for concept assertions $A(b)$). We set $c_i^j = a_ij$. We note that every element $d$ in the domain of $\mathcal{B}_i^j$ or $\mathcal{C}_i^j$ has the form $a j_1 \dots j_\ell$, for $\ell \geq 0$ and $a$ from the domain of $\mathcal{A}$. We write $\textit{orig}(d) = a$ to indicate that $a$ is the \emph{original} of $d$. Then, $\mathcal{P}(\mathcal{A},a)$ is defined as $\{(\mathcal{D}_1, a), \dots, (\mathcal{D}_n, a)\} \cup \{(\mathcal{E}_1, a), \dots, (\mathcal{E}_m, a)\}$, where $\mathcal{D}_i$ is the structure:\\
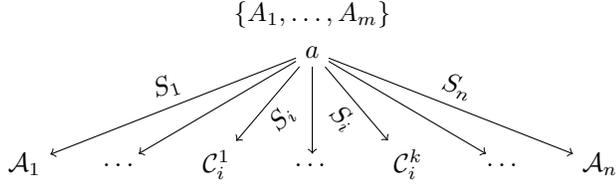
  \begin{tikzpicture}[grow=down, sloped, xscale = .85]
\node[label=above:{$\{A_1, \dots, A_m\}$}] {$a$}
    child {
        node {$\mathcal A_1$}
            edge from parent[->]
            node[above] {$S_1$}
    }
    child[->] {node {$\dots$}}
    child {
        node {$\mathcal C_i^1$}
            edge from parent[->]
            node[below] {$S_i$}
    }
    child[->] {node {$\dots$}}
    child {
        node {$\mathcal C_i^k$}
            edge from parent[->]
            node[below] {$S_i$}
    }
    child[->] {node {$\dots$}}
    child {
        node {$\mathcal A_n$}
            edge from parent[->]
            node[above] {$S_n$}
    };
\end{tikzpicture}\\
and each subtree $\mathcal{C}_i^j$ is rooted at $c_i^j$, and $\mathcal{E}_i$ is\\
  \begin{tikzpicture}[grow=down, sloped]
\node[label=above:{$\{A_1, \dots, A_{i-1}, A_{i+1}, \dots, A_m\}$}] {$a$}
    child {
        node {$\mathcal A_1$}
            edge from parent[->]
            node[above] {$S_1$}
    }
    child[->] {node {$\dots$}}
    child {
        node {$\mathcal A_i$}
            edge from parent[->]
            node[above] {$S_i$}
    }
    child[->] {node {$\dots$}}
    child {
        node {$\mathcal A_n$}
            edge from parent[->]
            node[above] {$S_n$}
    };
\end{tikzpicture}

Finally, the frontier $\mathcal{F}(\mathcal{A},a)$ is $\{(\mathcal{F}_1, a), \dots, (\mathcal{F}_n, a)\}$, where each $(\mathcal{F}_i, a)$ is obtained from $(\mathcal{D}_i, a) \in \mathcal{F}(\mathcal{A},a)$ in the following way. Let the domain of $\mathcal{D}_i$ be $\{a, d_1, \dots, d_k\}$. We construct a $j$-th copy $\mathcal{A}^j$ of $\mathcal{A}$, where domain of the structure $\mathcal{A}^j$ is the set of $b\;{-j}$, for $b$ in the domain of $\mathcal A$ and $S(b\;{-j}, c\;{-j})$ is in $\mathcal{A}^j$ iff $S(b, c)$ is in $\mathcal{A}$ (and similarly for concept assertions $A(b)$; note that we use negative numbers). Then, $\mathcal{F}_i$ is the union of $\mathcal{D}_i$, $\mathcal{A}^1$, \dots, $\mathcal{A}^k$ (note that they are all pairwise disjoint), where we add $S(b\;{-j}, d_j)$ to $\mathcal{F}_i$ iff $S(b,c)$ is in $\mathcal A$ and $\textit{orig}(d_j) = c$.


\subsection{Proof of Theorem~\ref{thm:firstone} $(i)$}

 Let $\dep(\q)$ be the maximum number of nested temporal operators in $\q$  (assuming that no subquery of $\q$ starting with a temporal operator is equivalent to $\top$). Let $\rdep(\q)$ be the length $n$ of the longest sequence $\exists S_1 \dots \exists S_n$ such that $\exists S_{i+1}$ is in the scope of $\exists S_i$ but not in the scope of a temporal operator in the scope of $\exists S_i$. The set of $\mathcal{Q}[\nxt,\Diamond]\otimes \mathcal{EL}$-queries $\q$ with $\dep(q) \le d$ and $\rdep(\q) \le r$, for some fixed $d,r < \omega$, contains finitely-many, say $N_{dr}$-many, non-equivalent queries. It is easy to construct data instances $\D^d$ and $\D^{dr}$ such that
	\begin{itemize}
		\item[--] $\D^d,a,0 \models \q$ iff $\dep(\q) \le d$, for all $\q$;
		
		\item[--] $\D^{dr},a,0 \models \q$ iff $\rdep(\q) \le r$, for all $\q$ with $\dep(\q) =  d$.
	\end{itemize}
	It is also not hard to show that any pair of nonequivalent queries $\q, \q' \in \mathcal{Q}[\nxt,\Diamond]\otimes \mathcal{EL}$ with $\dep(\q) = \dep(\q')$ and $\rdep(\q) = \rdep(\q')$ is distinguished by a data instance $\D$ with $\max(\D) \leq N_{dr}$ and $|\ind(\D)| \leq N_{dr}$.
	
	Now, given a $\q$ with $\dep(\q) = d$ and $\rdep(\q) = r$, we construct an example set $(E,a,0)$ with
	\begin{align*}
		E^+ \!= \{\D \mid \D,a, 0 \models \q, \max( \D) \leq N_{dr}, |\ind(\D)| \leq N_{dr}\},\\
		E^- \!= \{\D \mid \D,a, 0 \not \models \q, \max( \D) \leq N_{dr}, |\ind(\D)| \leq N_{dr}\}.
	\end{align*}
	%
	%
	It follows from the observations above that this example set uniquely characterises $\q$.

\subsection{Proof of Theorem~\ref{thm:firstone} $(ii)$}

We observe that the queries $\q \in \mathcal{Q}[\nxt,\Diamond]\otimes \mathcal{ELI}$ can be viewed as the formulas $\q(x)$ of the monodic fragment of temporal FO~\cite{DBLP:journals/apal/HodkinsonWZ00} (see also~\cite{DBLP:conf/epia/Schild93}). In particular, the queries $\q \in \Qp[\Diamond] \otimes \mathcal{ELI}$ naturally correspond to the formulas of the form:
\begin{align*}
\q_0&(x_0) =  \exists x_1,\dots, x_s \\
&\big(\bigwedge_{j \in I_1} A_j(x_{i_j}) \land \bigwedge_{k,\ell \in I_2} P_{k,\ell}(x_{i_k}, x_{i_\ell}) \land
\bigwedge_{m \in I_3} \Diamond \q_m(x_{i_m})\big),
\end{align*}
where $s \geq 0$, $0 \leq i_j, i_k, i_\ell, i_m \leq s$, $\q_m$ are of the same form as $\q_0$, and $x_{i_{m}} \neq x_{i_{m'}}$ for all $m, m' \in I_3$ with $m \neq m'$ (i.e., we do not allow more than one $\Diamond$-subformula for with the same free variable). Moreover, the set of atoms $\r_0 = \{A_j(x_{i_j} \mid j \in I_1 \} \cup \{ P_{k,\ell}(x_{i_k}, x_{i_\ell}) \mid k,\ell \in I_2\}$ is an $\mathcal{ELI}$-query with a distinguished variable $x_0$ (i.e., $\hat{\r}$ it is tree-shaped pointed data instance).

Thus, the $\q \in \Qp[\Diamond] \otimes \mathcal{ELI}$ queries $\q$ can be represented as trees:
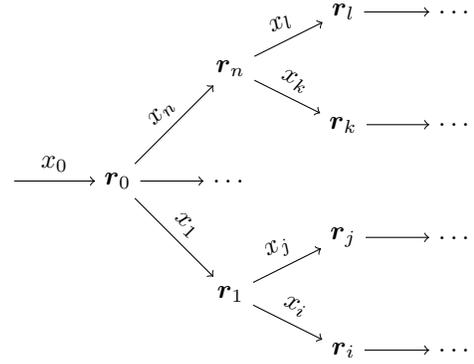
\begin{center}
\begin{tikzpicture}[grow=right, sloped]
\node {} child
{
node {$\r_0$}
    child {
        node {$\r_1$}
            child {
                node {$\r_{i}$} child[->] { node {$\dots$}}
                edge from parent[->]
                node[above] {$x_{i}$}
            }
            child {
                node {$\r_{j}$} child[->] { node {$\dots$}}
                edge from parent[->]
                node[above] {$x_{j}$}
            }
            edge from parent[->]
            node[above] {$x_{1}$}
    }
    child[->] { node {$\dots$}
    }
    child {
        node {$\r_n$}
            child {
                node {$\r_{k}$} child[->] { node {$\dots$}}
                edge from parent[->]
                node[above] {$x_{k}$}
            }
            child {
                node {$\r_{l}$} child[->] { node {$\dots$}}
                edge from parent[->]
                node[above] {$x_{l}$}
            }
            edge from parent[->]
            node[above] {$x_{n}$}
    }
    edge from parent[->]
    node[above] {$x_{0}$}
    };

\end{tikzpicture}
\end{center}
where the nodes are labelled with $\mathcal{ELI}$-queries $\r_i$ and the edges between the nodes are labelled with (query) variables, constructed as follows. The nodes of this tree, for a query $\q_0$, are from the set $\{\q_0\} \cup \{\q_i \mid \Diamond \q_i \text{ is a subformula of } \q_0\}$. The label of each $\q_i$ is the $\mathcal{ELI}$-query $\r_i$ constructed from the nontemporal atoms of $\q_i$; see $\r_0$ for $\q_0$ above. (We always omit the nodes on pictures as above; they are clear from the context.) The set of edges is defined as $\{(\q_i, \q_j)\mid \Diamond \q_j \text{ occurs in } \q_i \text{ not in scope of a }\Diamond\}$. The label of the edge $(\q_i, \q_j)$ is the (unique) free variable of $\q_j$. By the construction of $\q_i$, we note that there are no outgoing edges of a node with the same label.  For example, the query $\q = \Diamond (C \land \Diamond B) \land A \land \exists P. \exists P^-.\Diamond B$ is represented as follows:
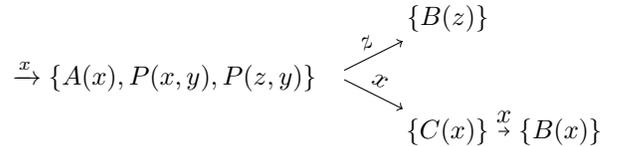
\begin{center}
\begin{tikzpicture}[grow=right, sloped]
\node[label=left:{$\xrightarrow{x}\{A(x), P(x,y), P(z,y)\}$} ] {}
    child {
        node {$\{ C(x)\}$}
            child {
                node {$\{B(x)\}$}
                edge from parent[->]
                node[above] {$x$}
            }
            edge from parent[->]
            node[above] {$x$}
    }
    child {
        node {$\{B(z)\}$}
            edge from parent[->]
            node[above] {$z$}
    };
\end{tikzpicture}
\end{center}

For every $\q \in \Qp[\Diamond] \otimes \mathcal{ELI}$ (in what follows we do not distinguish between $\q$ and its tree representation) and every node $\r_i$ in it, we define $\textit{dep}(\r_i)$ to be equal to the length of the path between the root $\r$ and $\r_i$. With such a $\q$, we associate the data instance $\D^{\q} = \D_0 \dots \D_n$ with $n = \max \{\textit{dep}(\r_i) \mid \r_i \text{ in }\q\}$, $\ind(\D^{\q})$ equal to the set of variables of $\q$, and $\D_i = \{A(x) \mid A(x) \in \r_i \text{ for some }\r_i \text{ in }\q \} \cup \{P(x,y) \mid P(x,y) \in \r_i \text{ for some }\r_i \text{ in }\q \}$. We also set $a^{\q} = x$, for the distinguished variable $x$ of $\q$. On the other hand, we associate with $\q$ the atemporal pointed database $(\mathcal A^{\q}, (x,0))$, where $x$ is the distinguished variable of $\q$, with
$\ind(\mathcal A^{\q}) = \{(y, m) \mid y \text{ occurs in some }\r_i \text{ in }\q \text{ and }\textit{dep}(\r_i) = m \}$, $A((y,m)) \in \mathcal A^{\q}$ iff $A(y) \in \r_i$ and $\textit{dep}(\r_i) = m$, and $P((y,m),(z,m)) \in \mathcal A^{\q}$ iff $P(y,z) \in \r_i$ and $\textit{dep}(\r_i) = m$. Moreover, for a fresh role name $T$, we add $T((y,m), (y,n))$ to $\mathcal A^{\q}$ for all $(y,m), (y,n) \in \ind(\mathcal A^{\q})$, such that $n = m+1$. We observe that $\mathcal A^{\q}$ is tree-shaped. Also, every connected (taking into account $T$-atoms) substructure of $\mathcal A^{\q}$ induces a query $\q' \in \Qp[\Diamond] \otimes \mathcal{ELI}$.
Denote by $\textit{cl}(\mathcal A^{\q})$ the structure that extends $\mathcal A^{\q}$ by adding the $T$-transitive closure.
\begin{lemma}\label{th:eli-diamond-1}
Let $\q, \q' \in \Qp[\Diamond] \otimes \mathcal{ELI}$ and assume that $\q$ (respectively, $\q'$) has a distinguished variable $x$ ($y$). Then $\q \models \q'$ iff $\D^{\q}, a^{\q} \models \q'$ iff $(\textit{cl}(\mathcal A^{\q}), x) \to (\textit{cl}(\mathcal A^{\q'}), y)$.
\end{lemma}

The following result shows the important property of the core of $\mathcal A^{\q}$ (note that $(i)$ does not follow from Theorem~\ref{thm:tenCate0}):
\begin{theorem}\label{th:eli-diamond-2}
  Let $\q \in \Qp[\Diamond] \otimes \mathcal{ELI}$. Then
  \begin{itemize}
  \item $(i)$ a pointed substructure $(\mathcal C, (x,0))$ of $(\textit{cl}(\mathcal A^{\q}), (x,0))$ that is a core is computable in polynomial time;

   \item $(ii)$ there is a query $\q' \in \Qp[\Diamond] \otimes \mathcal{ELI}$ such that $\textit{cl}(\mathcal A^{\q'}) = \mathcal C$.
   \end{itemize}
\end{theorem}

\paragraph{Frontier for 2D queries.} We now define a \emph{frontier} $\mathcal{F}(\q, \r, x)$ for a node $\r$ with a distinguished variable $x$ in a query $\q \in \Qp[\Diamond] \otimes \mathcal{ELI}$ by induction on the construction of (the tree of) $\q$:

\begin{itemize}
  \item if $\q$ is $\xrightarrow{x} \r$ (i.e., an $\mathcal{ELI}$-query $\r$ with a distinguished variable $x$), then we set $\mathcal{F}(\q, \r, x)$ equal to $\mathcal{F}(\r, x)$, which takes the form $\{ \xrightarrow{x} \r^1, \dots, \xrightarrow{x} \r^k \}$.
  \item otherwise, let $\q$ be \\
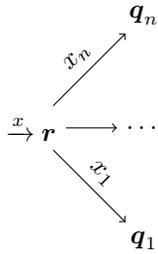
\begin{tikzpicture}[grow=right, sloped]
\node {$\xrightarrow{x} \r$}
    child {
        node {$\q_1$}
            edge from parent[->]
            node[above] {$x_1$}
    }
    child[->] {node {$\dots$}}
    child {
        node {$\q_n$}
            edge from parent[->]
            node[above] {$x_n$}
    };
\end{tikzpicture}\\
In this case, we set $\mathcal{F}(\q, \r, x)$ equal to $\avec{Q}_1 \cup \dots \cup \avec{Q}_n \cup \avec{Q}'$, where%
    \begin{itemize}
    \item assuming $\mathcal{F}(\q_i, \r_i, x_i) = \{\xrightarrow{x_i} \avec{f}^1, \dots, \xrightarrow{x_i} \avec{f}^\ell\}$, we set $\avec{Q}_i = \{\xrightarrow{x} \q_i^1, \dots, \xrightarrow{x} \q_i^\ell \}$ and $\q_i^j$:\\
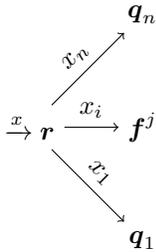
\begin{tikzpicture}[grow=right, sloped]
\node {$\xrightarrow{x} \r$}
    child {
        node {$\q_1$}
            edge from parent[->]
            node[above] {$x_1$}
    }
    child {
        node {$\avec{f}^j$}
            edge from parent[->]
            node[above] {$x_i$}
    }
    child {
        node {$\q_n$}
            edge from parent[->]
            node[above] {$x_n$}
    };
\end{tikzpicture}\\
    \item assuming $\mathcal{F}(\r, x) = \{ \xrightarrow{x} \r^1, \dots, \xrightarrow{x} \r^k \}$, we set $\avec{Q}' = \{ \xrightarrow{x} \q^1, \dots, \xrightarrow{x} \q^k \}$ and $\q^j$:\\
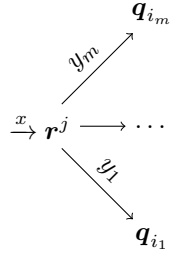
\begin{tikzpicture}[grow=right, sloped]
\node {$\xrightarrow{x} \r^j$}
    child {
        node {$\q_{i_1}$}
            edge from parent[->]
            node[above] {$y_1$}
    }
    child[->] {node {$\dots$}}
    child {
        node {$\q_{i_m}$}
            edge from parent[->]
            node[above] {$y_m$}
    };
\end{tikzpicture}\\
where $\{y_1, \dots, y_m\}$ is the domain of $\r^j$ and $\textit{orig}(y_\ell) = x_{i_\ell}$.
    \end{itemize}
\end{itemize}

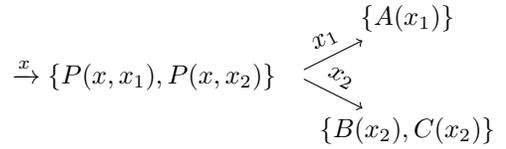
\begin{example}\em
Let $\q = \exists P.\Diamond A \land \exists P.\Diamond (B \land C)$, i.e.,
\begin{center}
\begin{tikzpicture}[grow=right, sloped]
\node[label=left:{$\xrightarrow{x}\{P(x,x_1), P(x,x_2)\}$} ] {}
    child {
        node {$\{ B(x_2), C(x_2)\}$}
            edge from parent[->]
            node[above] {$x_2$}
    }
    child {
        node {$\{A(x_1)\}$}
            edge from parent[->]
            node[above] {$x_1$}
    };
\end{tikzpicture}
\end{center}
Clearly, $\q$ is core. Then, $\mathcal{F}(\q, \{A(x_1)\}, x_1) = \{(\emptyset, x_1)\}$ and $\mathcal{F}(\q, \{B(x_2), C(x_2)\}, x_2) = \{(\{B(x_2)\}, x_2), (\{C(x_2)\}, x_2) \}$. Therefore, $\mathcal F(\q, \{P(x,x_1), P(x,x_2)\}, x) = \avec{Q}_1 \cup \avec{Q}_2 \cup \avec{Q}'$, where:
\begin{itemize}
\item $\avec{Q}_1$ contains a single structure:\\
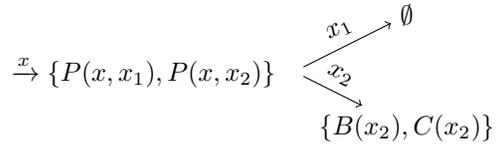
\begin{tikzpicture}[grow=right, sloped]
\node[label=left:{$\xrightarrow{x}\{P(x,x_1), P(x,x_2)\}$} ] {}
    child {
        node {$\{ B(x_2), C(x_2)\}$}
            edge from parent[->]
            node[above] {$x_2$}
    }
    child {
        node {$\emptyset$}
            edge from parent[->]
            node[above] {$x_1$}
    };
\end{tikzpicture}
\item $\avec{Q}_2$ contains two structures:\\
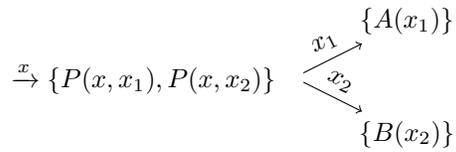
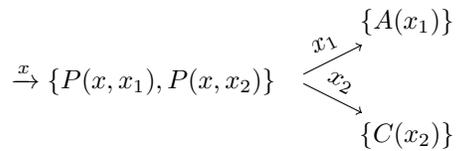
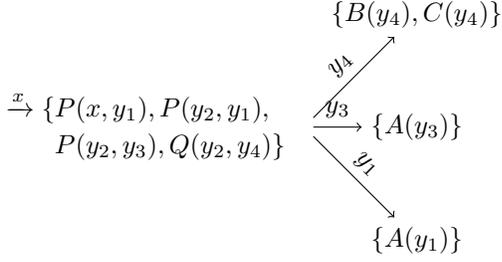
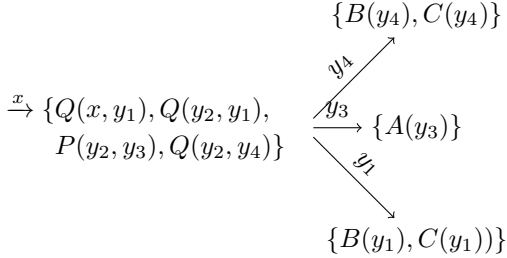
\begin{tikzpicture}[grow=right, sloped]
\node[label=left:{$\xrightarrow{x}\{P(x,x_1), P(x,x_2)\}$} ] {}
    child {
        node {$\{ B(x_2)\}$}
            edge from parent[->]
            node[above] {$x_2$}
    }
    child {
        node {$\{A(x_1)\}$}
            edge from parent[->]
            node[above] {$x_1$}
    };

\node[label=left:{$\xrightarrow{x}\{P(x,x_1), P(x,x_2)\}$} ] at (0, -2.2) {}
    child {
        node {$\{ C(x_2)\}$}
            edge from parent[->]
            node[above] {$x_2$}
    }
    child {
        node {$\{A(x_1)\}$}
            edge from parent[->]
            node[above] {$x_1$}
    };
\end{tikzpicture}
\item $\avec{Q}'$ contains two structures:\\
\begin{tikzpicture}[grow=right, sloped]
\node[label=left:{$\begin{aligned}\xrightarrow{x}\{&P(x,y_1), P(y_2,y_1),\\ &P(y_2, y_3), Q(y_2, y_4)\}\end{aligned}$} ] {}
    child {
        node {$\{ A(y_1)\}$}
            edge from parent[->]
            node[above] {$y_1$}
    }
    child {
        node {$\{A(y_3)\}$}
            edge from parent[->]
            node[above] {$y_3$}
    }
    child {
        node {$\{B(y_4), C(y_4)\}$}
            edge from parent[->]
            node[above] {$y_4$}
    };

\node[label=left:{$\begin{aligned}\xrightarrow{x}\{&Q(x,y_1), Q(y_2,y_1),\\ &P(y_2, y_3), Q(y_2, y_4)\}\end{aligned}$} ] at (0, -4) {}
    child {
        node {$\{ B(y_1), C(y_1))\}$}
            edge from parent[->]
            node[above] {$y_1$}
    }
    child {
        node {$\{A(y_3)\}$}
            edge from parent[->]
            node[above] {$y_3$}
    }
    child {
        node {$\{B(y_4), C(y_4)\}$}
            edge from parent[->]
            node[above] {$y_4$}
    };
\end{tikzpicture}
\end{itemize}
\end{example}

We observe the following important property, which can be proved by induction on the construction of (the tree of) $\q$ using Lemma~\ref{th:eli-diamond-1} and Theorem~\ref{th:eli-diamond-2}:
\begin{theorem}\label{th:eli-diamond-3}
Let $\q \in \Qp[\Diamond] \otimes \mathcal{ELI}$ be such that $\textit{cl}(\mathcal A^{\q})$ is a core and $\r$ be the root of $\q$ with a distinguished variable $x$. Then $(i)$ $\mathcal F(\q, \r, x) \subseteq \Qp[\Diamond] \otimes \mathcal{ELI}$; $(ii)$ $\mathcal F(\q, \r, x)$ is polynomial in the size of $\q$; $(iii)$ the following properties hold:
\begin{itemize}
\item $\q \models \q'$ for every $\q' \in \mathcal F(\q, \r, x)$;
\item $\q' \not \models \q$ for every $\q' \in \mathcal F(\q, \r, x)$;
\item if $\q \models \q''$, then either $\q'' \models \q$ or there exits $\q' \in \mathcal F(\q, \r, x)$ such that $\q' \models \q''$, for any $\q'' \in \Qp[\Diamond] \otimes \mathcal{ELI}$.
\end{itemize}
\end{theorem}

\paragraph{Constructing examples.} Given a $\q \in \Qp[\Diamond] \otimes \mathcal{ELI}$, we take $\q'$ from Theorem~\ref{th:eli-diamond-2} and set $E^+_{\q} = \{\D^{\q'}, a^{\q'} \}$. We take $E^-_{\q} = \{ \D^{\q''}, a^{\q''} \mid \q'' \in \mathcal F(\q', \r, x') \}$, where $\r$ is the root of $\q'$ with distinguished variable $x'$. Finally, using Lemma~\ref{th:eli-diamond-1} and Theorems~\ref{th:eli-diamond-2} and~\ref{th:eli-diamond-3}, we obtain the following:
\begin{theorem}
$(E^+_{\q}, E^-_{\q})$ uniquely characterises $\q$, for any $\q \in \Qp[\Diamond] \otimes \mathcal{ELI}$.
\end{theorem}

For $\q \in   \Qp[\nxt] \otimes \mathcal{ELI}$ the construction of $\mathcal F(\q, \r)$ is even easier. Indeed, it is sufficient to take $\textit{cl}(\mathcal A^{\q}) = \mathcal A^{\q}$. This completes the proof of Theorem~\ref{thm:firstone} $(ii)$.

\subsection{Proof of Theorem~\ref{thm:badd}}

	For $n \ge 1$, $1 \le i \le n$, $j \le i$, define $\Qp[\mathcal{EL}/\nxt,\Diamond]$-queries $\q_{i,j}^n$ recursively by taking:
	\begin{align*}
		& \q^n_{ii} = \r_i^n, \quad \q_{i,j-1}^n = B \land \nxt (B \land \Diamond (A \land \nxt \q_{i,j}^n)),\\
		& \r_n^n = \Diamond (B \land \Diamond A), \ \ \r_l^n = \Diamond (B \land \Diamond (A \land \nxt \s_{n-l})), \ \ l < n,\\
		& \s_1 = B \land \nxt (B \land \Diamond A), \ \ \s_{i+1} = B \land \nxt (B \land \Diamond (A \land \nxt \q_{i}))
	\end{align*}
	and set $\q_i^n = \q_{i,1}^n$. Minimal models of $\q_i^n$ look as follows:
	$$
	\underbrace{BB\emptyset^*A}_1 \dots \underbrace{BB\emptyset^*A}_{i-1} \underbrace{\emptyset^*B\emptyset^*A}_i \underbrace{BB\emptyset^*A}_{i+1}  \dots \underbrace{BB\emptyset^*A}_n
	$$
	%
	%
	Consider the $\Qp[\mathcal{EL}/\nxt,\Diamond]$-query $\q = \exists P.\q_1^n \land \dots \land \exists P.\q_n^n$. We claim that any unique characterisation of $\q$ contains at least $2^n$ positive examples (in $E^+$). Indeed, let $\mathcal Q$ be the set of all queries of the form $\q \land \exists P.\s$ with
	$$
	\s = \op_1 (B \land \Diamond (A \land \op_2 (B \land \Diamond (A \land \dots \land \op_n (B \land \Diamond A) \dots) ))),
	$$
	where each $\op_i$ is either $\nxt$ or $\Diamond \nxt$ if $i> 1$, and either blank or $\Diamond$ if $i = 1$. Observe that $|\mathcal Q| = 2^n$ and $\q' \models \q$ for each $\q' \in \mathcal Q$. On the other hand, $\q \not \models \q'$. Indeed, let $I_0$ be the set of indices $i$ such that $\op_i = \nxt$ and $i > 1$ or $\op_i$ is blank and $i = 1$. Let $I_1 = [1,n] \setminus I_0$. Define a data instance $\D_{I_0}$ with $\ind(\D_{I_0}) = \{a_0, \dots, a_n\}$ and $\max (\D_{\I_0}) = 3n-1$ as follows. To begin with, $\D_{I_0}$ contains $P(a_0, a_1, 0), \dots, P(a_0, a_n, 0)$. For $i \in [1, n]$ and $k \in [0, 3n-1]$, $\D_{I_0}$ contains $A(a_i, k)$ if $k \equiv 2 \,(\text{mod}\, 3)$, and $B(a_i, k)$ if $k \not \in [3(i-1), 3i)$ and $k \not \equiv 2 \,(\text{mod}\, 3)$. Also, for $i \in I_0$ and $k \in [3(i-1), 3i)$, $\D_{I_0}$ contains  $B(a_i, k)$ if $k \equiv 1 \,(\text{mod}\, 3)$. Finally, for each $i \in I_1$ and $k \in [3(i-1), 3i)$, it has $B(a_i, k)$ if $k \equiv 0 \,(\text{mod}\, 3)$. For example, $\D_{I_0}$ for $n=3$ and $I_0 = \{1,3\}$ is shown below:\\
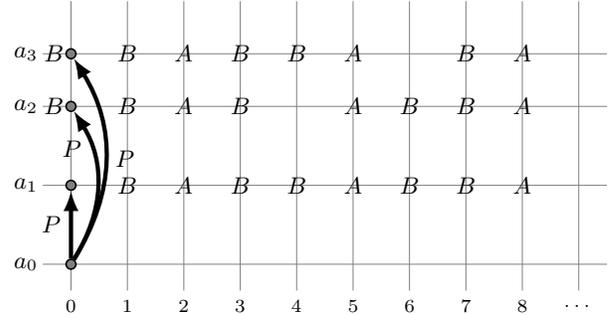
	\begin{tikzpicture}[>=latex, yscale=0.7, xscale=.75, semithick,
		spoint/.style={circle,fill=gray,draw=black,minimum size=1.3mm,inner sep=0pt},
		qpoint/.style={circle,fill=white,draw=black,minimum size=2mm,inner sep=0pt},
		graypoint/.style={circle,fill=white,draw=gray,minimum size=1.1mm,very thin,inner sep=0pt}]\footnotesize
		\begin{scope}[ultra thin]\small
			\draw[gray] (0.5, 0) -- ++(10,0); \node at (0.2, 0) {$a_0$};
			\draw[gray] (0.5, 1.5) -- ++(10,0); \node at (0.2, 1.5) {$a_1$};
			\draw[gray] (0.5, 3) -- ++(10,0); \node at (0.2, 3) {$a_2$};
			\draw[gray] (0.5, 4) -- ++(10,0); \node at (0.2, 4) {$a_3$};
			%
			\draw[gray] (1,-0.5) -- ++(0,5.5); \node at (1,-0.8) {\scriptsize $0$};
			\draw[gray] (2,-0.5) -- ++(0,5.5); \node at (2,-0.8) {\scriptsize $1$};
			\draw[gray] (3,-0.5) -- ++(0,5.5); \node at (3,-0.8) {\scriptsize $2$};
			\draw[gray] (4,-0.5) -- ++(0,5.5); \node at (4,-0.8) {\scriptsize $3$};
			\draw[gray] (5,-0.5) -- ++(0,5.5); \node at (5,-0.8) {\scriptsize $4$};
			\draw[gray] (6,-0.5) -- ++(0,5.5); \node at (6,-0.8) {\scriptsize $5$};
			\draw[gray] (7,-0.5) -- ++(0,5.5); \node at (7,-0.8) {\scriptsize $6$};
			\draw[gray] (8,-0.5) -- ++(0,5.5); \node at (8,-0.8) {\scriptsize $7$};
			\draw[gray] (9,-0.5) -- ++(0,5.5); \node at (9,-0.8) {\scriptsize $8$};
			\draw[gray] (10,-0.5) -- ++(0,5.5); \node at (10,-0.8) {\scriptsize $\dots$};
			
		\end{scope}
		%
		\node[spoint] (a1) at (1,0) {};
		\node[spoint] (b1) at (1,1.5) {};
		\node[spoint] (c1) at (1,3) {}; \node at (.7,3) {$B$};
		\node[spoint] (d1) at (1,4) {}; \node at (.7,4) {$B$};
		
		\node (a2) at (2,0) {};
		\node (b2) at (2,1.5) {$B$};
		\node (c2) at (2,3) {$B$};
		\node (d2) at (2,4) {$B$};
		
		\node (a3) at (3,0) {};
		\node (b3) at (3,1.5) {$A$};
		\node (c3) at (3,3) {$A$};
		\node (d3) at (3,4) {$A$};
		
		\node (a4) at (4,0) {};
		\node (b4) at (4,1.5) {$B$};
		\node (c4) at (4,3) {$B$};
		\node (d4) at (4,4) {$B$};
		
		\node (a5) at (5,0) {};
		\node (b5) at (5,1.5) {$B$};
		\node (c5) at (5,3) {};
		\node (d5) at (5,4) {$B$};
		
		\node (a6) at (6,0) {};
		\node (b6) at (6,1.5) {$A$};
		\node (c6) at (6,3) {$A$};
		\node (d6) at (6,4) {$A$};
		
		\node (a7) at (7,0) {};
		\node (b7) at (7,1.5) {$B$};
		\node (c7) at (7,3) {$B$};
		\node (d7) at (7,4) {};
		
		\node (a8) at (8,0) {};
		\node (b8) at (8,1.5) {$B$};
		\node (c8) at (8,3) {$B$};
		\node (d8) at (8,4) {$B$};
		
		\node (a9) at (9,0) {};
		\node (b9) at (9,1.5) {$A$};
		\node (c9) at (9,3) {$A$};
		\node (d9) at (9,4) {$A$};
		\begin{scope}[ultra thick]
			\draw[->] (a1) to  node[midway,left] {$P$} (b1);
			\draw[->,bend right] (a1) to  node[near end,left] {$P$} (c1);
			\draw[->,bend right] (a1) to  node[midway,right] {$P$} (d1);

		\end{scope}
	\end{tikzpicture}
	Then $\D_{I_0},a_0,0 \models \q$ but $\D_{I_0},a_0,0 \not \models \q'$. To illustrate, take $\s = B \land \Diamond (A \land \land \Diamond \nxt (B \land \Diamond (A \land \nxt (B \land \Diamond A))))$, for which $\q' = \q \land \exists P.\s \in \mathcal Q$ with $I_0 = \{1,3\}$. One can readily see that $\q'$ cannot be satisfied at $a_0,0$ in $\D_{I_0}$ depicted above.
	
	To\ complete the proof that $|E^+| \ge 2^n$, we prove the following property (cf. Example~\ref{thm:superpolb}):
	\begin{equation*}
		\D \models \q \ \ \&\ \ \D \not\models \q' \ \ \& \ \ (\q'\ne \q'') \ \ \Rightarrow \ \ \D \models \q'',
	\end{equation*}
	for all data instances $\D$ and $\q', \q'' \in \mathcal Q$. Indeed, take an arbitrary $\q$ such that $\D \models \q$ and let $\{a_0, \dots a_n\} \subseteq \ind(\D)$ be the elements such that $P(a_0, a_i, 0)$ is in $\D$ and $\D, a_i, 0 \models \q_i^n$, for each $1 \leq i \leq n$. Let $(\avec{K}_{i,1}, \dots, \avec{K}_{i,n})$ be a vector of pairs/triples over $\mathbb N$ such that, for $j \neq i$, $\avec{K}_{i,j} = (b, b', a)$ and $\avec{K}_{i,i} = (b,a)$, where:
	\begin{itemize}
		\item $b < b' < a$, $b < a$ in every $\avec{K}_{i,j}$;
		
		\item $b$ from $\avec{K}_{i,j}$ is strictly greater than $a$ from $\avec{K}_{i,j-1}$;
		
		\item $a$ from $\avec{K}_{i,j}$ is strictly smaller than $b$ from $\avec{K}_{i,j+1}$;
		
		\item $\D, a_i, b \models B$, $\D, a_i, b' \models B$, $\D, a_i, a \models A$ for all $(b, b', a) \in \avec{K}_{i,j}$;
		
		\item $\D, a_i, b \models B$, $\D, a_i, a \models A$ for all $(b, a) \in \avec{K}_{i,i}$.
	\end{itemize}
	Note that, for each $1 \leq i \leq n$, a vector $(\avec{K}_{i,1}, \dots, \avec{K}_{i,n})$ as above exists.
	Take $\q' \in \mathcal Q$ such that $\D \not\models \q'$ and let $I_0$ be the set corresponding to $\q'$. We observe that $\D, a_i, 0 \not \models \s$, for all $1 \leq i \leq n$. Therefore, for any $i \in I_0$ and any vector $(\avec{K}_{i,1}, \dots, \avec{K}_{i,n})$, we have that $b > a+1$ for $a$ from $\avec{K}_{i,i-1}$ and $b$ from $\avec{K}_{i,i}$, if $i > 0$, and $b > 0$ for $b$ from $\avec{K}_{i,i}$, if $i = 1$. Similarly, for any $i \in [1,n] \setminus I_0$ and any vector $(\avec{K}_{i,1}, \dots, \avec{K}_{i,n})$, we have that $b = a+1$ for $a$ from $\avec{K}_{i,i-1}$ and $b$ from $\avec{K}_{i,i}$, if $i > 0$, and $b = 0$ for $b$ from $\avec{K}_{i,i}$, if $i = 1$. Now, take an arbitrary $\q'' = \q \land \exists P.\s'' \in \mathcal Q$ such that $\q'' \neq \q'$ and let $J_0$ be the set corresponding to it. Clearly, $J_0 \neq I_0$. Suppose first that there exists $i \in J_0 \setminus I_0$. Therefore, $i \in [1,n] \setminus I_0$ and then, as it is easy to see, $\D, a_i, 0 \models \s''$. Thus, $\D, a_0, 0 \models \q''$ as was required. Now, if there is $i \in I_0 \setminus J_0$, the proof is analogous and left to the reader.

\subsection{Proof of Theorem~\ref{thm:nextdiamond2}}
Recall that we consider the class
$\Qp[\nxt,\Diamondw](\mathcal{ELI})$ of queries of the form
\begin{equation}\label{dnpathappendix}
	\q = \r_0 \land \op_1 (\r_1 \land \op_2 (\r_2 \land \dots \land \op_n \r_n) ),
\end{equation}
where the $\r_{i}$ are $\mathcal{ELI}$-queries and $\op_i \in \{\nxt,\Diamond,\Diamondw\}$. We first generalise the
normal form introduced for $\mathcal{Q}_{p}[\Diamondw,\nxt]$.
Any $\q$ in $\Qp[\nxt,\Diamondw](\mathcal{ELI})$
can be represented as a sequence
$$
\r_{0}(t_0), R_{1}(t_0,t_1), \dots, \r_{n-1}(t_{n-1}), R_{n}(t_{n-1},t_n),\r_{n}(t_n),
$$
where $R_{i}\in \{\suc,<,\leq\}$ and $\r_i$ is an $\mathcal{ELI}$-query. As before, we divide $\q$ into \emph{blocks} $\q_i$ such that
\begin{align}\label{fullq2appendix}
	\q = \q_{0} \mathcal{R}_{1} \q_{1} \dots \mathcal{R}_{n} \q_{n}
\end{align}
with $\mathcal{R}_{i} = R_{1}^{i}(t_{0}^{i},t_{1}^{i}) \dots  R_{n_{i}}^{i}(t_{n_{i}-1}^{i},t_{n_{i}}^{i})$, for \mbox{$R_{j}^{i}\in \{<,\leq\}$},
\begin{align*}
	\q_{i}= \r_{0}^{i}(s_{0}^{i}) \suc (s_{0}^{i},s_{1}^{i}) \r_{1}^{i}(s_{1}^{i}) \dots \suc(s_{k_{i}-1}^{i},s_{k_{i}}^{i}) \r_{k_{i}}^{i}(s_{k_{i}}^{i})
\end{align*}
and $s_{k_{i}}^{i}=t_{0}^{i+1}$, $t_{n_{i}}^{i}=s_{0}^{i}$. If $k_{i}=0$, the block $\q_{i}$ is \emph{primitive}. A primitive block $\q_{i}=\r_{0}^{i}(s_{0}^{i})$ with $i>0$ such that $\r_{0}^{i}$ is not equivalent to a conjunction of $\mathcal{ELI}$-queries that are not equivalent to $\r_{0}^{i}$ is called a \emph{lone conjunct}.
Now, we say that $\q$ is in \emph{normal form} if the
following conditions hold:
\begin{description}
	\item[(n1$'$)] $\r_{0}^{i}\not\equiv\top$ if $i>0$, and $\r_{k_{i}}^{i}\not\equiv\top$ if $i>0$ or $k_{i}>0$
	(thus, of all the first/last $\r$ in a block only $\r_0^0$ can be trivial);
	
	\item[(n2$'$)] each $\mathcal{R}_{i}$ is either a single $t_{0}^{i}\leq t_{1}^{i}$ or a sequence of $<$;
	
	\item[(n3$'$)] $\r_{k_{i}}^{i}\not\models \r_{0}^{i+1}$ if $\q_{i+1}$ is primitive and $R_{i+1}$ is $\le$;
	
	\item[(n4$'$)] $\r_{0}^{i+1}\not\models\r_{k_{i}}^{i}$ if $i>0$, $\q_{i}$ is primitive and $R_{i+1}$ is $\le$.
\end{description}
\begin{lemma}
	Every query in $\Qp[\nxt,\Diamondw](\mathcal{ELI})$ is equivalent to a query in normal form that can be computed in polynomial time.
\end{lemma}
We call a query in $\Qp[\nxt,\Diamondw](\mathcal{ELI})$ \emph{safe} if it is equivalent to a query in normal form in $\Qp[\nxt,\Diamondw](\mathcal{ELI})$
without lone conjuncts.

\medskip
\noindent
{\bf Theorem~\ref{thm:nextdiamond2}.}
{\em
	$(i)$ A query $\q\in \Qp[\nxt,\Diamondw](\mathcal{ELI})$ is uniquely characterisable within $\Qp[\nxt,\Diamondw](\mathcal{ELI})$ iff $\q$ is safe.
	
	$(ii)$ Those queries that are uniquely characterisable within $\Qp[\nxt,\Diamondw](\mathcal{ELI})$ are actually polynomially characterisable within $\Qp[\nxt,\Diamondw](\mathcal{ELI})$.
	
	$(iii)$ The class $\Qp[\nxt,\Diamondw](\mathcal{ELI})$ is polynomially characterisable for bounded query size.
	
	$(iv)$ The class $\Qp[\nxt,\Diamond](\mathcal{ELI})$ is polynomially characterisable.
}

\medskip
We show how the construction of positive and negative examples provided in the proof of Theorem~\ref{thm:nextdiamond} can be generalised.

Suppose $\q$ in normal form~\eqref{fullq2appendix} does not contain lone conjuncts.
Let $b$ be again the number of $\nxt$ and $\Diamond$ in $\q$ plus 1.
We construct $E = (E^+, E^-)$ characterising $\q$ as follows. Pick an individual name $a$ and let, for every $\mathcal{ELI}$-query $\r$ in $\q$, $\hat{\r}$ denote the pointed tree-shaped data instance defined by $\r$ (note that we take the same individual $a$ for every $\r$).

For each block $\q_i$ in~\eqref{fullq2appendix}, we take two
temporal data instances
$$
\qw_{i} = \hat{\r}_{0}^{i}, \dots \hat{\r}_{k_{i}}^{i}
$$
$$
\qw_{i} \Join \qw_{i+1} = \hat{\r}_{0}^{i}, \dots \hat{\r}_{k_{i}}^{i}\cup \hat{\r}_{0}^{i+1}, \dots \hat{\r}_{k_{i+1}}^{i+1}.
$$
The set $E^{+}$ contains the data instances given by
\begin{itemize}
	\item[--] $\D_{b} = \qw_{0} \emptyset^{b} \dots \qw_{i} \emptyset^{b} \qw_{i+1} \dots \emptyset^{b} \qw_{n}$,
	
	\item[--] $\D_{i} = \qw_{0} \emptyset^{b} \dots \qw_{i} \! \Join \! \qw_{i+1} \dots \emptyset^{b} \qw_{n}$ if $\mathcal{R}_{i+1}$ is $\leq$,
	
	\item[--] $\D_{i} = \qw_{0} \emptyset^{b} \dots \qw_{i} \emptyset^{n_{i+1}} \qw_{i+1} \dots \emptyset^{b} \qw_{n}$ otherwise.
\end{itemize}
The set $E^{-}$ contains all data instances of the form
\begin{itemize}
	\item[--] $\D_i^- = \qw_{0} \emptyset^{b} \dots \qw_{i} \emptyset^{n_{i+1} - 1} \qw_{i+1} \dots \emptyset^{b} \qw_{n}$ if $n_{i+1} > 1$;
	
	\item[--] $\D^-_i = \qw_{0} \emptyset^{b} \dots \qw_{i} \! \Join \! \qw_{i+1} \dots \emptyset^{b} \qw_{n}$ if $\mathcal{R}_{i+1}$ is a single $<$,
\end{itemize}
and also the data instances obtained from $\D_{b}$ by
\begin{description}
	\item[\rm (a)] replacing $\hat{\r}^i_j \ne \emptyset$ by an element of $\mathcal{F}(\hat{\r}^i_j)$ or removing the whole $\hat{\r}^i_j = \emptyset$, for $i \ne 0$ and $j \ne 0$, from some $\qw_i$;
	
	\item[\rm (b)] replacing $\qw_i = \hat{\r}_0^i, \dots \hat{\r}^i_l,\hat{\r}^i_{l+1} \dots \hat{\r}_{k_i}^i$ ($k_i > 0$) by $\qw'_i \emptyset^b \qw''_i$, where $\qw'_i = \hat{\r}_0^i \dots \hat{\r}^i_l$, $\qw''_i = \hat{\r}^i_{l+1} \dots \hat{\r}_{k_i}^i$ and $l \ge 0$;
	
	
	\item[\rm (c)] replacing some $\hat{\r}_l^i \ne \emptyset$, $0 < l < k_i$, by $\hat{\r}_l^i \emptyset^b \hat{\r}_l^i$;
	
	
	\item[\rm (d)] replacing $\hat{\r}^i_{k_i}$ ($k_i > 0$)
	with $\mathcal{A}\emptyset^b \hat{\r}^i_{k_i}$, for some $\mathcal{A}\in \mathcal{F}(\hat{\r}^i_{k_i})$, or replacing $\hat{\r}^i_{0}$ ($k_i > 0$) with $\hat{\r}^i_{0}  \emptyset^b \mathcal{A}$, for some $\mathcal{A}\in \mathcal{F}(\hat{\r}^i_{0})$;
	\item[\rm (e)] replacing $\hat{\r}_0^0 \ne \emptyset$ with $\mathcal{A}\emptyset^{b}\hat{\r}^{0}_{0}$, for $\mathcal{A}\in \mathcal{F}(\hat{\r}^0_0)$, if $k_0 = 0$, and with $\hat{\r}_0^0 \emptyset^b \hat{\r}_0^0$ if $k_0 > 0$.
\end{description}
Let $\q$ be of the form \eqref{fullq2appendix}. We generalise the notion of a homomorphism as follows. Let $\mathcal{D}=\mathcal{A}_{0},\ldots,\mathcal{A}_{n}$
and $a\in \ind(\D)$.
A mapping $h$ from the set $\var(\q)$ of variables in $\q$ to $[0,\len(\D)]$ is a \emph{generalised homomorphism} from $\q$ to $\mathcal{D}$ for $a$ if
$h(t_{0})=0$, $\mathcal{A}_{h(t)}\models \r(a)$ if $\r(t)$, $h(t') = h(t) +1$ if $\suc(t,t') \in \q$, and $h(t)\, R\, h(t')$ if $R(t,t') \in \q$ for $R \in \{<,\leq\}$. Then one can show that $\mathcal{D},a,0\models \q$ if there exists
a generalised homomorphism from $\q$ to $\D$ for $a$.

It is now almost trivial to extend the proof of Theorem~\ref{thm:nextdiamond}
to a proof of Theorem~\ref{thm:nextdiamond2} by replacing homomorphisms by generalised homomorphisms. For example, block surjectivity is generalised in a straightforward way as follows: a generalised homomorphism $h \colon \q' \to D_{b}$ is \emph{block surjective} if every point in every block $\qw_i$ of $\D_{b}$ is in the range $\textit{ran}(h)$ of $h$. To define type surjectivity let, for $\ell\in \textit{ran}(h)$, $\r_{\ell}$ denote the conjunction of all $\mathcal{ELI}$-queries $\r'$ with $\r'(t)$ in $\q'$ such that $h(t)=\ell$. Clearly $\mathcal{A}_{\ell}\models \r_{\ell}(a)$. $h$ is \emph{type surjective} if $\mathcal{A}\not\models \r_{\ell}(a)$, for every $(\mathcal{A},a)\in \mathcal{F}(\hat{\r})$.

\subsection{Proof of Theorem~\ref{uno2}}
Suppose we are given a $\mathcal{P}^{\Sigma}[\U](\mathcal{EL})$-query
$$
\q= \r_0 \land (\el_1 \U (\r_1 \land ( \el_2 \U ( \dots (\el_n \U \r_n) \dots )))).
$$
%
To show that it is uniquely characterised by the polynomial-size example set with $(\mathfrak p'_0)$--$(\mathfrak p'_2)$ and $(\mathfrak n'_0)$--$(\mathfrak n'_2)$, consider any $\mathcal{Q}^{\Sigma}[\U](\mathcal{EL})$-query
$$
\q' = \r'_0 \land (\el'_1 \U (\r'_1 \land ( \el'_2 \U ( \dots (\el'_m \U \r'_m) \dots )))).
$$
such that $\q\not\equiv\q'$.

We now define a map $f$ that reduces the 2D case to the 1D case. Consider the alphabet
$$
\Gamma =\{\r_0,\dots,\r_n,\el_1,\dots,\el_n,\r'_0,\dots,\r_m,\el'_1,\dots,\el'_m, \}\setminus\{\bot\}
$$
in which we regard the $\mathcal{EL}$-queries $\r_i,\el_i,\r'_j,\el'_j$ as symbols.  Let $\hat{\Gamma}=\{(\hat{\avec a},a) \mid \avec a \in\Gamma\}$, that is, $\hat{\Gamma}$ consists of the pointed databases corresponding to the $\mathcal{EL}$-queries $\avec{a} \in \Gamma$.

For any $\mathcal{EL}$ instance query $\avec{a}$, we set
$$
f(\avec{a}) = \{\avec{b} \in \Gamma \mid (\hat{\avec{a}},a) \models \avec{b}\}.
$$
Similarly, for any $\mathcal{EL}$ pointed data instance $(\A,a)$, we set
$$
f(\A,a) = \{\avec{b} \in \Gamma \mid (\A,a) \models \avec{b}\}
$$
and, for any temporal data instance $\D = (\delta_0,\dots,\delta_k)$ with $\mathcal{EL}$ pointed data instances $(\delta_i,a_i)$, set
$$
f(\D) = (f(\delta_0,a_0), \dots, f(\delta_k,a_k)),
$$
which is an \LTL{}-data instance over the signature $\Gamma$. Finally, we define a query
$$
f(\q) = \rho_0 \land (\lambda_1 \U (\rho_1 \land (\lambda_2 \U( \dots (\lambda_n \U \rho_n) \dots ))))
$$
by taking $\rho_i = f(\r_i)$ and  $\lambda_i = f(\el_i)$, and similarly for $\q'$.

By definition, $f(\q)$ is a $\mathcal{P}^\Gamma[\U]$-query (indeed, since $\hat\r_i,a\not\models\el_i$, we have $\el_i\in f(\el_i)\setminus f(\r_i)$, and since $\hat\el_i,a\not\models\r_i$, we have $\r_i\in f(\r_i)\setminus f(\el_i)$), and $f(\q')$ is a $\mathcal{Q}_p^\Gamma[\U]$-query such that $f(\q) \not\equiv f(\q')$: it follows immediately from the definition that, for any data instance $\D$, we have $\D \models \q$ iff $f(\D) \models f(\q)$ and similarly for $\q'$. By Theorem~\ref{uno}, $\q$ and $\q'$ are separated by the corresponding example set with $(\mathfrak p_0)$--$(\mathfrak p_2)$ and $(\mathfrak n_0)$--$(\mathfrak n_2)$. Notice that the positive examples from $(\mathfrak p_0)$--$(\mathfrak p_2)$ are exactly the $f$-images of the examples $(\mathfrak p'_0)$--$(\mathfrak p'_2)$. So if $f(\q)$ and $f(\q')$ are separated by some $\D$ from $(\mathfrak p_0)$--$(\mathfrak p_2)$, the corresponding member of $(\mathfrak p'_0)$--$(\mathfrak p'_2)$ separates $\q$ and $\q'$.

So suppose $f(\q)$ and $f(\q')$ are separated by some $\D$ from $(\mathfrak n_0)$--$(\mathfrak n_2)$. If $\D=\Gamma^n$, then it means that the temporal depth of $f(\q')$ is less than the temporal depth of $f(\q')$, so $m<n$, and the queries $\q$ and $\q'$ are separated by $\mathcal A_\Sigma^n$.

Suppose $\D=\Gamma^{n-i}(\Gamma\setminus \{\avec{a}\})\Gamma^{i}$. Since $\D\not\models f(\q)$, we have $f(\r_i)\not\subseteq\Gamma\setminus\{\avec{a}\}$, and so $\hat\r_i,a\models \avec a$ but $f(\r'_i)\subseteq\Gamma\setminus\{\avec{a}\}$. Then $\hat{\r}'_i,a\not\models \avec a$, and so $\r'_i\not\models\r_i$.  Therefore, there is $(\mathcal A,a)\in \mathcal S(\{\r_i\})$ such that $\mathcal A\models\r'_i(a)$ and $\mathcal A_\Sigma^{n-i}\mathcal A\mathcal A_\Sigma^{i}$ separates $\q$ and $\q'$. 

The cases when $\D$ is from $(\mathfrak n_1)$ or $(\mathfrak n_2)$ are treated in a similar manner.

If $\q\in \mathcal{P}^{\Sigma}[\U](\mathcal{ELI})$, the characterisation is exponential as the size of $(\mathfrak n_0)$--$(\mathfrak n_2)$ is exponential because the exponential size of constructed split partners for $\mathcal{ELI}$-queries.

\subsection{Proof of Theorem~\ref{thm:split}}

{\bf Theorem~\ref{thm:split}.}
{\em Fix $n>0$. For any set $Q$ of $\mathcal{EL}(\Sigma)$-queries with $|Q|\leq n$, one can compute in polynomial time a split partner $\mathcal{S}(Q)$ of $Q$ in $\mathcal{EL}(\Sigma)$.	
For $\mathcal{ELI}$, one can compute a split partner in exponential time, which is optimal as even for singleton sets $Q$ of $\mathcal{ELI}(\Sigma)$-queries, no polynomial-size split partner of $Q$ in $\mathcal{ELI}(\Sigma)$ exists in general.
}
\begin{proof}
	We begin by proving for every set $Q$ of $\mathcal{EL}(\Sigma)$-queries with $|Q|\leq n$, one can compute in polynomial time a split partner $\mathcal{S}(Q)$ of $Q$ in $\mathcal{EL}(\Sigma)$.	We prove the statement for $n=1$, the generalisation is straightforward. Let $Q=\{\q\}$. The construction is by induction over the depth of $\q$.
	Assume $\text{depth}(\q)=0$. Thus $\q=\bigwedge_{i=1}^{k}A_{i}$ with $A_{i}$ atomic concepts. Then let for $i \leq k$:
	\begin{eqnarray*}
		\mathcal{A}_{A_{i}} & = & \{B(a) \mid B \in \Sigma\setminus \{A\}\}\cup\\
		& &  \{ R(a,b), R(b,b) \mid R\in \Sigma\} \cup \\
		&  & \{ B(b) \mid B\in \Sigma\}
	\end{eqnarray*}
	and set $\mathcal{S}(\q) = \{ (\mathcal{A}_{A_{i}},a) \mid 1\leq i \leq k\}$. We show that $\mathcal{S}(\q)$ is as required. Assume
	$$
	\q'=\bigwedge_{i=1}^{m_{1}}B_{i} \wedge \bigwedge_{i=1}^{m_{2}} \exists R_{i}.\q_{i}
	$$
	If $\q'\not\models\q$, then there exist $A_{i}$ with $A_{i}\not\in \{B_{i}\mid 1\leq i \leq m_{1}\}$. Then $\mathcal{A}_{A_{i}} \models \q'(a)$, as required.
	Conversely, if $\mathcal{A}_{A_{i}} \models \q'(a)$ for some $A_{i}$, then $A_{i}\not\in \{B_{i}\mid 1\leq i \leq m_{1}\}$. Hence $\q'\not\models\q$, as required.
	
	\medskip
	
	Assume now that $\text{depth}(\q)=n+1$ and that split partners $\mathcal{S}(\q')$ have been defined for queries $\q'$ of depth $\leq n$. Assume
	$$
	\q=\bigwedge_{i=1}^{n_{1}}A_{i} \wedge \bigwedge_{i=1}^{n_{2}} \exists S_{i}.\q_{i}.
	$$
	Then assume $\mathcal{S}(\q_{i})= \{(\mathcal{A}_{1},a_{1}),\ldots,(\mathcal{A}_{k_{i}},a_{k_{i}})\}$ and let $c_{1},\ldots,c_{k_{i}}$ be fresh individuals. Define for $i \leq n_{2}$ the data instance
	\begin{eqnarray*}
		\mathcal{A}_{i} & = & \{B(a) \mid B \in \Sigma\}\cup\\
		&   &  \{ R(a,b), S(b,b), B(b) \mid R\in \Sigma\setminus\{S_{i}\},B,S\in \Sigma\} \cup \\
		& &  \{ S_{i}(a,c_{j}) \mid 1\leq j \leq k_{i}\} \cup\\
		& &  \mathcal{A}_{1}(c_{1}/a_{1}) \cup \cdots \cup \mathcal{A}_{k_{i}}(c_{k_{i}}/a_{k_{i}})
	\end{eqnarray*}
	with $\mathcal{A}(c/a)$ the result of replacing $a$ by $c$ in $\mathcal{A}$. Let
	$\mathcal{S}(\q)$ be the union of $\mathcal{S}(\bigwedge_{i=1}^{n_{1}}A_{i})$ and $\{(\mathcal{A}_{i},a)\mid 1\leq i\leq n_{2}\}$.
	
	We show that $\mathcal{S}(\q)$ is as required. Assume
	$$
	\q'=\bigwedge_{i=1}^{m_{1}}B_{i} \wedge \bigwedge_{i=1}^{m_{2}} \exists R_{i}.\q_{i}'
	$$
	If $\q'\not\models\q$, then either there exists $A_{i}$ with $A_{i}\not\in \{B_{i}\mid 1\leq i \leq m_{1}\}$ or there exists $\exists S_{i}.\q_{i}$ such that for every
	$R_{j}.\q_{j}$ with $S_{i}=R_{j}$, we have $\q_{j}'\not\models \q_{i}$. In the former case
	$\mathcal{A}_{A_{i}}\models \q'(a)$ and in the latter case, by induction hypothesis,
	$\mathcal{A}_{i}\models \q'(a)$. Conversely, if $\mathcal{A}\not\models \q'(a)$ for some $\mathcal{A},a\in \mathcal{S}(\q)$, then either there exists $A_{i}$ with $A_{i}\not\in \{B_{i}\mid 1\leq i \leq m_{1}\}$ or there exists $\exists S_{i}.\q_{i}$ such that for every
	$R_{j}.\q_{j}$ with $S_{i}=R_{j}$, there exists $\mathcal{A},a\in \mathcal{S}(\q_{i})$ such that $\mathcal{A}\not\models \q_{i}'(a)$. By induction hypothesis, $\q_{i}'\not\models \q_{i}$. But then $\q'\not\models \q$, as required.
	
	\bigskip
	
	We next show that for every set $Q$ of $\mathcal{ELI}(\Sigma)$-queries one can compute in exponential time a split partner $\mathcal{S}(Q)$ of $Q$ in $\mathcal{ELI}(\Sigma)$.	
	
	Assume $Q=\{\q_{1},\ldots,\q_{n}\}$. Let $S$ denote the closure under single negation of the set of subqueries of queries in $Q$. A \emph{type} $t$ is a subset of $S$ with the following properties:
	\begin{itemize}
		\item if $C_{1}\wedge C_{2}\in S$, then $C_{1},C_{2}\in t$ iff $C_{1}\sqcap C_{2}\in t$;
		\item if $\neg C \in t$, then $C\in t$ iff $\neg C \not\in t$.
	\end{itemize}
    We say that types $t_{1},t_{2}$ are \emph{$P$-coherent} if $C\in t_{2}$ implies $\exists P.C\in t_{1}$ for all $\exists P.C\in S$ and $C\in t_{1}$ implies $\exists P^{-}.C\in t_{2}$ for all $\exists P^{-}.C\in S$. Let $T$ denote the set of all types. We next obtain from $T$ in exponential time the set of all satisfiable types using a standard type elimination procedure: eliminate from $T$, recursively, all $t$ such that there exists $\exists P.C\in t$ such that there does not exist any $t'\in T$ such that $t,t'$ are $P$-coherent and $t'$ contains $C$ or there exists $\exists P^{-}.C\in t$ such that there does not exist any $t'\in T$ such that $t',t$ are $P$-coherent and $t'$ contains $C$. Denote by $T'$ the resulting set of types. $T'$ contains exactly the satisfiable types. Define a $\Sigma$-data instance $\mathcal{A}$ as the set of all $A(t)$ with $t\in T'$ and $A\in t$
    and $P(t,t')$ with $t,t'\in T'$ and $t,t'$ $P$-coherent, where we regard the types in $T'$ as individual names. Note that 
\begin{itemize}
	\item $\mathcal{D}\models C(t)$ iff $C\in t$ for all $C\in S$ and $t\in T'$;
    \item for any $\Sigma$-data instance $\mathcal{B}$ the mapping $a\mapsto t_{\mathcal{B}}(a)$ defined by setting $t_{\mathcal{B}}(a)=\{ C \in S \mid \mathcal{B}\models C(a)\}$ is a $\Sigma$-homomorphism from $\mathcal{B}$ to $\mathcal{A}$. 
\end{itemize}
Now set $$
\mathcal{S}(Q) = \{ (\mathcal{A},t) \mid t\in T', Q \cap t = \emptyset\}.
$$  
We show that $\mathcal{S}(Q)$ is a split partner of $Q$ in $\mathcal{ELI}(\Sigma)$. Assume that $\q'$ is an $\mathcal{ELI}(\Sigma)$-query. If $\mathcal{A}\models \q'(t)$ for some $(\mathcal{A},t)\in \mathcal{S}(Q)$, then by (i) $\q'\not\models \q$ for any $\q\in \mathcal{Q}$. If $\q'\not\models \q$ for all $\q\in \mathcal{Q}$, then we find a pointed $\Sigma$-data instance $\mathcal{B},a$ with $\mathcal{B}\models \q'(a)$ and $\mathcal{B}\not\models \q(a)$ for all $\q\in Q$. By (ii) $\mathcal{A}\models \q'(t_{\mathcal{B}}(a))$ and by (i) $(\mathcal{A},t_{\mathcal{B}}(a))\in \mathcal{S}(Q)$.

\bigskip
	
	We finally show that even for singleton sets $Q$ of $\mathcal{ELI}(\Sigma)$-queries, in general no polynomial size split partner of $Q$ in 
	$\mathcal{ELI}(\Sigma)$ exists. Let
	$$
	\q= \exists r.\bigwedge_{i=1}^{n}\exists r^{-}.A_{i},
	$$
	$Q=\{\q\}$, and $\Sigma=\{r,A_{1},\ldots,A_{n}\}$. We show that any split partner of $Q$ in $\mathcal{ELI}(\Sigma)$ contains at least $2^{n}$ pointed data instances.
	
	Let for any $i$, $\bar{A}_{i}= A_{i} \wedge \cdots A_{i-1} \wedge A_{i+1} \wedge \cdots \wedge A_{n}$ and let for every $X\subseteq \{1,\ldots,n\}$: 
	$$
	\q_{X}= (\bigwedge_{i\in X}A_{i}) \wedge \bigwedge_{i\not\in X}\exists r.\exists r^{-}.\bar{A}_{i}
	$$
	Observe that $\q_{X}\not\models \q$, for all $X\subseteq \{1,\ldots,n\}$. However, there does not exist any pointed data instance $\mathcal{D},a$ such that 	$\mathcal{D}\not\models \q(a)$ and $\mathcal{D}\models \q_{X_{1}}(a)$ and $\mathcal{D}\models \q_{X_{2}}(a)$ for distinct $X_{1}$ and $X_{2}$.
	To see this consider $X_{1}\not=X_{2}$ and assume $\mathcal{D}\models \q_{X_{i}}(a)$ for $i=1,2$. Take $i\in (X_{1}\setminus X_{2}) \cup (X_{2}\setminus X_{1})$. Then 
	$\mathcal{D}\models A_{i}(a)$ and $\mathcal{D}\models \exists r.\exists r^{-}.\bar{A}_{i}(a)
	$
	entail $\mathcal{D}\models \q(a)$. \\ \mbox{}
\end{proof}

We show that the lower bound shown above entails the lower bound for Theorem~\ref{uno2}.
We use the query 
$$
\q= \exists r.\bigwedge_{i=1}^{n}\exists r^{-}.A_{i},
$$
from the proof above and show that $\nxt\q$ is not polynomially characterisable within $\mathcal{Q}_{p}^\Sigma[\U](\mathcal{ELI})$, where $\Sigma=\{r,A_{1},\ldots,A_{n}\}$.
Let the queries $\q_{X}$ be defined as above. We show that if $X_{1}\not=X_{2}$, then there does not exist a pointed temporal database $\mathcal{D},a$
such that $\mathcal{D},a,0\not\models \nxt\q$ but $\mathcal{D},a,0\models \q_{X_{1}}\U\q$ and
$\mathcal{D},a,0\models \q_{X_{2}}\U\q$. Then it follows that at least $2^{n}$ distinct 
negative example are required to characterise $\q$ within $\mathcal{Q}_{p}^\Sigma[\U](\mathcal{ELI})$. To prove our claim assume that  $\mathcal{D},a,0\models \q_{X_{1}}\U\q$ and
$\mathcal{D},a,0\models \q_{X_{2}}\U\q$ but that $\mathcal{D},a,0\not\models \nxt\q$.
Then $\mathcal{D},a,1\not\models \q$ but $\mathcal{D},a,1\models \q_{X_{1}}$ and $\mathcal{D},a,1\models \q_{X_{2}}$. But by the proof above no such data instance exists.

\section{Proofs for Section~\ref{Sec:learning}}
We provide further details of the proof of Theorem~\ref{thm:learning}.

\medskip
\noindent
{\bf Theorem~\ref{thm:learning}.}
{\em
	$(i)$ The class of safe queries in $\Qp[\nxt, \Diamondw](\mathcal{ELI})$
	is polynomial-time learnable with membership queries.
	
	$(ii)$ The class $\Qp[\nxt, \Diamondw](\mathcal{ELI})$ is polynomial-time learnable with membership queries if the learner knows the size of the target query in advance.
	
	$(iii)$ The class $\Qp[\nxt,\Diamond](\mathcal{ELI})$ is polynomially-time learnable with membership queries.
}

\medskip

We start by completing the proof for the 1D case. For $(i)$, it remains to consider
\textbf{step 4}. At that point of the computation, the algorithm has identified all
blocks of $\q$ but not the sequences of $\Diamond$ and $\Diamondw$ between
them.  Suppose that $\D = \qw_0\emptyset^b\dots\qw_i\emptyset^b\qw_{i+1}\dots\emptyset^b\qw_n$.
We construct $\q'$ from the blocks of $\D$ by selecting $\mathcal{R}_{i+1}$ as
follows: If $\D_{i} \models \q$ for $\D_{i} = \qw_{0} \emptyset^{b} \dots
\qw_{i} \! \Join \! \qw_{i+1} \dots \emptyset^{b} \qw_{n}$ then set
$\mathcal{R}_{i+1}$ to be $\leq$.  Otherwise, let $n_{i+1}$ be the smallest
number such that $\D_i\models \q$ for $\D_{i} = \qw_{0} \emptyset^{b} \dots
\qw_{i} \emptyset^{n_{i+1}} \qw_{i+1} \dots \emptyset^{b} \qw_{n}$. Set
$\mathcal{R}_{i+1}$ to be a sequence of $<$ of length $n_{i+1}$. It is now easy to see that
$\q'$ fits the example set $(E^{+},E^{-})$ and so $\q' \equiv \q$, as required.

\medskip

For $(ii)$ we have to replace \textbf{step 3} by a computation step that ensures that
after applying (f) one does not obtain a data instance $\D$ with $\D\models \q$. Recall
that for bound $n=|\q|$ and $\q_{i}= \rho(s)$ a lone conjunct in $\q$ with
$\rho=\{A_{1},\ldots,A_{k}\}$ the rule (f) is defined as follows:
\begin{description}
	\item[\rm (f)] replace $\rho$ with $(\rho\setminus \{A_{1}\}\emptyset^{b}\cdots \emptyset^{b} \rho\setminus \{A_{k}\})^{n}$.
\end{description}is satisfied.
Now, after having computed $\D$ in \textbf{step 2}, and a data instance $\D'$ is obtained from $\D$
by applying rule $(f)$ such that $\D'\models \q$, then replace $\D$ with $\D'$ and return to \textbf{step 2}. If no such $\D'$ exists, proceed to \textbf{step 4}. To bound the number of applications of rule $(f)$ notice that at the end of
\textbf{step 2}, the number of time points in $\D$ other than $\emptyset$ does not
exceed $|\q|$. Indeed, any time point in $\D$ not in the range of some
$h:\q\to\D$ would be eliminated by rule $(a)$. We obtain a polynomial bound on the number
of applications by observing that each application of rule $(f)$ removes a symbol from $\rho$.

\medskip

We now move to the 2D case. The proof extends the argument given above for the 1D case and uses the example set $E=(E^{+},E^{-})$ defined in the proof of Theorem~\ref{thm:nextdiamond2}. Intuitively, whenever in the argument above we replace $\rho$ by $\rho\setminus \{A\}$ for a set $\rho$ of atoms, we now replace $\hat{\r}$ by an element of the frontier $\mathcal{F}(\hat{\r})$. There is only one difficulty: in the 1D case the algorithm starts with data instances $\sigma\cdots \sigma$ with $\sigma$ the signature of the target query, whereas now we have to start with data instances  $\mathcal{A}_{\Sigma}\cdots\mathcal{A}_{\Sigma}$ with $\Sigma$ the signature of $\q$ and $\mathcal{A}_{\Sigma}=\{R(a,a),A(a)\mid R,A\in \Sigma\}$. The atemporal data instance $\mathcal{A}_{\Sigma}$, however, is not tree-shaped, and we have not yet discussed frontiers for data instances that are not tree-shaped.
Indeed, in \cite{DBLP:conf/icdt/CateD21}, frontiers are not only computed for tree-shaped data instances but for a generalisation called \emph{c-acyclic} data instances with cycles through the distinguished node. We could at this point introduce the relevant machinery from \cite{DBLP:conf/icdt/CateD21} and work with frontiers for c-acyclic data instances.
Instead, we show that one can use the machinery we have introduced already and work with frontiers of tree-shaped data instances. But we require a straightforward intermediate step that transfers
the data instance $\mathcal{A}_{\Sigma}$ into a tree-shaped data instance $\mathcal{D}$ with $\mathcal{D}\models \q$ using membership queries.

\newcommand*{\rw}{\ensuremath{\hat{r}}}

\medskip

We adjust \textbf{step 1} of the learning algorithm from the 1D case as follows.
Assume $\q$ is the target query and let $\Sigma=\sig(\q)$. We aim to identify an initial temporal data instance $\D_0 = \A_0,\dots,\A_n$ with a
designated individual $a$ such that
\begin{itemize}
	\item $\D_0, a, 0\models\q$,
	\item if $\D'$ is obtained from $\D_0$ by removing an atom then
	$\D', a, 0\not\models\q$, and
	\item all $\A_i$ are tree-shaped with distinguished node $a$.
\end{itemize}
By asking incrementally membership queries of the form `$\A_\Sigma^k,a,0\models \q$?', we
identify the number of $\nxt$ and $\Diamond$ in $\q$.
Let $b = \min\{k\mid(\A_\Sigma)^k,a,0\models \q\}+1$.

Let $\D_0 = \A_0,\dots, \A_n$, where $n=b-2$ and, initially, $A_i = \A_\Sigma$
for $i=1,\dots,n$. Before progressing, we make $\A_i$ tree-shaped by applying
the following unwind and minimise operations:
\begin{description}
	\item[unwind] Suppose that $\A_i$ contains an atom $S(c,c)$.
	Then introduce fresh individuals $c_{R},c_{R^{-}}$ for every binary predicate $R$ with $R(c,c)\in \A_{i}$, remove all $R(c,c)$ from $\A_{i}$, and add instead $R(c,c_{R})$, $R(c_{R^{-}},c)$, and $A(c_{R})$, $A(c_{R^{-}})$
	$R'(c_{R},c_{R})$, and $R'(c_{R^{-}},c_{R^{-}})$ for all $A,R'\in \Sigma$.
	Let $\A'_{i}$ denote the resulting data instance and set
    $\D' = \A_0,\dots,\A_{i-1}, \A'_i, \A_{i+1},\dots, \A_n$.
	\item[minimise]
	Remove exhaustively atoms from $\A_i'$ as long as $\D',a,0\models\q$.
\end{description}
Observe that if $\D\models \q$ and $\D'$ is obtained from $\D$ by an unwinding step, then $\D'\models \q$. After the minimise step, the size of $\D'$ does not exceed the size of $\q$. Therefore we can replace $\D_0$ with $\D'$. By applying the unwind-minimise steps exhaustively, we eventually eliminate all loops from $\A_i$. It remains to notice that the minimise step can be implemented by querying the membership oracle.

The remaining steps of the learning algorithm remain the same as before by replacing the removal of an atom by taking an element of the frontier of a tree-shaped data instance.
For $(iii)$, notice that in part $(ii)$ the learner actually only needs to know
the temporal depth of the goal query, not the overall size. Hence $(iii)$ also
reduces to $(ii)$.

\fi

\end{document}